\documentclass[11pt]{article}

\usepackage{amssymb}
\usepackage[all]{xy}\NoCompileMatrices
\usepackage{fullpage}
\usepackage{euler}

\usepackage{hyperref}

\renewcommand{\paragraph}[1]{\medskip\noindent\textbf{#1}\ }
\newcommand{\mysubparagraph}[1]{\noindent\textit{#1}\ }

\newcommand{\Def}[1]{{\em #1\/}}
\newcommand{\Rem}[1]{{\em #1\/}}
\newcommand{\eg}{{\it e.g.}}
\newcommand{\ie}{{\it i.e.}}
\newcommand{\cf}{{\it cf.}}
\newcommand{\viz}{{\it viz.}}

\newcommand{\qed}{\hfill$\Box$}

\newcommand{\myiff}{\mathrel{\Leftarrow\!\!\!\!\!\!\Rightarrow}}

\newcommand{\Set}{\mathbf{Set}}

\newcommand{\TypeSet}{\mathcal{T}}
\newcommand{\BaseTypeSet}{\mathrm{T}}
\newcommand{\abasetype}{\theta}
\newcommand{\atype}{\tau}
\newcommand{\antype}{\sigma}
\newcommand{\TypeClosure}[1]{\widetilde{#1}}
\newcommand{\semfun}{\mbox{\boldmath$s$}}
\newcommand{\trivsemfun}{\mbox{\boldmath$t$}}
\newcommand{\freesemfun}{\mbox{\boldmath$f$}}

\newcommand{\gsemfun}{\ext\semfun}

\newcommand{\yon}{\mathbf{y}}
\newcommand{\commayon}{\ext\yon}
\newcommand{\ctxtemb}[1]{\langle #1 \rangle}
\newcommand{\ext}[1]{\overline{#1}}
\newcommand{\clone}[1]{\langle #1 \rangle}

\newcommand{\nt}[1]{\underline{#1}}
\newcommand{\VarPSh}[1]{\mathrm{V}_{#1}}

\newcommand{\pair}[1]{\langle #1 \rangle}
\newcommand{\setof}[1]{\{\, #1 \,\}}
\newcommand{\suchthat}{\mid}
\newcommand{\dom}{\mathrm{dom}}
\newcommand{\im}{\mathrm{im}}
\newcommand{\union}{\cup}

\newcommand{\Lan}{\mathrm{Lan}}
\newcommand{\fstarg}{\xy(0,-1)\ar@{-}(2,-1)\endxy}
\newcommand{\simplefstarg}{\_\!\_}
\newcommand{\sndarg}{\xy(0,-1)\ar@{=}(2,-1)\endxy}
\newcommand{\ev}{\varepsilon}
\newcommand{\proj}{\pi}

\newcommand{\projop}{\pi}
\newcommand{\unitop}{\pair{ }}
\newcommand{\pairop}{\pair}
\newcommand{\lambdaop}[2]{\lambda\, #1.\, #2}
\newcommand{\semlambdaop}[2]
{\lambda\:\!\!\!\!\lambdaop{#1}{#2}}

\newcommand{\F}[1]{\mathbb{F}_{#1}}
\newcommand{\FreeCCC}{\mathcal{F}_{\mathrm{ccc}}}
\newcommand{\cat}[1]{\mathcal{#1}}
\newcommand{\Semcat}{\cat S}
\newcommand{\catc}{\scat C}
\newcommand{\scat}[1]{\mathbb{#1}}
\newcommand{\op}{\mathrm{op}}
\newcommand{\iso}{\cong}
\newcommand{\id}{\mathrm{id}}

\newcommand{\comp}{\circ}

\newcommand{\obj}[1]{\mbox{$\mid\!#1\!\mid$}}
\newcommand{\scriptobj}[1]{\mid\!#1\!\mid}
\newcommand{\comma}{\!\downarrow\!}
\newcommand{\terminalobj}{1}
\newcommand{\Hom}[1]{\mbox{\Large$[$}#1\mbox{\Large$]$}}
\newcommand{\arHom}[1]{[#1]}
\newcommand{\puncture}{-^\shortmid_\shortmid-}

\newcommand{\SSSTerm}{\Sigma}
\newcommand{\SSSNE}{\Sigma_1}
\newcommand{\SSSNF}{\Sigma_2}
\newcommand{\EEE}[1]{\mathrm{E}_{#1}}
\newcommand{\VVVar}{{\mathscr{X}}}
\newcommand{\VVVarr}{{\mathscr{Y}}}

\newcommand{\arrowtype}{\mathtt{=}\!\!\mathtt{>}}
\newcommand{\unittype}{\mathtt{1}}
\newcommand{\producttype}{*}
\newcommand{\TermSet}{{\mathscr{L}}}
\newcommand{\TermPSh}[1]{\TermSet_{#1}}
\newcommand{\TerminalPSh}{1}
\newcommand{\Termsem}[1]{\ell_{#1}}
\newcommand{\extTermsem}[1]{\ext\ell_{#1}}
\newcommand{\NETerm}{{\mathscr{M}}}
\newcommand{\NEPSh}[1]{{\NETerm_{#1}}}
\newcommand{\SubNEPSh}[1]{{\mathscr{P}_{#1}}}
\newcommand{\SubNEmap}[1]{\imath_{#1}}
\newcommand{\NEsem}[1]{m_{#1}}
\newcommand{\extNEsem}[1]{\ext{m}_{#1}}
\newcommand{\NEobj}[1]{\mu_{#1}}
\newcommand{\NFTerm}{{\mathscr{N}}}
\newcommand{\NFPSh}[1]{{\NFTerm_{#1}}}
\newcommand{\SubNFPSh}[1]{{\mathscr{Q}_{#1}}}
\newcommand{\SubNFmap}[1]{\jmath_{#1}}
\newcommand{\NFsem}[1]{n_{#1}}
\newcommand{\extNFsem}[1]{\ext{n}_{#1}}
\newcommand{\NFobj}[1]{\eta_{#1}}
\newcommand{\SemPSh}[1]{{\mathscr{C}_{#1}}}
\newcommand{\NSSemPSh}[1]{{\mathscr{S}_{#1}}}
\newcommand{\NSsem}[1]{\sigma_{#1}}
\newcommand{\Varobj}[1]{\nu_{#1}}
\newcommand{\norm}{\mathsf{norm}}
\newcommand{\nf}[2]{#2\mbox{-}\mathsf{nf}_{#1}}
\newcommand{\simplenf}[1]{\mathsf{nf}_{#1}}
\newcommand{\Gluecat}{\mbox{$\Set^{\F{}\comma\TypeClosure\BaseTypeSet}\comma\clone{\semfun\sembracket\fstarg}$}}

\newcommand{\Sproj}{\pi}
\newcommand{\Pproj}{\varpi}

\newcommand{\unquotemap}[1]{\mathrm{u}_{#1}}
\newcommand{\quotemap}[1]{\mathrm{q}_{#1}}
\newcommand{\vmap}{\mathrm{v}}

\newcommand{\varmap}{\mathsf{var}}
\newcommand{\appmap}{\mathsf{app}}
\newcommand{\fstmap}{\mathsf{fst}}
\newcommand{\sndmap}{\mathsf{snd}}
\newcommand{\absmap}{\mathsf{abs}}
\newcommand{\unitmap}{\mathsf{unit}}
\newcommand{\pairmap}{\mathsf{pair}}

\newcommand{\subst}[2]{^{#1}\!/_{#2}}

\newcommand{\entails}{\vdash}
\newcommand{\NEentails}{\vdash_\mathscr{M}}
\newcommand{\NFentails}{\vdash_\mathscr{N}}
\newcommand{\llbracket}{[\![}
\newcommand{\rrbracket}{]\!]}
\newcommand{\sembracket}[1]{\llbracket #1 \rrbracket}

\newcommand{\TerminalLogRel}{\top}
\newcommand{\ProductLogRel}{\wedge}
\newcommand{\ArrowLogRel}{\supset}
\newcommand{\semLogRel}{\mathscr{R}}
\newcommand{\defLogRel}{\mathscr{D}}
\newcommand{\lowerLogRel}{\mathscr{A}}
\newcommand{\upperLogRel}{\mathscr{B}}
\newcommand{\mLogRel}{\mathscr{DM}}
\newcommand{\nLogRel}{\mathscr{DN}}
\newcommand{\K}[1]{\mbox{\underline{$\mathrm{K}$}}\langle #1\rangle}
\newcommand{\Int}{\mathcal{I}}

\newcommand{\TerminalObj}{1}
\newcommand{\InitialObj}{0}
\newcommand{\CatProduct}{\times}
\newcommand{\CatCoproduct}{+}
\newcommand{\CatArrow}{\shortRightarrow}

\newcommand{\TerminalType}{\unittype}
\newcommand{\UnitType}{\unittype}
\newcommand{\ProductType}{\producttype}
\newcommand{\ArrowType}{\arrowtype}
\newcommand{\eval}{\varepsilon}

\newcommand{\arity}{\mbox{\boldmath$\varsigma$}}

\newcommand{\betaetaeq}{=_{\beta\eta}}

\newtheorem{theorem}{Theorem}
\newtheorem{corollary}[theorem]{Corollary}
\newtheorem{lemma}[theorem]{Lemma}
\newtheorem{proposition}[theorem]{Proposition}
\newtheorem{definition}[theorem]{Definition}
\newtheorem{example}[theorem]{Example}

\newenvironment{proof}
{\begin{trivlist} \item[] {\sc Proof:}}{\qed \end{trivlist}}
\newenvironment{proof*}
{\begin{trivlist} \item[] {\sc Proof:}}{\end{trivlist}}
\newenvironment{remark}
{\begin{trivlist} \item[] {\it Remark.}}{\end{trivlist}}

\newcommand{\diagram}{\xymatrix} 

\newcommand{\arrowlength}{5} 
\newcommand{\longarrowlength}{10} 
\newcommand{\longlongarrowlength}{15} 
\newcommand{\medarrowlength}{7} 

\newcommand{\myarrow}[3] 
   {\mathrel{\xy\ar(#1,0)^-{#2}_-{#3}\endxy}}
\renewcommand{\rightarrow}
   {\myarrow{\arrowlength}{}{}}

\newcommand{\stackrightarrow}[1]
   {\myarrow{\arrowlength}{#1}{}}

\newcommand{\longstackrightarrow}[1]
   {\myarrow{\longarrowlength}{#1}{}}
\newcommand{\stackleftarrow}[1]
  {\myarrow{-\arrowlength}{}{#1}}

\newcommand{\medstackrightarrow}[1]
   {\myarrow{\medarrowlength}{#1}{}}
\newcommand{\longlongstackrightarrow}[1]
   {\myarrow{\longlongarrowlength}{#1}{}}

\newcommand{\myArrow}[3] 
   {\mathrel{\xy\ar@{=>}(#1,0)^-{#2}_-{#3}\endxy}}
\renewcommand{\Rightarrow}
   {\myArrow{\arrowlength}{}{}}
\renewcommand{\Leftarrow}
   {\myArrow{-\arrowlength}{}{}}

\newcommand{\mymono}[2] 
   {\mathrel{\,\;\xy\ar@{>->}(#1,0)^-{#2}\endxy}}
\newcommand{\rightmono}
   {\mymono{\arrowlength}{}}

\newcommand{\stackrightmono}[1]
   {\mymono{\arrowlength}{#1}}
\newcommand{\longstackrightmono}[1]
   {\mymono{\longarrowlength}{#1}}
\newcommand{\longlongstackrightmono}[1]
   {\mymono{\longlongarrowlength}{#1}}

\newcommand{\doublestackrightmono}[2] 
   {\mathrel{\,\;\xy\ar@{>->}(\arrowlength,0)^-{#1}_-{#2}\endxy}}

\newcommand{\myepi}[2] 
   {\mathrel{\,\;\xy\ar@{->>}(#1,0)^-{#2}\endxy}}

\newcommand{\myembedding}[2] 
   {\mathrel{\:\xy\ar@{^(->}(#1,0)^-{#2}\endxy}}
\newcommand{\rightembedding}
   {\myembedding{\arrowlength}{}}

\newcommand{\longstackrightembedding}[1]
   {\myembedding{\longarrowlength}{#1}}

\newcommand{\verticalarrow}[3] 
{\xy (0,4)*+{\mbox{\scriptsize$#1$}}
     \ar(0,-4)*+{\mbox{\scriptsize$#3$}} _-{\mbox{$#2$}} 
 \endxy}

\newcommand{\mymapsto}[2] 
  {\mathrel{\xy\ar@{|->}(#1,0)^-{#2}\endxy}}
\renewcommand{\mapsto}
  {\mymapsto{\arrowlength}{}}

\newcommand{\mydasharrow}[3] 
   {\mathrel{\xy\ar@{.>}(#1,0)^-{#2}_-{#3}\endxy}}

\newcommand{\adjunction}[2] 
{\diagram{\ar[r]<-1ex>_-{#1} \ar@0[r]|-\top& \ar[l]<-1ex>_-{#2}}}

\newcommand{\reflection}[2] 
{\diagram{\ar[r]<-1ex>_-{#1} \ar@0[r]|-\top& \ar@{_(->}[l]<-1ex>_-{#2}}}

\newcommand{\coreflection}[2] 
{\diagram{\ar@{^(->}[r]<-1ex>_-{#1} \ar@0[r]|-\top& \ar[l]<-1ex>_-{#2}}}

\newcommand{\fib}[3] 
{\xy (0,4)*+{\mbox{\scriptsize$#1$}} 
     \ar(0,-4)*+{\mbox{\scriptsize$#3$}} _-{\mbox{$#2$}} 
 \endxy}
\newcommand{\fibmapsto}[3] 
{\xy (0,4)*+{\mbox{\scriptsize$#1$}} 
 \ar@{|.>}(0,-4)*+{\mbox{\scriptsize$#3$}} _-{\mbox{$#2$}} 
 \endxy}

\newcommand{\leftparallelarrow} 
  {\diagram@C=15pt{& \ar[l]<-1ex> \ar[l]<1ex>}}
\newcommand{\stackleftparallelarrow}[2]              
  {\diagram{& \ar[l]<-1ex>_-{#1} \ar[l]<1ex>^-{#2}}} 
\newcommand{\stackleftrightarrow}[2]                 
  {\diagram{\ar[r]<0.5ex>^-{#1} & \ar[l]<0.5ex>^-{#2}}} 

\newcommand{\shortarrowlength}{4} 

\newcommand{\shortRightarrow}{\myArrow{\shortarrowlength}{}{}}

\usepackage{autofe}
\usepackage[utf8]{inputenc}
\usepackage{newunicodechar}
\usepackage{agda}

\newunicodechar{⊔}{\ensuremath{\sqcup}}
\newunicodechar{Σ}{\ensuremath{\sum}}
\newunicodechar{Π}{\ensuremath{\Pi}}
\newunicodechar{⊤}{\ensuremath{\top}}

\newunicodechar{𝔽}{\ensuremath{\mathbb{F}}}
\newunicodechar{↑}{\ensuremath{\uparrow}}
\newunicodechar{↓}{\ensuremath{\downarrow\!\!}}

\newunicodechar{≪}{\ensuremath{\langle}}
\newunicodechar{≫}{\ensuremath{\rangle}}

\newunicodechar{⟦}{\ensuremath{\llbracket}}
\newunicodechar{⟧}{\ensuremath{\rrbracket}}

\newunicodechar{∘}{\ensuremath{\circ}}
\newunicodechar{′}{\ensuremath{'}}

\newunicodechar{₀}{}
\newunicodechar{₁}{\ensuremath{_1}}
\newunicodechar{₂}{\ensuremath{_2}}

\newunicodechar{ₘ}{\ensuremath{_{\mathsf m}}}
\newunicodechar{ₙ}{\ensuremath{_{\mathsf n}}}
\newunicodechar{ᵥ}{\ensuremath{_{\mathsf v}}}
\newunicodechar{ₚ}{\ensuremath{_{\Pi}}}

\newunicodechar{⇒}{\ensuremath{\ArrowType}}
\newunicodechar{⋆}{\ensuremath{\ProductType}}

\newunicodechar{Δ}{\ensuremath{\Delta}}
\newunicodechar{Γ}{\ensuremath{\Gamma}}
\newunicodechar{Ξ}{\ensuremath{\Xi}}

\newunicodechar{θ}{\ensuremath{\abasetype}}
\newunicodechar{α}{\ensuremath{\atype}}
\newunicodechar{β}{\ensuremath{\antype}}
\newunicodechar{υ}{\ensuremath{\UnitType}}

\newunicodechar{ε}{\ensuremath{\varepsilon}}

\newunicodechar{λ}{\ensuremath{\lambda}}
\newunicodechar{ρ}{\ensuremath{\rho}}

\newunicodechar{π}{\ensuremath{\pi}}

\newunicodechar{𝓣}{\ensuremath{\TypeClosure\BaseTypeSet}}
\newunicodechar{𝓛}{\ensuremath{\TermSet}}
\newunicodechar{𝓜}{\ensuremath{\NETerm}}
\newunicodechar{𝓝}{\ensuremath{\NFTerm}}
\newunicodechar{𝓧}{\ensuremath{\TypeSet}}

\newunicodechar{≡}{\equiv}

\begin{document}


\title{Semantic Analysis of Normalisation by Evaluation\\ 
       for Typed Lambda Calculus\footnote{This is 
         a slight revision, with an implementation, of the full
         version, with proofs, of February~2003 for the extended
         abstract~\cite{FioreNBE} published in October~2002.}}


\author{Marcelo Fiore\thanks{This research was supported by an EPSRC Advanced
                             Research Fellowship~(2000--2005) and partially
                             supported by EPSRC grant
                             EP/V002309/1~(2021--2024).}\\ 
  Department of Computer Science and Technology\\
  University of Cambridge}

\date{August 2022}

\maketitle

\begin{abstract} 
This paper studies normalisation by evaluation for typed lambda calculus from
a categorical and algebraic viewpoint.  
The first part of the paper analyses the lambda definability result of Jung
and Tiuryn via Kripke logical relations and shows how it can be adapted to
unify definability and normalisation, yielding an extensional normalisation
result.  
In the second part of the paper the analysis is refined further by considering
intensional Kripke relations (in the form of Artin-Wraith glueing) and shown
to provide a function for normalising terms, casting normalisation by
evaluation in the context of categorical glueing. 
The technical development includes an algebraic treatment of the syntax and
semantics of the typed lambda calculus that allows the definition of the
normalisation function to be given within a simply typed metatheory.
A normalisation-by-evaluation program in a dependently-typed functional
programming language is synthesised.
\end{abstract}

\section*{\Large Introduction}

Normalisation by evaluation for typed lambda calculus was first considered by
Berger and Schwichtenberg~\cite{BS} from a type and proof theoretic viewpoint,
and later investigated from the point of view of logic~\cite{Berger}, type
theory~\cite{CCoquand}, category theory~\cite{AHS,CDS,Reynolds}, and partial
evaluation~\cite{Danvy,Filinski99}.  This work gives a new categorical and
algebraic perspective on the topic.  

\paragraph{Outline.}
Normalisation by evaluation will be broadly viewed as the technique of giving
semantics in (metalanguages for) non-standard models from which normalisation
information can be extracted~(\cf~\cite{MartinLof}). 
In this light, we will investigate the following problems.
\begin{enumerate}
\item[I.]\Rem{Extensional normalisation problem}: To define normal terms and
  establish that every term \mbox{$\beta\eta$-equals} one in normal form.

  That is, writing $\TermPSh\tau(\Gamma)$ for the set of terms of type~$\tau$
  in context $\Gamma$, to identify a set of normal terms 
  $\NFPSh\tau(\Gamma) \subseteq \TermPSh\tau(\Gamma)$ and show that for every 
  term $t\in\TermPSh\tau(\Gamma)$ there exists a normal term 
  $N\in\NFPSh\tau(\Gamma)$ such that $t \betaetaeq N$.

\item[II.]\Rem{Intensional normalisation problem}: To define, and prove the
  correctness of, a normalisation function associating normal forms to terms.

  More precisely, to construct functions
  $$
  \simplenf{\tau,\Gamma}:
     \TermPSh\tau(\Gamma) \rightarrow \NFPSh\tau(\Gamma)
  $$
  satisfying the following three properties.
  \begin{enumerate}
  \item[(1)]For all normal terms $N\in\NFPSh\tau(\Gamma)$, the syntactic
    equality $\simplenf{\tau,\Gamma}(N)$ $= N$ holds. 

  \item[(2)]For all terms $t\in\TermPSh\tau(\Gamma)$, the semantic equality
    $\simplenf{\tau,\Gamma}(t) \betaetaeq t$ holds. 

  \item[(3)]For all pair of 
    terms $t,t'\in\TermPSh\tau(\Gamma)$, if $t\betaetaeq t'$ then 
    $\simplenf{\tau,\Gamma}(t) = \simplenf{\tau,\Gamma}(t')$.
  \end{enumerate}
  
  In the context of normalisation by evaluation, the correctness
  condition~(1) has seldom been considered ---the exception
  being~\cite{Reynolds}.  However, it is both natural and interesting.  
  For instance, together with the correctness condition~(3) it implies that
  $\beta\eta$-equal normal terms are syntactically equal, which in turn,
  together with the correctness condition~(2), entails the stronger version
  of extensional normalisation that every term \mbox{$\beta\eta$-equals} a
  unique normal term.  
\end{enumerate}

These problems will be respectively dealt with in Parts~I and~II of the paper.  
Part~I provides a unifying view of definability and normalisation leading to
an extensional normalisation result.  This analysis, besides unifying the two
hitherto unrelated problems of definability and normalisation, motivates and
elucidates the notions of neutral and normal terms, which are here distilled
from semantic considerations.  
Part~II shows that an intensional view of Part~I amounts to the traditional
technique of normalisation by evaluation.  This development leads to a
treatment of normalisation by evaluation via the Artin-Wraith glueing
construction, finally formalising the observation that normalisation by
evaluation is \Rem{closely related} to categorical
glueing~\cite{CoquandDybjer}. 

More in detail, the paper is organised as follows.
Section~I.1 briefly recalls the syntax and categorical semantics of the
typed lambda calculus.  
Section~I.2 presents an analysis of the lambda definability result of Jung
and Tiuryn via Kripke logical relations~\cite{JungTiuryn} leading to an
extensional normalisation result.
Section~II.1 describes the rudiments of a theory of typed abstract syntax
with variable binding which is used to put the typed lambda calculus in an
algebraic framework.
This algebraic view is exploited in Section~II.2 to structure the development
of an intensional version of Section~I.2 culminating in the technique of
normalisation by evaluation.

\paragraph{Related work.}
The treatment of extensional normalisation presented here is similar to Tait's
approach to strong normalisation via computability
predicates~\cite{Tait,Girard} for the typed lambda calculus, and also to
Krivine's approach to normalisation~\cite[Chapter~III]{Krivine} for the
untyped lambda calculus.  The precise relationships need to be investigated.  

The analysis of normalisation by evaluation pursued here is categorical and,
as such, is related to~\cite{AHS,CDS,Reynolds,ADHS}.   

The approach of \u{C}ubri\'c, Dybjer, and Scott~\cite{CDS} is in the context
of so-called $\mathcal{P}$-category theory; which is, roughly, a version of
category theory equipped with an intensional notion of equality formalised by
partial equivalence relations.  The intensional information needed for the
purpose of normalisation will be captured here in the context of traditional
category theory via Artin-Wraith glueing.

In Altenkirch, Hofmann, and Streicher~\cite{AHS}, normalisation by evaluation
is reconstructed categorically in a model obtained via an ad~hoc 
\Rem{twisted-glueing} construction.  This model embodies objects with both
syntactic and semantic components, and translations between them essentially
encoding a correctness predicate.  In contrast, we adopt a purely semantic
view, working with intensional logical relations in models given by the
traditional categorical-glueing construction~\cite{Wraith}.

Another important point of departure between this work and the other
categorical ones is the algebraic treatment of the subject, which led to a
deeper understanding of the normalisation function.  

\section*{\Large Part~I}

\subsection*{I.1\quad Typed lambda calculus}

For the purpose of establishing notation, we briefly recall the syntax and
semantics of the typed lambda calculus.  For details
see,~\eg,~\cite{LambekScott,Crole,Taylor}.

\paragraph{Syntax.} 
The types of the simply typed lambda calculus are given by the grammar
\begin{equation}\label{SimpleTypes}
\begin{array}{rcl}
\tau 
& ::= & \abasetype \mid \UnitType \mid \tau_1 \ProductType \tau_2
        \mid \tau_1 \ArrowType \tau_2
\end{array}
\end{equation}
where $\abasetype$ ranges over base types.  We write
$\TypeClosure\BaseTypeSet$ for the set of simple types generated by a set of
base types $\BaseTypeSet$.  

The grammar for the terms is 
$$\begin{array}{rcl}
t 
& ::= & x \mid \unitop \mid \projop_1(t) \mid \projop_2(t) \mid
        \pairop{t_1,t_2} \mid t(t') \mid \lambdaop{x:\tau}t
\end{array}$$ 
where $x$ ranges over (a countably infinite set of) variables.  The notion of
free and bound variables are standard.  As usual, we will identify terms up
to the renaming of bound variables.  

Typing contexts, with types in a set $\TypeSet$, are defined as functions $V
\rightarrow \TypeSet$ where the domain of the context, $V$, is a finite subset
of the set of variables.  Under this view, for a variable~$x$, a type~$\tau$,
and a context~$\Gamma$, we let \mbox{$(x:\tau) \in \Gamma$} stand for
\mbox{$x\in\dom(\Gamma)$} and $\Gamma(x) = \tau$.  For distinct variables
\mbox{$x_i$~($i=1,n$),} we use the notation $\pair{x_i:\tau_i}_{i=1,n}$ for
the context  \mbox{$\setof{ x_1,\ldots, x_n } \rightarrow \TypeSet$} mapping
$x_i$ to~$\tau_i$.  For a context~$\Gamma$, a variable $x$, and a type $\tau$,
the notation $\Gamma, x:\tau$ presupposes $x \not\in \dom(\Gamma)$ and denotes
the context $\dom(\Gamma)\union\setof{x} \rightarrow \TypeSet$ mapping every
$y \in \dom(\Gamma)$ to $\Gamma(y)$, and $x$ to $\tau$.

The well-typed terms $\Gamma \entails t: \tau$ in context (where $\Gamma$ is
a typing context, $t$ is a term, and $\tau$ is a type) are given by the usual
rules; see Figure~\ref{Terms}.
\begin{figure}
\hrulefill
$$\begin{array}{c}
\begin{array}{c}
\\ \hline
\Gamma \vdash x:\tau
\end{array}
\quad(x:\tau) \in \Gamma
\\[1mm]%
\begin{array}{c}
\\ \hline
\Gamma \entails \unitop: \UnitType
\end{array}
\\[5mm]%
\begin{array}{c}
\Gamma \entails t: \tau_1\ProductType\tau_2
\\ \hline
\Gamma \entails \projop_i(t): \tau_i
\end{array}
\quad(i=1,2)
\\[6mm]%
\begin{array}{c}
\Gamma \entails t_i: \tau_i \quad(i=1,2)
\\ \hline
\Gamma \entails \pairop{t_1,t_2}: \tau_1\ProductType\tau_2
\end{array}
\\[6mm]%
\begin{array}{c}
\Gamma \entails t: \tau'\ArrowType\tau
\quad
\Gamma \entails t': \tau'
\\ \hline
\Gamma \entails t(t'): \tau
\end{array}
\\[6mm]%
\begin{array}{c}
\Gamma, x:\tau' \entails t: \tau
\\ \hline
\Gamma \entails \lambdaop{x:\tau'}t: \tau'\ArrowType\tau
\end{array}
\end{array}$$
\caption{Well-typed terms}\label{Terms}
\hrulefill
\end{figure}

\paragraph{Semantics.} 
The appropriate mathematical universes for giving semantics to the typed
lambda calculus are cartesian closed
categories~\cite{LambekScott,Crole,Taylor};~\ie,~categories with terminal
object, binary products, and exponentials (for which we respectively use
the notation~$\TerminalObj$,~$\CatProduct$, and~$\CatArrow$).  

For an interpretation $\semfun: \BaseTypeSet \rightarrow \Semcat$
of base types in a cartesian closed category, we let 
$\semfun\sembracket{\fstarg}: 
   \TypeClosure\BaseTypeSet \rightarrow \Semcat$ 
be the extension to simple types as prescribed by a chosen cartesian closed
structure.  That is,
$\semfun\sembracket\abasetype = \semfun(\abasetype)$ (for 
$\abasetype \in \BaseTypeSet$), 
$\semfun\sembracket\TerminalType = \TerminalObj$, and
$\semfun\sembracket{\tau\ProductType\tau'} 
   = \semfun\sembracket\tau \CatProduct \semfun\sembracket{\tau'}$
and
$\semfun\sembracket{\tau\ArrowType\tau'} 
   = \semfun\sembracket\tau \CatArrow \semfun\sembracket{\tau'}$
(for $\atype, \atype' \in \TypeClosure\BaseTypeSet$).  
As usual, the interpretation of types is extended to contexts by
setting  
$\semfun\sembracket\Gamma
   = \prod_{x\in\dom(\Gamma)} \semfun\sembracket{\Gamma(x)}$ 
for all contexts~$\Gamma$.  
Finally, the semantics of a term $\Gamma \entails t: \tau$ as a morphism  
$\semfun\sembracket\Gamma \rightarrow \semfun\sembracket\tau$ in
$\Semcat$ is denoted $\semfun\sembracket{\Gamma \entails t: \tau}$. 

\section*{I.2\quad From definability to normalisation}

Kripke relations were introduced by Jung and Tiuryn in~\cite{JungTiuryn} for
the purpose of characterising lambda definability.  We will analyse this
result and provide a corresponding extensional normalisation result. 

\paragraph{Kripke relations.} 
For a functor $\arity: \catc \rightarrow \Semcat$, a 
\Def{$\catc$-Kripke relation} $R$ of arity $\arity$ over an object $A$ of
$\Semcat$ is a family 
$\setof{ R(c) \subseteq \Semcat(\arity(c),A) }_{c \in \scriptobj{\,\catc\,}}$
satisfying the following condition.
\begin{description}
\item[](Monotonicity)
  For every $\rho: c' \rightarrow c$ in $\catc$ and every 
  \mbox{$a: \arity(c) \rightarrow A$} in $R(c)$, the map 
  $a \comp \arity(\rho): \arity(c') \rightarrow A$ is in $R(c')$. 
\end{description}
In other words, a $\catc$-Kripke relation $R$ of arity~$\arity$ over an object
$A$ is a unary predicate ${R\rightembedding\Semcat(\arity(\fstarg),A)}$ over
the $\catc^\op$-variable set of $A$-valued morphisms
$\Semcat(\arity(\fstarg),A): \catc^\op \rightarrow \Set$ in the functor
category $\Set^{\catc^\op}$ of $\catc^\op$-variable sets, referred to as
presheaves.

The category of Kripke relations $\K{\arity}$ of arity 
$\arity: \catc \rightarrow \Semcat$ has objects given by pairs $(R,A)$
consisting of an object $A$ of $\Semcat$ and a 
$\catc$-Kripke relation of arity $\arity$ over $A$, and morphisms 
\mbox{$f: (R,A) \rightarrow (R',A')$} given by maps $f: A \rightarrow A'$ in
$\Semcat$ such that, for all \mbox{$a: \arity(c) \rightarrow A$} in $R(c)$,
the composite \mbox{$f \comp a: \arity(c) \rightarrow A'$} is in $R'(c)$.
Composition and identities are as in $\Semcat$.

\begin{example}
The category of $\catc$-Kripke relations of arity the unique functor to the
terminal category is (isomorphic to) the complete Heyting algebra of
subterminal objects of the presheaf topos $\Set^{\catc^\op}$.  
\end{example}

The following proposition is
well-known~(see,~\eg,~\cite{Alimohamed,MaReynolds}).   
\begin{proposition} \label{KripkeCCC}
Let $\catc$ be a small category and let $\Semcat$ be a cartesian closed
category.  For a functor \mbox{$\arity: \catc \rightarrow \Semcat$}, the
category of Kripke relations $\K{\arity}$ is cartesian closed and the
forgetful functor \mbox{$\K{\arity} \rightarrow \Semcat: (R,A) \mapsto A$}
preserves the cartesian closed structure strictly.
\end{proposition}
The cartesian closed structure of $\K{\arity}$ is given as follows. 
\begin{description}
\item[](Products)
  The terminal object is $(\TerminalLogRel,\TerminalObj)$ where $\TerminalObj$
  is terminal in $\Semcat$ and where 
  $\TerminalLogRel(c) = \setof{ \arity(c) \rightarrow \TerminalObj}$ for all
  $c$ in $\catc$.  

  The product $(R,A)\CatProduct(R',A')$ of $(R,A)$ and $(R',A')$ is 
  $$
  (R,A) 
    \stackleftarrow\proj 
  (R \ProductLogRel R' , A \CatProduct A')
    \stackrightarrow{\proj'} 
  (R',A')
  $$ 
  where 
  $A \stackleftarrow\proj A \CatProduct A' \stackrightarrow{\proj'} A'$
  is the product of $A$ and $A'$ in $\Semcat$, and where 
  $a: \arity(c) \rightarrow A \CatProduct A'$ is in 
  $(R \ProductLogRel R')(c)$ iff $\proj \comp a: \arity(c) \rightarrow A$ is
  in $R(c)$ and $\proj' \comp a: \arity(c) \rightarrow A'$ is in $R'(c)$.

\item[](Exponentials)
  The exponential $(R,A)\CatArrow(R',A')$
  of $(R,A)$ and $(R',A')$ is 
  $$
  (R \ArrowLogRel R' , A \CatArrow A') \CatProduct (R,A)
  \stackrightarrow\eval
  (R',A')
  $$ 
  where $(A \CatArrow A') \CatProduct A \stackrightarrow\eval A'$ is the
  exponential of $A$ and $A'$ in $\Semcat$, and where 
  ${f: \arity(c) \rightarrow A \CatArrow A'}$ is in $(R \ArrowLogRel R')(c)$
  iff for every $\rho: c' \rightarrow c$ in $\catc$ and 
  $a: \arity(c') \rightarrow A$ in $R(c')$, the composite 
  $\eval \comp \pair{f \comp \arity(\rho) , a}: \arity(c') \rightarrow A'$ is
  in $R'(c')$.
\end{description}

The Fundamental Lemma of logical relations is a consequence of
Proposition~\ref{KripkeCCC}. 

\begin{lemma}\Rem{(Fundamental Lemma~\cite{Plotkin,Statman})}
For an interpretation of base types
$\Int: \BaseTypeSet \rightarrow \K{\arity}:
       \abasetype \mapsto (\semLogRel_\abasetype,\Int_0(\abasetype))$,
the interpretation 
\begin{center}
$\Int_0\sembracket{\Gamma \entails t: \tau}:
   \Int_0\sembracket{\Gamma} \rightarrow \Int_0\sembracket{\tau}$
in $\Semcat$
\end{center}
of a term $\Gamma \entails t: \tau$ yields a morphism 
\mbox{$\Int\sembracket{\Gamma} \rightarrow \Int\sembracket{\tau}$} in
$\K{\arity}$; that is, for 
$\Int\sembracket\Gamma = (\semLogRel_\Gamma,\Int_0\sembracket\Gamma)$ 
and $\Int\sembracket\tau = (\semLogRel_\tau,\Int_0\sembracket\tau)$,  
the following diagram
$$\xymatrix@C=5pt{
\raisebox{1.5mm}{$\semLogRel_\Gamma$}
\ar[rrrrrr] 
\ar@{^(->}[d] &&&&&& 
\raisebox{1mm}{$\semLogRel_\tau$} \ar@{^(->}[d]  
\\
\Semcat(\arity(\fstarg),\Int_0\sembracket\Gamma) 
\ar[rrrrrr]_-{\Int_0\sembracket{\Gamma \entails t: \tau}\comp\fstarg}
&&&&&&
\Semcat(\arity(\fstarg),\Int_0\sembracket\tau) 
}$$
commutes in $\Set^{\catc^\op}\!\!$ (for a necessarily unique natural map 
\mbox{$\semLogRel_\Gamma \rightarrow \semLogRel_\tau$}).
\end{lemma}

\paragraph{Definability.}
The definability result of Jung and Tiuryn~\cite{JungTiuryn} uses Kripke
relations varying over a poset of contexts ordered by context extension.
Here, however, to parallel the development with the one to follow in Part~II,
we will consider Kripke relations varying over a category of contexts and
context renamings. 

\begin{definition} \label{ContextCat}
For a set of types $\TypeSet$, we let $\F{}\comma\TypeSet$ be the category
with objects given by contexts $\Gamma$ with types in~$\TypeSet$, and
with morphisms $\Gamma \rightarrow \Gamma'$ given by type-preserving
context renamings; that is, by functions 
$\rho: \dom(\Gamma) \rightarrow \dom(\Gamma')$ such that for all
variables $x \in \dom(\Gamma)$, the types $\Gamma(x)$ and 
$\Gamma'(\rho x)$ are equal.  We write $\F{}[\TypeSet]$ for 
$(\F{}\comma\TypeSet)^\op$.
\end{definition}

With respect to an interpretation 
\mbox{$\semfun: \BaseTypeSet \rightarrow \Semcat$} 
of base types in a cartesian closed category, we write
$\semfun\sembracket{\fstarg}$ for 
the canonical semantic functor 
$\F{}[\TypeClosure\BaseTypeSet] \rightarrow \Semcat$ interpreting
contexts and their renamings.  This is explicitly given by 
\[
\semfun\sembracket\rho 
\ = \
\pair{ \semfun\sembracket{\Gamma' \entails \rho x: \tau}}_{(x:\tau)\in\Gamma}
\ = \
\pair{\proj_{\rho x} }_{x\in\dom(\Gamma)}
\ : \ \semfun\sembracket{\Gamma'} \rightarrow \semfun\sembracket\Gamma
\]
for all $\rho: \Gamma \rightarrow \Gamma'$ in 
$\F{}\comma\TypeClosure\BaseTypeSet$. 

For every type $\tau \in \TypeClosure\BaseTypeSet$, the
\Def{definability relation} 
$$
\defLogRel_\tau(\Gamma) 
= 
\setof{ \
  \semfun\sembracket{\Gamma \entails t: \tau }
    \suchthat \Gamma \entails t: \tau
  \ }
\subseteq
  \Semcat(\semfun\sembracket\Gamma,\semfun\sembracket\tau)
$$
is an $\F{}[\TypeClosure\BaseTypeSet]$-Kripke relation of arity 
\mbox{$\semfun\sembracket\fstarg: 
   \F{}[\TypeClosure\BaseTypeSet] \rightarrow \Semcat$}
over $\semfun\sembracket\tau$, and the family of definability relations 
$\setof{ \defLogRel_\tau }_{\tau \in \TypeClosure\BaseTypeSet}$ has the
following logical characterisation.
\begin{lemma}\Rem{(Definability Lemma~\cite{JungTiuryn,Alimohamed})} 
\label{Definability_Lemma}
Let $\semfun: \BaseTypeSet \rightarrow \Semcat$ be an interpretation of base
types in a cartesian closed category.
Setting $\semLogRel_\abasetype = \defLogRel_\abasetype$ for all base types
$\abasetype \in \BaseTypeSet$ and letting $\semLogRel_\tau$ be given by the
cartesian closed structure of the category of Kripke relations
$\K{\semfun\sembracket\fstarg: 
      \F{}[\TypeClosure\BaseTypeSet] \rightarrow \Semcat}$
for the other types~$\tau \in \TypeClosure\BaseTypeSet$, it follows that
$\semLogRel_\tau = \defLogRel_\tau$ for all types 
$\tau \in \TypeClosure\BaseTypeSet$.
\end{lemma}
The usual proof of the Definability Lemma 
is by induction on the structure of types using the explicit description of
the cartesian closed structure in categories of Kripke relations given above;
see~\cite{JungTiuryn,Alimohamed} (and~\cite{FioreSimpson} for the case with
sum types).  However, there is a more conceptual argument based on
establishing that the definability relations satisfy the following closure
properties:
$$\begin{array}{rcl}
\defLogRel_{\TerminalType}
& = & 
\TerminalLogRel
\\[2mm]
\defLogRel_{\mbox{\scriptsize$\tau \ProductType \tau'$}} 
& = & 
\defLogRel_\tau \ProductLogRel \defLogRel_{\tau'} 
\\[2mm]
\defLogRel_{\tau \ArrowType \tau'} 
& = & 
\defLogRel_\tau \ArrowLogRel \defLogRel_{\tau'} 
\end{array}$$
which is, in effect, what the usual calculations really amount to.

The above analysis can be refined further.  Indeed, the fact
that neither of the following inclusions 
\begin{equation} \label{D_subseteq_R_subseteq_D}
\defLogRel_\tau \subseteq \semLogRel_\tau \subseteq \defLogRel_\tau
\end{equation}
in isolation is strong enough to re-establish the inductive hypothesis in the
Definability Lemma, suggests considering a more general situation in which the
Kripke logical relations $\semLogRel_\tau$ are bounded by possibly distinct
Kripke relations (unlike the situation in~(\ref{D_subseteq_R_subseteq_D})).  

We are thus led to the following Basic Lemma.  Notice the mixed-variance
treatment of exponentiation.  This is akin to Krivine's approach to
normalisation for the untyped lambda calculus using \Rem{adapted pairs} of
subsets of lambda terms~\cite[Chapter~III, pages~33--39]{Krivine}.
\begin{lemma}\Rem{(Basic Lemma)} \label{Basic_Lemma}
Consider an interpretation 
\mbox{$\Int_0: \BaseTypeSet \rightarrow \Semcat$}
of base types in a cartesian closed category $\Semcat$. 

With respect to a functor $\arity: \catc \rightarrow \Semcat$, let 
$\pair{ (\lowerLogRel_\tau,\Int_0\sembracket\tau) }
   _{\tau\in\TypeClosure\BaseTypeSet}$ 
and 
$\pair{ ( \upperLogRel_\tau,\Int_0\sembracket\tau ) }
   _{\tau\in\TypeClosure\BaseTypeSet}$ 
be two families of Kripke relations in $\K{\arity}$ indexed by types
such that 
$$\begin{array}{c}
\upperLogRel_\TerminalType \, = \, \TerminalLogRel 
\\[2mm]
\lowerLogRel_{\mbox{\scriptsize$\sigma\ProductType\tau$}}
\, \subseteq \, 
(\lowerLogRel_{\sigma} \ProductLogRel \lowerLogRel_{\tau})
\qquad\qquad 
(\upperLogRel_{\sigma} \ProductLogRel \upperLogRel_{\tau})
\, \subseteq \, 
\upperLogRel_{\mbox{\scriptsize$\sigma\ProductType\tau$}}
\\[2mm]
\lowerLogRel_{\sigma\ArrowType\tau}
\, \subseteq \, 
(\upperLogRel_{\sigma} \ArrowLogRel \lowerLogRel_{\tau})
\qquad\qquad 
(\lowerLogRel_{\sigma} \ArrowLogRel \upperLogRel_{\tau})
\, \subseteq \,
\upperLogRel_{\sigma\ArrowType\tau}
\end{array}$$
For a family of Kripke relations
$\pair{ ( \semLogRel_\abasetype , \Int_0\sembracket\abasetype ) }
   _{\abasetype\in\BaseTypeSet}$ 
in $\K{\arity}$ indexed by base types, let 
$\pair{ ( \semLogRel_\tau , \Int_0\sembracket\tau ) }_
   {\tau\in\TypeClosure\BaseTypeSet}$ 
be the family of Kripke relations indexed by types induced by the
cartesian closed structure of $\K{\arity}$.

If
$\lowerLogRel_\abasetype
   \subseteq \semLogRel_\abasetype
     \subseteq \upperLogRel_\abasetype$
for all base~types~$\abasetype \in \BaseTypeSet$, then
\begin{enumerate}
\item \label{Basic_Lemma_One}
$\lowerLogRel_\tau
   \subseteq \semLogRel_\tau
   \subseteq \upperLogRel_{\tau}$ 
for all types~$\tau \in \TypeClosure\BaseTypeSet$, and thus

\item \label{Basic_Lemma_Two}
for all terms 
$\Gamma\entails t:\tau$~$($with~$\Gamma=\pair{x_i:\tau_i}_{i=1,n})$
and morphisms $a_i: \arity(c) \rightarrow \Int_0\sembracket{\tau_i}$ in 
$\lowerLogRel_{\tau_i}(c)$~$(1 \leq i \leq n, c \in \obj{\catc})$,
we have that
$\Int_0\sembracket{\Gamma \entails t: \tau} \comp \pair{a_1,\ldots,a_n}:
   \arity(c) \rightarrow \Int_0\sembracket{\tau}$
is in $\upperLogRel_\tau(c)$.
\end{enumerate}
\end{lemma}
%
\begin{proof}
The proof of the first part is by induction on the structure of types.  
This uses the facts that
\begin{center}
$R\subseteq\TerminalLogRel$ for all $(R,\TerminalObj)$ in $\K{\arity}$ 
\end{center}
and that, for Kripke relations $(R_i,A_i)$ and $(R'_i,A_i)$ in
$\K{\arity}$ ($i=1,2$), 
\begin{center}
if $R_1 \subseteq R'_1$ and $R_2 \subseteq R'_2$ then 
$(R_1\ProductLogRel R_2) \subseteq (R'_1\ProductLogRel R'_2)$ and 
$(R'_1\ArrowLogRel R_2) \subseteq (R_1\ArrowLogRel R'_2)$
\end{center}
which follows from the functoriality of binary products and exponentials using
the observation that, for $(R,A)$ and $(R',A)$ in
$\K{\arity}$, 
\[ 
R \subseteq R'
\myiff 
\id_A: (R,A) \rightarrow (R',A)
\mbox{ in $\K{\arity}$}
\quad.
\] 

The proof of the second part follows from considering
the interpretation $\Int: \BaseTypeSet \rightarrow \K{\arity}$ mapping a base
type $\abasetype$ to the Kripke relation
$( \semLogRel_\abasetype , \Int_0\sembracket\abasetype )$ and noticing that, 
by the first part and the Fundamental Lemma of logical relations, the diagram
below in $\Set^{\catc^\op}$
\begin{equation} \label{Basic_Lemma_Diagram}
\begin{array}{c}\xymatrix@C=0pt{
\lowerLogRel_\Gamma 
\ar@{^(->}[r] 
\ar@{^(->}[dr] & \semLogRel_\Gamma \ar@{^(->}[d] 
\ar@{->}[rrrrrrrrrr] 
&&&&&&&&&& \semLogRel_\tau \ar@{^(->}[d] 
\ar@{^(->}[r] 
& \upperLogRel_\tau \ar@{^(->}[dl] 
\\
& 
\Semcat(\arity(\fstarg),\Int_0\sembracket\Gamma)
\ar[rrrrrrrrrr]_-{\Int_0\sembracket{\Gamma \entails t: \tau} \comp \fstarg}
&&&&&&&&&& 
\Semcat(\arity(\fstarg),\Int_0\sembracket\tau)
}\end{array}
\end{equation}
commutes, where for $\Gamma = \pair{ x_i:\tau_i }_{i=1,n}$, 
$\lowerLogRel_\Gamma 
 = 
 \lowerLogRel_{\tau_1}\ProductLogRel\ldots\ProductLogRel\lowerLogRel_{\tau_n}$
and 
$\semLogRel_\Gamma 
 = 
 \semLogRel_{\tau_1}\ProductLogRel\ldots\ProductLogRel\semLogRel_{\tau_n}$.
\end{proof}

The Basic Lemma yields the Definability Lemma by considering 
$\lowerLogRel_\tau = \defLogRel_\tau = \upperLogRel_\tau$ in 
the category of Kripke relations
$\K{\semfun\sembracket\fstarg: 
   \F{}[\TypeClosure\BaseTypeSet] \rightarrow \Semcat}$ 
for the given interpretation 
$\semfun: \BaseTypeSet \rightarrow \Semcat$.  We will now see that the Basic
Lemma can be also applied to obtain an extensional normalisation result (see
Lemma~\ref{Extensional_Normalisation_Lemma}).     

\paragraph{Normalisation.}
For an interpretation \mbox{$\semfun: \BaseTypeSet \rightarrow \Semcat$}
of base types in a cartesian closed category we aim at defining families  
$\setof{ ( \mLogRel_\tau , \semfun\sembracket\tau ) 
   }_{\tau\in\TypeClosure\BaseTypeSet}$
and 
$\setof{ ( \nLogRel_\tau , \semfun\sembracket\tau ) 
         }_{\tau\in\TypeClosure\BaseTypeSet}$
of $\F{}[\TypeClosure\BaseTypeSet]$-Kripke relations of arity
$\semfun\sembracket\fstarg: 
   \F{}[\TypeClosure\BaseTypeSet] \rightarrow \Semcat$
of definable morphisms such that
$$\begin{array}{c}
(i)\enspace
   \nLogRel_\TerminalType \, = \, \TerminalLogRel 
\\[2mm]
(ii)\enspace 
    \mLogRel_{\mbox{\scriptsize$\sigma\ProductType\tau$}}
    \,\subseteq\,
    (\mLogRel_{\sigma} \ProductLogRel \mLogRel_{\tau})
\qquad\qquad 
(iii)\enspace
     (\nLogRel_{\sigma} \ProductLogRel \nLogRel_{\tau})
     \,\subseteq\,
     \nLogRel_{\mbox{\scriptsize$\sigma\ProductType\tau$}}
\\[2mm]
(iv)\enspace
    \mLogRel_{\sigma\ArrowType\tau}
    \,\subseteq\,
    ( \nLogRel_{\sigma} \ArrowLogRel \mLogRel_{\tau} )
\qquad\qquad 
(v)\enspace
   ( \mLogRel_{\sigma} \ArrowLogRel \nLogRel_{\tau} ) 
   \,\subseteq\,
   \nLogRel_{\sigma\ArrowType\tau}
\\[2mm]
(vi)\enspace
    \mLogRel_{\abasetype} \subseteq \nLogRel_{\abasetype} 
    \quad (\abasetype \in \BaseTypeSet)
\\[2mm]
(vii)\enspace
     \proj_x: 
     \semfun\sembracket\Gamma \rightarrow \semfun\sembracket{\tau}
     \in \mLogRel_{\tau}(\Gamma)
\quad ((x:\tau)\in\Gamma)
\end{array}$$
so that, by the second part of the Basic Lemma, we get (setting
$\semLogRel_\abasetype = \mLogRel_\abasetype$ for all
$\abasetype \in \BaseTypeSet$, and 
$a_i = \proj_i: 
 \semfun\sembracket\Gamma \rightarrow \semfun\sembracket{\tau_i}$ 
for $\Gamma = \pair{x_i:\tau_i}_{i=1,n}$) that, for all terms 
$\Gamma \entails t: \tau$, 
\begin{equation} \label{Extensional_Normalisation}
\mbox{
$\semfun\sembracket{\Gamma \entails t: \tau}: 
  \semfun\sembracket\Gamma \rightarrow \semfun\sembracket\tau$
is in $\nLogRel_\tau(\Gamma)$\quad.
}
\end{equation}

The above will be achieved by distilling the semantic closure
properties~$(i)$--$(vii)$ into two syntactic typing systems $\NEentails$ and
$\NFentails$ with respect to which the definitions 
\begin{equation} \label{SemNeutralTerms}
\mLogRel_\tau(\Gamma) 
\ = \
\setof{ \ 
        \semfun\sembracket{\Gamma \entails M:\tau} 
        \suchthat
        \Gamma \NEentails M: \tau
        \ }
\end{equation}
\begin{equation} \label{SemNormalTerms}
\nLogRel_\tau(\Gamma) 
\ = \
\setof{ \
        \semfun\sembracket{\Gamma \entails N:\tau} 
        \suchthat
        \Gamma \NFentails N: \tau
        \ }
\end{equation}
will provide the required Kripke relations (see
Proposition~\ref{NEandNF_KripkeRelations}).  The conditions $(i)$--$(vii)$
amount, roughly, to the following properties. 
\begin{itemize}
\item
  The system~$\NEentails$ should contain variables
  (condition $(vii)$), and be closed under
  projections~(condition~$(ii)$) and under the application to terms in
  the system~$\NFentails$ (condition~$(iv)$).

\item
  The system~$\NFentails$ should contain the unit~(condition~$(i)$),
  and should be closed under pairing~(condition~$(iii)$) and under
  abstraction~(condition~$(v)$).

\item
  Every term of base type in the system~$\NEentails$ should be in the
  system~$\NFentails$~(condition~$(vi)$).
\end{itemize}
Formally, the systems are given by the rules in Figure~\ref{NN_Terms}.

\medskip
Thus, from purely semantic considerations, we have synthesised the notions of
neutral normal forms~(\viz, those derivable in the system~$\NEentails$) and of
long $\beta\eta$-normal forms~(\viz,~those derivable in the
system~$\NFentails$), henceforth respectively referred to as \Rem{neutral} and
\Rem{normal} terms, and characterised as follows.
\begin{figure}
\hrulefill
$$\begin{array}{c}
\begin{array}{c}
\\ \hline 
\Gamma \NEentails x: \tau
\end{array}
\quad(x:\tau) \in \Gamma
\\[6mm]%
\begin{array}{c}
\Gamma \NEentails M: \tau_1 \ProductType \tau_2
\\ \hline
\Gamma \NEentails \projop_i(M): \tau_i
\end{array}
\quad(i=1,2)
\\[6mm]%
\begin{array}{c}
\Gamma \NEentails M: \tau \ArrowType \tau'
\quad
\Gamma \NFentails N: \tau
\\ \hline
\Gamma \NEentails M(N): \tau'
\end{array}
\\[4mm]%
\begin{array}{c}
\\ \hline
\Gamma \NFentails \unitop: \UnitType
\end{array}
\\[6mm]%
\begin{array}{c}
\Gamma \NFentails N_i: \tau_i \quad (i=1,2)
\\ \hline
\Gamma \NFentails \pairop{N_1,N_2}: \tau_1 \ProductType \tau_2
\end{array}
\\[6mm]%
\begin{array}{c}
\Gamma, x: \tau \NFentails N: \tau'
\\ \hline
\Gamma \NFentails \lambdaop{x:\tau}N: \tau \ArrowType \tau'
\end{array}
\\[6mm]%
\begin{array}{c}
\Gamma \NEentails M: \abasetype
\\ \hline
\Gamma \NFentails M: \abasetype
\end{array}
\quad\mbox{($\abasetype$ a base type)}
\end{array}$$
\caption{Neutral and normal terms}\label{NN_Terms}
\hrulefill
\end{figure}

\begin{proposition}\label{Neutral_and_Normal_Characterisation}
\mbox{}
\begin{enumerate}
\item (Neutral terms)
\[
\Gamma \NEentails t: \tau
\myiff
\begin{array}[t]{l}
[ \ \exists\ (x:\tau)\in\Gamma.\ t = x \ ]
\\[2mm]
\enspace \vee \
[ \ \exists\ \Gamma\NEentails M:\tau\ProductType\tau'.\  t = \projop_1(M) \ ]
\enspace \vee \
[ \ \exists\ \Gamma\NEentails M:\tau'\ProductType\tau.\  t = \projop_2(M) \ ]
\\[2mm]
\enspace \vee \
[ \ \exists\ \Gamma\NEentails M:\tau'\ArrowType\tau\, , 
            \Gamma\NFentails N:\tau' . \  
  t = M(N) \ ]
\end{array}
\]

\item (Normal terms) \label{Normal_Characterisation}
\begin{itemize}
\item  
  $\Gamma\NFentails t:\UnitType \myiff t = \unitop$

\item  
  $\Gamma\NFentails t:\tau_1\ProductType\tau_2 
   \myiff 
   [ \ \exists\ \Gamma\NFentails N_1:\tau_1\, , \Gamma \NFentails N_2:\tau_2 .\
     t = \pairop{N_1,N_2}
     \ ]$

\item  
  $\Gamma\NFentails t:\tau\ArrowType\tau' 
   \myiff 
   [ \ \exists\ \Gamma,x:\tau\NFentails N:\tau'.\ t = \lambdaop{x:\tau}N \ ]$

\item  For $\abasetype$ a base type, 
  $\Gamma\NFentails t:\theta \myiff \Gamma\NEentails t:\theta$.
\end{itemize}
\end{enumerate}
\end{proposition}

Neutral and normal terms are closed under context renamings and thereby
semantically induce Kripke relations.

\begin{proposition}\label{NEandNF_KripkeRelations}
Let $\semfun: \BaseTypeSet \rightarrow \Semcat$ be an interpretation of base
types in a cartesian closed category.
For all types $\tau\in\TypeClosure\BaseTypeSet$, the
definitions~(\ref{SemNeutralTerms}) and~(\ref{SemNormalTerms}) respectively
yield  $\F{}[\TypeClosure\BaseTypeSet]$-Kripke relations $\mLogRel_\tau$ and
$\nLogRel_\tau$ of arity 
$\semfun\sembracket\fstarg: 
   \F{}[\TypeClosure\BaseTypeSet] \rightarrow \Semcat$
satisfying conditions~$(i)$--$(vii)$.
\end{proposition}
\begin{proof}
The first part is a corollary of the facts that 
\begin{equation} \label{Neutral_Renaming_Invariance}
\Gamma\NEentails M:\tau 
\myiff 
\forall\ \rho:\Gamma\to\Gamma' \mbox{ in } \F{}\comma\TypeClosure\BaseTypeSet .\
\Gamma'\NEentails M[\subst{\rho x}x]_{x\in\dom(\Gamma)} : \tau 
\end{equation}
and 
\begin{equation} \label{Normal_Renaming_Invariance}
\Gamma\NFentails N:\tau 
\myiff 
\forall\ \rho:\Gamma\to\Gamma' \mbox{ in } \F{}\comma\TypeClosure\BaseTypeSet .\
\Gamma'\NFentails N[\subst{\rho x}x]_{x\in\dom(\Gamma)} : \tau 
\quad.
\end{equation}

The second part follows by the construction of the systems $\NEentails$ and
$\NFentails$.  
\end{proof}

From Proposition~\ref{NEandNF_KripkeRelations}, we
have~(\ref{Extensional_Normalisation}) and therefore, 
from~(\ref{SemNormalTerms}) and
Proposition~\ref{Neutral_and_Normal_Characterisation}%
(\ref{Normal_Characterisation}), 
we obtain the following Extensional Normalisation Lemma.  
\begin{lemma}~\Rem{(Extensional Normalisation Lemma)} 
\label{Extensional_Normalisation_Lemma} 
Let $\semfun:\BaseTypeSet\rightarrow\Semcat$ be an interpretation of base
types in a cartesian closed category.  For every term 
$\Gamma \entails t: \tau$ there exists a long $\beta\eta$-normal term 
$\Gamma \NFentails N: \tau$ such that  
\[
\semfun\sembracket{\Gamma \entails t: \tau}
\ = \
\semfun\sembracket{\Gamma \entails N: \tau}
\ : \
\semfun\sembracket\Gamma \rightarrow \semfun\sembracket\tau
\]
in $\Semcat$.
\end{lemma}

Specialising the Extensional Normalisation Lemma for the canonical
interpretation of types in the free cartesian closed category generated
by them we obtain the following syntactic result~(\cf~\cite{StreicherNBE}).
\begin{corollary} \label{Corollary_BetaEtaLongNormalForms}
Every simply typed term is $\beta\eta$-equal to one in long $\beta\eta$-normal
form.
\end{corollary}

The above does not give information about the long~$\beta\eta$-normal form
associated to a term because Kripke relations are extensional predicates.
What is needed instead for this purpose is a notion of \Rem{intensional}
Kripke relation in which the extension of the predicate is witnessed (or
realised).  Technically, this amounts to revisiting the development in
categories obtained by the Artin-Wraith glueing construction~\cite{Wraith}.
This will be done in Part~II.  To do it at an appropriate abstract,
syntax-independent level we will first consider the  
typed lambda calculus algebraically.

\section*{\Large Part~II}

\section*{II.1\quad Algebraic typed lambda calculus}

We provide an algebraic setting for the syntax and semantics of the typed
lambda calculus following and extending the theory of~\cite{FiorePlotkinTuri}.
In particular, we describe the typed abstract syntax of simply typed and of
neutral and normal terms as initial algebras, and show how the usual semantics
corresponds to unique algebra homomorphisms from the initial (term) algebras
to suitable semantic algebras.   

\subsection*{II.1.1\quad Syntax}

Categories of contexts, which we study next, play a crucial role in
describing abstract syntax with variable binding; see~\cite{FiorePlotkinTuri}
for further details.

\paragraph{Free (co)cartesian categories.}
The category of untyped contexts and renamings $\F{}$ with objects given by
finite subsets of (the countably infinite set of) variables and morphisms
given by all functions is the free cocartesian category on one generator.  

More generally, the free cocartesian category over a set $\TypeSet$ can be
described as the comma category $\F{}\comma\TypeSet$ of contexts with types
in the set $\TypeSet$ and type-preserving context renamings.  
(That is, $\F{}\comma\TypeSet$ is the category with objects given by maps 
$\Gamma: V \rightarrow \TypeSet$ where $V$ is in $\F{}$, and with morphisms 
$\rho: \Gamma \rightarrow \Gamma'$ given by functions 
\mbox{$\rho: \dom(\Gamma) \rightarrow \dom(\Gamma')$} such that 
$\Gamma = \Gamma' \comp \rho$.) 
The initial object $(\InitialObj \rightarrow \TypeSet)$ in
$\F{}\comma\TypeSet$ is the empty context; whilst the coproduct in
$\F{}\comma\TypeSet$ is
$$\begin{array}{rcl}
(V \stackrightarrow\Gamma \TypeSet)
\CatCoproduct
(V' \stackrightarrow{\Gamma'} \TypeSet)
& = &
(V \CatCoproduct V' \longstackrightarrow{[\Gamma,\Gamma']} \TypeSet)
\end{array}$$

As before, we write $\F{}[\TypeSet]$ for $(\F{}\comma\TypeSet)^\op$.
Further, we write $\ctxtemb\fstarg: \TypeSet \rightarrow \F{}[\TypeSet]$ for
the universal embedding (mapping $\tau$ to 
$(1 \stackrightarrow\tau \TypeSet)$) and exhibiting $\F{}[\TypeSet]$ as the
free cartesian category over $\TypeSet$. 

\paragraph{Typed abstract syntax with variable binding.}
The semantic universe on which to consider the algebras for the typed
lambda calculus over a set of base types $\BaseTypeSet$ is the functor 
category $\Set^{\F{}\comma\TypeClosure\BaseTypeSet}$ of
\mbox{$\F{}\comma\TypeClosure\BaseTypeSet$-variable sets}, referred to as 
(covariant)~presheaves.  (Recall that
$\Set^{\F{}\comma\TypeClosure\BaseTypeSet}$ has objects given by
functors \mbox{$\F{}\comma\TypeClosure\BaseTypeSet \rightarrow \Set$} and
morphisms $\varphi: P \rightarrow P'$ 
given by natural transformations; that is,~families of functions
\mbox{$\varphi 
 = 
 \setof{ \varphi_\Gamma: P(\Gamma) \rightarrow P'(\Gamma)
         }_{\Gamma\in\scriptobj{\,\F{}\comma\TypeClosure\BaseTypeSet\,}}$}
such that 
$\varphi_{\Gamma'} \comp P(\rho) = P'(\rho) \comp \varphi_\Gamma$ for
all $\rho:\Gamma\rightarrow\Gamma'$ in~$\F{}\comma\TypeClosure\BaseTypeSet$.)

The structure of $\Set^{\F{}\comma\TypeClosure\BaseTypeSet}$ allowing the
interpretation of variables and binding operators is described below.
\begin{itemize}
\item
The presheaf of variables of type $\tau \in \TypeClosure\BaseTypeSet$ is 
\begin{equation} \label{VPSh}
\VarPSh\tau = \yon\ctxtemb\tau
\end{equation}
in $\Set^{\F{}\comma\TypeClosure\BaseTypeSet}$ where
$$\begin{array}{rcl}
\F{}[\TypeClosure\BaseTypeSet] & \xymatrix{\ar@{^(->}[r]^\yon&} &
\Set^{\F{}\comma\TypeClosure\BaseTypeSet} \\[2mm]
\Gamma & \xymatrix{\ar@{|->}[r]&} &
(\F{}\comma\TypeClosure\BaseTypeSet)(\Gamma,\fstarg)
\end{array}$$
is the Yoneda embedding.

Hence, 
$\VarPSh\tau(\Gamma) \iso \setof{ x \suchthat (x:\tau) \in \Gamma }$.

\item
For every type $\tau \in \TypeClosure\BaseTypeSet$, the parameterisation
functor 
$\fstarg\CatProduct\ctxtemb\tau: 
   \F{}[\TypeClosure\BaseTypeSet] \rightarrow \F{}[\TypeClosure\BaseTypeSet]$
induces the following situation
\begin{equation} \label{Adjoint_Situation_For_F[T]}
\begin{array}{c}\xymatrix@C=5pt{\ar@{}[rrrrrd]|(.5){\stackrel{\Lan}{\iso}}
\F{}[\TypeClosure\BaseTypeSet]\ \ar@{^(->}[rrrrrr]^-\yon
\ar[d]_-{\fstarg\CatProduct\ctxtemb\tau} &&&&&& 
\Set^{\F{}\comma\TypeClosure\BaseTypeSet}
\ar[d]<-1ex>_-{\fstarg\CatProduct\yon\ctxtemb\tau} \ar@{}[d]|-{\dashv}
\\
\F{}[\TypeClosure\BaseTypeSet]\ \ar@{^(->}[rrrrrr]_-\yon &&&&&& 
\Set^{\F{}\comma\TypeClosure\BaseTypeSet}
\ar[u]<-1ex>_-{\Set^{(\fstarg\CatCoproduct\ctxtemb\tau)}}
}\end{array}
\end{equation}
Thus, in $\Set^{\F{}\comma\TypeClosure\BaseTypeSet}$, the exponential
$P^{\VarPSh\tau}$ of the presheaf $\VarPSh\tau$ and a presheaf $P$ can be
explicitly described as \mbox{$P(\fstarg\CatCoproduct\ctxtemb\tau)$}. 

Hence, $P^{\VarPSh\tau}(\Gamma) \iso P(\Gamma\CatCoproduct\ctxtemb\tau)$.
\end{itemize}

A \Def{typed lambda algebra} over a set of base types
$\BaseTypeSet$ is a $\TypeClosure\BaseTypeSet$-sorted algebra with
carrier given by a family 
$\setof{ \VVVar_\tau }_{\tau\in\TypeClosure\BaseTypeSet}$ of
presheaves in $\Set^{\F{}\comma\TypeClosure\BaseTypeSet}$ equipped
with the following operations:
\begin{center}\begin{tabular}{lcl}
(Variables) && 
  $\VarPSh\tau \rightarrow \VVVar_\tau$
\\[2mm]
(Unit) &&
  $\TerminalPSh \rightarrow \VVVar_\UnitType$
\\[2mm]
(First Projection) &&
  $\VVVar_{\tau\ProductType\tau'} \rightarrow \VVVar_\tau$
\\[2mm]
(Second Projection) &&
  $\VVVar_{\tau'\ProductType\tau} \rightarrow \VVVar_\tau$
\\[2mm]
(Pairing) &&
  $\VVVar_\tau \CatProduct \VVVar_{\tau'} 
     \rightarrow
   \VVVar_{\tau\ProductType\tau'}$
\\[2mm]
(Application) && 
  $\VVVar_{\tau'\ArrowType\tau} \CatProduct \VVVar_{\tau'} 
     \rightarrow 
   \VVVar_\tau$
\\[2mm]
(Abstraction) &&
  $(\VVVar_{\tau'})^{\VarPSh\tau} 
     \rightarrow
   \VVVar_{\tau\ArrowType\tau'}$
\end{tabular}\end{center}
Informally, one thinks of the sets 
$\VVVar_\tau(\Gamma)$~($\tau\in\TypeClosure\BaseTypeSet$,
$\Gamma \in \obj{\F{}\comma\TypeClosure\BaseTypeSet}$) as the
$\tau$-sorted elements of the algebra~$\VVVar$ in the context $\Gamma$.
Note that under this interpretation the abstraction operation corresponds to
a natural family of mappings  
$$
\VVVar_{\tau'}(\Gamma\CatCoproduct\ctxtemb\tau) 
\rightarrow
\VVVar_{\tau\ArrowType\tau'}(\Gamma)
$$
associating an element of sort $\tau'$ in the 
context \mbox{$\Gamma\CatCoproduct\ctxtemb\tau$}~(that is,~the
context~$\Gamma$ extended with a fresh variable of type $\tau$) with an
element of sort $\tau\ArrowType\tau'$ in the context $\Gamma$.

In the tradition of categorical algebra, the category of typed lambda algebras
can be defined as the category of \mbox{$\SSSTerm$-algebras} for a signature
endofunctor~$\SSSTerm$ on
$(\Set^{\F{}\comma\TypeClosure\BaseTypeSet})^{\TypeClosure\BaseTypeSet}$.
This endofunctor is induced by the above operations as follows, 
for $\abasetype\in\BaseTypeSet$ and $\tau,\tau'\in\TypeClosure\BaseTypeSet$:  
$$\begin{array}{lcll}
(\SSSTerm\VVVar)_\abasetype
& = & \VarPSh\abasetype \CatCoproduct \EEE\abasetype(\VVVar)
\\[2mm]
(\SSSTerm\VVVar)_\UnitType
& = & \VarPSh\UnitType \CatCoproduct \TerminalPSh \CatCoproduct 
      \EEE\UnitType(\VVVar)
\\[2mm]
(\SSSTerm\VVVar)_{\tau\ProductType\tau'}
& = & \VarPSh{\tau\ProductType\tau'} \CatCoproduct 
      (\VVVar_\tau \CatProduct \VVVar_{\tau'}) \CatCoproduct
      \EEE{\tau\ProductType\tau'}(\VVVar)
\\[2mm]
(\SSSTerm\VVVar)_{\tau\ArrowType\tau'}
& = & \VarPSh{\tau\ArrowType\tau'} \CatCoproduct 
      (\VVVar_{\tau'})^{\VarPSh\tau} \CatCoproduct
      \EEE{\tau\ArrowType\tau'}(\VVVar)
\end{array}$$
where 
$$\begin{array}{cl}
\EEE\tau(\VVVar) 
\; = \; \coprod_{\tau'\in\TypeClosure\BaseTypeSet}\:
        \VVVar_{\tau\ProductType\tau'} 
        \CatCoproduct
        \VVVar_{\tau'\ProductType\tau} 
        \CatCoproduct
        (\VVVar_{\tau'\ArrowType\tau}\CatProduct\VVVar_{\tau'})
\end{array}$$
is the signature endofunctor corresponding to the projections and application
operations onto $\tau$.

The initial $\SSSTerm$-algebra 
$\TermPSh{} = \setof{ \TermPSh\tau }_{\tau\in\TypeClosure\BaseTypeSet}$
with its structure
\begin{equation}\label{LamInitialAlg}
\begin{array}{rcl}
\VarPSh\abasetype \CatCoproduct \EEE\abasetype(\TermPSh{})
&
\stackrightarrow\iso
&
\TermPSh\abasetype
\\[2mm]
\VarPSh\UnitType \CatCoproduct \TerminalPSh \CatCoproduct 
\EEE\UnitType(\TermPSh{})
& \stackrightarrow\iso &
\TermPSh\UnitType
\\[2mm]
\VarPSh{\tau\ProductType\tau'} \CatCoproduct 
(\TermPSh\tau \CatProduct \TermPSh{\tau'}) \CatCoproduct
      \EEE{\tau\ProductType\tau'}(\TermPSh{})
& \stackrightarrow\iso &
\TermPSh{\tau\ProductType\tau'}
\\[2mm]
\VarPSh{\tau\ArrowType\tau'} \CatCoproduct 
(\TermPSh{\tau'})^{\VarPSh\tau} \CatCoproduct
\EEE{\tau\ArrowType\tau'}(\TermPSh{})
& \stackrightarrow\iso &
\TermPSh{\tau\ArrowType\tau'}
\end{array}
\end{equation}
can be explicitly described as the family of presheaves of terms 
$$
\TermPSh\tau(\Gamma) 
  = \setof{ \ t \suchthat \Gamma \entails t: \tau \ }
$$
with presheaf action given by variable renaming~(that is,~by the
mapping associating $\Gamma \entails t: \tau$ to 
\mbox{$\Gamma' \entails t[\subst{\rho x}x]_{x\in\dom(\Gamma)}: \tau$}
for any $\rho: \Gamma \rightarrow \Gamma'$ in
$\F{}\comma\TypeClosure\BaseTypeSet$), and with operations 
$$\begin{array}{lcl}
\varmap_\tau 
  & : & \VarPSh\tau \rightarrow \TermPSh\tau
\\[2mm]
\unitmap_\UnitType 
  & : & \TerminalPSh \rightarrow \TermPSh\UnitType
\\[2mm]
\fstmap^{(\tau')}_\tau 
  & : & \TermPSh{\tau\ProductType\tau'} \rightarrow \TermPSh{\tau}
\\[2mm]
\sndmap^{(\tau')}_\tau 
  & : & \TermPSh{\tau'\ProductType\tau} \rightarrow \TermPSh{\tau}
\\[2mm]
\pairmap_{\tau\ProductType\tau'} 
  & : & \TermPSh\tau \CatProduct \TermPSh{\tau'} 
        \rightarrow
        \TermPSh{\tau\ProductType\tau'}
\\[2mm]
\appmap^{(\tau')}_\tau 
  & : & \TermPSh{\tau'\ArrowType\tau} \CatProduct \TermPSh{\tau'}
        \rightarrow
        \TermPSh\tau
\\[2mm]
\absmap_{\tau\ArrowType\tau'} 
  & : & (\TermPSh{\tau'})^{\VarPSh\tau} 
        \rightarrow
        \TermPSh{\tau\ArrowType\tau'}
\end{array}$$
corresponding to the typing rules in Figure~\ref{Terms}.

\smallskip
A full theory of typed abstract syntax with variable binding incorporating
substitution along the lines of~\cite{FiorePlotkinTuri} can be developed
(see,~\eg,~\cite{FioreHur,FioreSzamozvancev}).  This is however not necessary
for the purposes of the paper.

\medskip
The notions of neutral and normal terms are given by mutual
induction~(see Figure~\ref{NN_Terms}) and, as such, the associated
algebraic notion corresponds to considering a signature endofunctor on
the product category 
$(\Set^{\F{}\comma\TypeClosure\BaseTypeSet})^{\TypeClosure\BaseTypeSet}
 \CatProduct
 (\Set^{\F{}\comma\TypeClosure\BaseTypeSet})^{\TypeClosure\BaseTypeSet}$.
This endofunctor, with components $\pair{\SSSNE,\SSSNF}$, is defined
below:
$$
\left\{\begin{array}{rcl}
(\SSSNE(\VVVar,\VVVarr))_\tau 
=  \VarPSh\tau \CatCoproduct \EEE\tau(\VVVar,\VVVarr)
\end{array}\right.
\qquad\qquad
\left\{\begin{array}{lcl}
(\SSSNF(\VVVar,\VVVarr))_\abasetype 
& = & \VarPSh\abasetype \CatCoproduct \EEE\abasetype(\VVVar,\VVVarr)
\\[2mm]
(\SSSNF(\VVVar,\VVVarr))_\UnitType 
& = & \TerminalPSh
\\[2mm]
(\SSSNF(\VVVar,\VVVarr))_{\tau\ProductType\tau'} 
& = & 
\VVVarr_\tau \CatProduct \VVVarr_{\tau'}
\\[2mm]
(\SSSNF(\VVVar,\VVVarr))_{\tau\ArrowType\tau'} 
& = & 
(\VVVarr_{\tau'})^{\VarPSh\tau} 
\end{array}\right.
$$
where
\[\textstyle
\EEE\tau(\VVVar,\VVVarr) 
 = \coprod_{\tau'\in\TypeClosure\BaseTypeSet}\:
     \VVVar_{\tau\ProductType\tau'}
     \CatCoproduct
     \VVVar_{\tau'\ProductType\tau}
     \CatCoproduct
     (\VVVar_{\tau'\ArrowType\tau}\CatProduct\VVVarr_{\tau'})
\]
for $\abasetype\in\BaseTypeSet$ and $\tau,\tau'\in\TypeClosure\BaseTypeSet$.

We write $(\NEPSh{},\NFPSh{})$ for the initial $\pair{\SSSNE,\SSSNF}$-algebra
with structure, for $\abasetype\in\BaseTypeSet$ and
$\tau,\tau'\in\TypeClosure\BaseTypeSet$, as follows:
\begin{equation} \label{NEandNFPShsOne}
\left\{\begin{array}{rcl}
\VarPSh\tau \CatCoproduct \EEE\tau(\NEPSh{},\NFPSh{})
& \stackrightarrow\iso &
\NEPSh\tau
\end{array}\right.
\qquad\qquad
\left\{\begin{array}{rcl}
\VarPSh\abasetype \CatCoproduct \EEE\abasetype(\NEPSh{},\NFPSh{})
& \stackrightarrow\iso &
\NFPSh\abasetype
\\[2mm]
\TerminalPSh
& \stackrightarrow\iso &
\NFPSh\UnitType
\\[2mm]
\NFPSh\tau \CatProduct \NFPSh{\tau'}
& \stackrightarrow\iso &
\NFPSh{\tau\ProductType\tau'}
\\[2mm]
(\NFPSh{\tau'})^{\VarPSh\tau} 
& \stackrightarrow\iso &
\NFPSh{\tau\ArrowType\tau'}
\end{array}\right.
\end{equation}
Note that we have an isomorphism
\begin{equation} \label{normISO}
\norm 
  \, : \,
  \NEPSh\abasetype 
  \iso 
  \VarPSh\abasetype \CatCoproduct \EEE\abasetype(\NEPSh{},\NFPSh{})
  \iso
  \NFPSh\abasetype
\end{equation}
for all $\abasetype\in\BaseTypeSet$. 

Explicitly, the presheaves $\NEPSh\tau$ and $\NFPSh\tau$ can be respectively
described as the neutral and normal terms
$$
\begin{array}{l}
\NEPSh\tau(\Gamma) 
\ = \ \setof{\, M\ \suchthat\ \Gamma \NEentails M: \tau\, }
\end{array}
\qquad\qquad
\begin{array}{r}
\NFPSh\tau(\Gamma) 
\ = \ \setof{\, N\ \suchthat\ \Gamma \NFentails N: \tau\, }
\end{array}
$$
with presheaf action given by variable 
renaming~(recall~(\ref{Neutral_Renaming_Invariance}) 
and~(\ref{Normal_Renaming_Invariance})), 
and with operations
\begin{equation} \label{NEandNFPShsTwo}
\ \left\{\begin{array}{lcl}
\varmap_\tau 
& \ : & 
\VarPSh\tau \rightarrow \NEPSh\tau
\\[2mm]
\fstmap^{(\tau')}_\tau 
& \ : & 
\NEPSh{\tau\ProductType\tau'} \rightarrow \NEPSh{\tau}
\\[2mm]
\sndmap^{(\tau')}_\tau 
& \ : & 
\NEPSh{\tau'\ProductType\tau} \rightarrow \NEPSh{\tau}
\\[2mm]
\appmap^{(\tau')}_\tau 
& \ : & 
\NEPSh{\tau'\ArrowType\tau} \CatProduct \NFPSh{\tau'}
  \rightarrow \NEPSh\tau
\end{array}\right.
\qquad\qquad
\left\{\begin{array}{lcl}
\varmap_\abasetype 
& \!\! : & 
\VarPSh\abasetype \rightarrow \NFPSh\abasetype
\\[2mm]
\fstmap^{(\tau')}_\abasetype 
& \!\! : & 
\NEPSh{\abasetype\ProductType\tau'} \rightarrow \NFPSh{\abasetype}
\\[2mm]
\sndmap^{(\tau')}_\abasetype 
& \!\! : & 
\NEPSh{\tau'\ProductType\abasetype} \rightarrow \NFPSh{\abasetype}
\\[2mm]
\appmap^{(\tau')}_\abasetype 
& \!\! : & 
\NEPSh{\tau'\ArrowType\abasetype} \CatProduct \NFPSh{\tau'}
  \rightarrow \NFPSh\abasetype
\\[2mm]
\unitmap_\UnitType 
& \!\! : & 
\TerminalPSh \stackrightarrow\iso \NFPSh\UnitType
\\[2mm]
\pairmap_{\tau\ProductType\tau'} 
& \!\! : & 
\NFPSh\tau \CatProduct \NFPSh{\tau'} 
  \stackrightarrow\iso \NFPSh{\tau\ProductType\tau'}
\\[2mm]
\absmap_{\tau\ArrowType\tau'} 
& \!\! : & 
(\NFPSh{\tau'})^{\VarPSh\tau} 
  \stackrightarrow\iso \NFPSh{\tau\ArrowType\tau'}
\end{array}\right.
\end{equation}
corresponding to the typing rules in Figure~\ref{NN_Terms}. 

Note that every $\SSSTerm$-algebra $\VVVar$ induces a canonical 
$\pair{\SSSNE,\SSSNF}$-algebra structure on the pair $(\VVVar,\VVVar)$ and 
hence, by initiality, homomorphic interpretations 
$(\NEPSh{},\NFPSh{}) \rightarrow (\VVVar,\VVVar)$.  Applying this
observation to the initial $\SSSTerm$-algebra $\TermPSh{}$ we obtain the
embeddings $\NEPSh{} \rightmono \TermPSh{}$ and 
$\NFPSh{} \rightmono \TermPSh{}$ of neutral and normal terms into terms.

\paragraph{Structural induction.}
Initial algebras have the following associated structural induction
principle~\cite{LehmannSmyth}. 
\begin{equation} \label{Induction_Principle}
\begin{minipage}{14.5cm}{\em
Let $\alpha: FA \rightarrow A$ be an
initial algebra for an endofunctor~$F$, if the subobject 
$m: P \rightmono A$ satisfies the closure property of being a
sub~$F$-algebra of $A$, in the sense that the diagram 
$$\xymatrix{
\raisebox{0.55mm}{$FP$} \ar[d]_-{Fm} \ar@{-->}[r] & \ar@{>->}[d]^-m
\raisebox{1.1mm}{$P$} \\
FA \ar[r]_-{\alpha}^-\iso & A
}$$
commutes for a (necessarily unique) map $FP \rightarrow P$, then 
$m: P \rightmono A$ is an isomorphism. 
}\end{minipage}
\end{equation}

For the initial $\SSSTerm$-algebra~$\TermPSh{}$ (resp.~the initial
$\pair{\SSSNE,\SSSNF}$-algebra~$(\NEPSh{},\NFPSh{})$) the structural induction
principle corresponds, in elementary terms, to proving a property of terms
(resp.~of neutral and normal terms) by induction (resp.~simultaneous
induction) on their derivation.
The structural induction principle for the initial 
$\pair{\SSSNE,\SSSNF}$-algebra~$(\NEPSh{},\NFPSh{})$ features in the proof of
Theorem~\ref{THEOREM}.

\subsection*{II.1.2\quad Semantics}

As we will see below, every interpretation of base types in a cartesian closed
category induces a canonical semantic typed lambda algebra with respect to
which the unique algebra homomorphism from the initial (term) algebra is the
usual semantics of simply typed terms. 

\paragraph{Nerve functor.}
Every functor $\arity: \catc \rightarrow \Semcat$ induces the following
situation \begin{equation}\label{Relative_Hom_Diagram}
\begin{array}{c}\xymatrix@C=5pt{ 
\ar@{}[rrrrrrrd]|-{\raisebox{5mm}{$\nt\arity\mbox{$\Downarrow${\tiny Lan}}$}}
\catc\ \ar@{^(->}[rrrrrrr]^-\yon
\ar[drrrr]_-{\arity} &&&&&&&
\Set^{\catc^\op} \\
&&&& \Semcat \ar[urrr]_-{\clone\arity} &&& 
}\end{array}
\end{equation}
where $\clone\arity(A) = \Semcat(\arity(\fstarg),A)$ and where
$(\nt\arity_\Gamma)_{\Gamma'} = \arity_{\Gamma',\Gamma}: 
   \catc(\Gamma',\Gamma) \rightarrow \Semcat(\arity(\Gamma'),\arity(\Gamma))$.
We refer to $\clone\arity: \Semcat \rightarrow \Set^{\catc^\op}$ as the
\Def{$\arity$-nerve functor} and to the presheaf $\clone\arity(A)$ as the
\Def{$\arity$-nerve} of $A$.

\medskip 
Two important properties of nerve functors follow.
\begin{proposition} \label{RelativeHom_Properties}
For a functor $\arity: \catc \rightarrow \Semcat$ where $\catc$ is small, the
nerve functor $\clone\arity: \Semcat \rightarrow \Set^{\catc^\op}$ preserves
limits.
Further, for $\arity$ and $\catc$ cartesian and $\Semcat$ cartesian
closed, it also commutes with exponentiation by representables in the sense
that there is a canonical natural isomorphism  
$$\xymatrix@C=50pt{\ar@{}[dr]|(.5)\iso
\Semcat \ar[d]_-{\clone\arity}  
\ar[r]^-{\arity(\Gamma) \;\CatArrow\, (\_\!\_)} & \Semcat 
\ar[d]^-{\clone\arity}
\\
\Set^{\catc^\op} \ar[r]_-{(\_\!\_)^{\yon(\Gamma)}}
&
\Set^{\catc^\op} 
}$$
for all $\Gamma \in \obj\catc$.
\end{proposition}  
\begin{proof}
The first part is well-known and follows from the canonical natural
isomorphism 
$$\begin{array}{rcl}
\Semcat(\arity(\Gamma),L)
& \iso & 
\lim_{\Delta\in\scat D}\: \Semcat(\arity(\Gamma),D(\Delta))
\\[2mm]
f & \mapsto & \pair{\, \proj_\Delta \comp f \,}_{\Delta\in\scat D}
\end{array}
\quad\qquad
(\Gamma\in\catc)
$$
available for any diagram $D: \scat D \rightarrow \Semcat$ with limit
$\pair{\, \proj_\Delta: L \rightarrow D(\Delta) \,}_{\Delta\in\scat D}$ 
in $\Semcat$. 

For the second part, note that for $\Gamma\in\obj\catc$ we have the
following situation~(generalising~(\ref{Adjoint_Situation_For_F[T]}))
$$
\begin{array}{c}\xymatrix@C=5pt{\ar@{}[rrrrrd]|(.5){\stackrel{\Lan}{\iso}}
\catc\ \ar@{^(->}[rrrrrr]^-\yon
\ar[d]_-{\fstarg\CatProduct\Gamma} &&&&&& 
\Set^{\catc^\op}
\ar[d]<-1ex>_-{\fstarg\CatProduct\yon(\Gamma)} \ar@{}[d]|-{\dashv}
\\
\catc\ \ar@{^(->}[rrrrrr]_-\yon &&&&&& 
\Set^{\catc^\op}
\ar[u]<-1ex>_-{\Set^{(\fstarg\CatProduct\Gamma)^\op}}
}\end{array}
$$
from which it follows that
$
(P^{\yon(\Gamma)})(\Delta)
\iso
P(\Delta\CatProduct\Gamma)
$ 
naturally in $\Delta\in\catc$.  
We thus obtain a canonical isomorphism
$$
\begin{array}{rcl}
(\clone\arity A)^{\yon(\Gamma)}(\Delta)
& \iso & (\clone\arity A)(\Delta\CatProduct\Gamma)
\ = \ \Semcat(\arity(\Delta\CatProduct\Gamma),A)
\\[2mm]
& \iso & \Semcat(\arity(\Delta)\CatProduct\arity(\Gamma),A)
\\[2mm]
& \iso & \Semcat(\arity(\Delta),\arity(\Gamma)\CatArrow A)
\ = \ \clone\arity(\arity(\Gamma)\CatArrow A)(\Delta)
\end{array}
$$
natural in $\Gamma,\Delta\in\catc$ and $A\in\Semcat$.
\end{proof}

\paragraph{Initial algebra semantics.}
Using the nerve functor
$\clone{\semfun\sembracket\fstarg}
 : \Semcat \rightarrow \Set^{\F{}\comma\TypeClosure\BaseTypeSet}$
induced by the cartesian extension
$\semfun\sembracket\fstarg:\F{}[\TypeClosure\BaseTypeSet]\rightarrow\Semcat$
of an interpretation \mbox{$\semfun:\BaseTypeSet\rightarrow\Semcat$} of base 
types in a cartesian closed category, the operations
$$\begin{array}{lcl}
\proj_1
& : &
\semfun\sembracket{\tau} \CatProduct \semfun\sembracket{\tau'}
\rightarrow 
\semfun\sembracket\tau
\\[2mm]
\proj_2
& : &
\semfun\sembracket{\tau'} \CatProduct \semfun\sembracket\tau 
\rightarrow 
\semfun\sembracket\tau
\\[2mm]
\eval
& : &
(\semfun\sembracket{\tau} \CatArrow \semfun\sembracket{\tau'}) 
\CatProduct
\semfun\sembracket{\tau}
\rightarrow 
\semfun\sembracket{\tau'}
\end{array}$$
in $\Semcat$ can be lifted to $\Set^{\F{}\comma\TypeClosure\BaseTypeSet}$ to
provide a \emph{semantic typed lambda algebra} structure on the family
\begin{equation} \label{Semantic_Algebra_Carrier}
\SemPSh{} 
\ = \
\setof{ 
        \ \Semcat(\semfun\sembracket\fstarg,\semfun\sembracket\tau) \
        }_{\tau \in \TypeClosure\BaseTypeSet}
\ = \
\setof{ \ \clone\semfun(\semfun\sembracket\tau) \
        }_{\tau \in \TypeClosure\BaseTypeSet}
\end{equation}
The operations are as follows:
\begin{equation} \label{Semantic_Algebra_Structure}
\begin{minipage}{.925\textwidth}
\begin{enumerate}
\item \label{Semantic_Operations_Var}
$\VarPSh\tau 
\longstackrightarrow{\mbox{\scriptsize$\semfun\sembracket\fstarg$}}
\clone\semfun(\semfun\sembracket\tau)$

\item \label{Semantic_Operations_Unit}
$\TerminalPSh \stackrightarrow\iso \clone\semfun(\semfun\sembracket\UnitType)$

\item
$\clone\semfun(\semfun\sembracket{\tau\ProductType\tau'})
 \longstackrightarrow{\mbox{\scriptsize$\clone\semfun(\proj_1)$}}
 \clone\semfun(\semfun\sembracket\tau)$

\item
$\clone\semfun(\semfun\sembracket{\tau'\ProductType\tau})
 \longstackrightarrow{\mbox{\scriptsize$\clone\semfun(\proj_2)$}}
 \clone\semfun(\semfun\sembracket\tau)$

\item \label{Semantic_Operations_Pair}
$\clone\semfun(\semfun\sembracket\tau)
 \CatProduct
 \clone\semfun(\semfun\sembracket{\tau'})
 \stackrightarrow\iso
 \clone\semfun(\semfun\sembracket{\tau\ProductType\tau'})$

\item \label{Semantic_Operations_Eval}
$
\clone\semfun(\semfun\sembracket{\tau'\ArrowType\tau})
\CatProduct
\clone\semfun(\semfun\sembracket{\tau'})
\stackrightarrow\iso
\clone\semfun((\semfun\sembracket{\tau'}\CatArrow\semfun\sembracket\tau)
              \CatProduct
              \semfun\sembracket{\tau'})
\longstackrightarrow{\mbox{\scriptsize$\clone\semfun(\eval)$}}
\clone\semfun(\semfun\sembracket{\tau})
$

\item \label{Semantic_Operations_Abs}
$(\clone\semfun(\semfun\sembracket{\tau'}))^{\VarPSh\tau}
 \stackrightarrow\iso
 \clone\semfun(\semfun\sembracket{\tau\ArrowType\tau'})$
\end{enumerate}
\end{minipage}
\end{equation}
(Note that item~\ref{Semantic_Operations_Var} relies on
diagram~(\ref{Relative_Hom_Diagram}) while
items~\ref{Semantic_Operations_Unit}, \ref{Semantic_Operations_Pair},
\ref{Semantic_Operations_Eval}, and \ref{Semantic_Operations_Abs} rely on
Proposition~\ref{RelativeHom_Properties}.  Similar applications of this
proposition will be used throughout without further reference.)

\medskip
By initiality, the semantic typed lambda algebra induces semantic homomorphic
interpretations $\Termsem{}: \TermPSh{} \rightarrow \SemPSh{}$ and
$(\NEsem{},\NFsem{}): (\NEPSh{},\NFPSh{}) \rightarrow (\SemPSh{},\SemPSh{})$.
These are related as shown below
\begin{equation}\label{Initial_Algebra_Semantics}
\begin{array}{c}\xymatrix{
\NEPSh{}\ \ar[rd]_-{\NEsem{}} \ar@{>->}[r] & \TermPSh{} \ar[d]_-{\Termsem{}} 
& \ar@{>->}[l] \ar[ld]^-{\NFsem{}} \ \NFPSh{} \\ 
& \SemPSh{} & 
}\end{array}
\end{equation}
Indeed, by the initiality of $(\NEPSh{},\NFPSh{})$,
(\ref{Initial_Algebra_Semantics}) directly follows from the fact that the
homomorphism property of $\Termsem{}: \TermPSh{} \rightarrow \SemPSh{}$
amounts to the commutativity of the diagrams in Appendix~\ref{Appendix_l} and
that the homorphism property of
$(\NEsem{},\NFsem{}): (\NEPSh{},\NFPSh{}) \rightarrow (\SemPSh{},\SemPSh{})$
amounts to the commutativity of the diagrams in Appendix~\ref{Appendix_(m,n)}.

Explicitly, for $\tau\in\TypeClosure\BaseTypeSet$, the mapping
$\Termsem\tau: \TermPSh\tau \rightarrow \SemPSh\tau$ is the standard semantic
interpretation of terms
\begin{equation}\label{StandardSemanticInterpretation}
t\in\TermPSh\tau(\Gamma)
\xymatrix{\ar@{|->}[r]^-{\Termsem\tau}&}
\semfun\sembracket{\Gamma \entails t: \tau}
  \in\Semcat(\semfun\sembracket\Gamma,\semfun\sembracket\tau)
\end{equation}
whilst $\NEsem\tau: \NEPSh\tau \rightarrow \SemPSh\tau$ and 
$\NFsem\tau: \NFPSh\tau \rightarrow \SemPSh\tau$ are, respectively, the 
semantic interpretations of neutral and normal terms.  

\section*{II.2\quad Normalisation by evaluation via categorical glueing}

We will now see how, by working with \emph{intensional Kripke relations}, the
analysis of normalisation given in Section~I.2 amounts to normalisation by
evaluation.   
As in that section, we will work with semantic models of (covariant)
presheaves in $\Set^{\F{}\comma\TypeClosure\BaseTypeSet}$ over nerves induced
by interpretations~$\semfun$ 
of the set of base types~$\BaseTypeSet$ in arbitrary cartesian closed
categories~(see~(\ref{mu_Interpretation})).  
This level of generality allows the definition of normalisation functions
$\nf\tau\semfun: \TermPSh\tau \rightarrow \NFPSh\tau$
($\atype\in\TypeClosure\BaseTypeSet$) in
$\Set^{\F{}\comma\TypeClosure\BaseTypeSet}$ over the
$\semfun\sembracket\fstarg$-nerve of $\semfun\sembracket\atype$
(Corollary~\ref{Corollary_Correctness_Two}) that are parametric on the
interpretation~$\semfun$.  
Crucially, the normalisation functions will be shown to be parametrically
polymorphic, in the sense of being interpretation
independent~(Corollary~\ref{Corollary_nfCoincidence}).
This is methodologically important.  
Firstly, as in Corollary~\ref{Corollary_BetaEtaLongNormalForms}, the
consideration of the universal interpretation of base types into the free
cartesian closed category over them leads to our solution of the
\emph{intensional normalisation problem}~(see the discussions after 
Corollaries~\ref{Corollary_Correctness_Two} and~\ref{Corollary_nfCoincidence}
in~\S~\emph{Normalisation function} below) stated in the Introduction.  
Secondly, the consideration of the trivial interpretation of base types in the
trivial cartesian closed category leads to a normalisation algorithm from
which a normalisation program is 
synthesised~(see~\S~\emph{Normalisation algorithm} below).  

\paragraph{Intensional Kripke relations.}
The category of intensional $\catc$-Kripke relations of arity 
$\arity: \catc \rightarrow \Semcat$ is defined as the glueing
of $\Set^{\catc^\op}$ and $\Semcat$ along the nerve functor 
$\clone\arity: \Semcat \rightarrow \Set^{\catc^\op}$.  
That is, as the comma category
\mbox{$\Set^{\catc^\op}\comma\clone\arity$} of objects given by triples 
$(P,p,A)$ with $P\in\obj{\Set^{\catc^\op}}$, $A\in\obj{\Semcat}$, and
$p: P \rightarrow \clone\arity(A)$ in $\Set^{\catc^\op}$, and of
morphisms $(P,p,A) \rightarrow (P',p',A')$ given  
by pairs 
$$
(~\varphi:P\rightarrow P'\mbox{ in $\Set^{\catc^\op}$} 
 ,~
 f:A\rightarrow A'\mbox{ in $\Semcat$}
 ~)
$$
such that the diagram
$$
\begin{array}{c}\xymatrix{
P \ar[d]_-p \ar[rr]^-{\varphi} && P' \ar[d]^-{p'} \\
\clone\arity(A) \ar[rr]_-{\clone\arity(f)} && \clone\arity(A')
}\end{array}
\quad
\mbox{ in $\Set^{\catc^\op}$}
$$
commutes. 

\begin{example} \label{Trivial_Glueing}
The category of intensional $\catc$-Kripke relations of arity the unique
functor to the terminal category is (isomorphic to) the presheaf topos
$\Set^{\catc^\op}$.  
\end{example}

As it is well-known (see,~\eg,~\cite{LambekScott,Crole,Taylor}), for $\Semcat$
cartesian closed, the glueing category $\Set^{\catc^\op}\comma\clone\arity$ is
also cartesian closed. 
Indeed, the cartesian closed structure of
$\Set^{\catc^\op}\comma\clone\arity$ is given as follows.  
\begin{description}
\item[](Products)
  The terminal object is 
  $(\TerminalPSh,t,\TerminalObj)$ where $t$ is the unique map 
  $\TerminalPSh\stackrightarrow\iso\clone\arity(\TerminalObj)$.

  The binary 
  product $(P,p,A)\CatProduct(Q,q,B)$
  of $(P,p,A)$ and $(Q,q,B)$
  is $(P\CatProduct Q,r,A\CatProduct B)$ where $r$ is the composite
  $
  P\CatProduct Q
    \longstackrightarrow{p\CatProduct q}
      \clone\arity(A)\CatProduct\clone\arity(B)
        \stackrightarrow\iso
          \clone\arity(A\CatProduct B)$.

\item[](Exponentials)
  The exponential \mbox{$(P,p,A)\CatArrow(Q,q,B)$} of $(P,p,A)$ and
  $(Q,q,B)$ is $(R,r,A\CatArrow B)$ in the pullback diagram
  \begin{equation}\label{Exponential_By_Pullback}
  \begin{array}{c}\xymatrix{\ar@{}[rrrrd]|-{\mathrm{pb}}
  R \ar[rrrr] \ar[d]_-{r} && && \ar[d]^-{q^P} Q^{P}
  \\
  \clone\arity(A\CatArrow B) \ar[rr] &&
  (\clone\arity B)^{(\clone\arity A)}
  \ar[rr]_-{(\clone\arity B)^p} && 
  (\clone\arity B)^{P}
  }\end{array}
  \end{equation}
  where the map 
  $\clone\arity(A\CatArrow B) \rightarrow (\clone\arity B)^{(\clone\arity A)}$ 
  is the exponential transpose of the composite
  $$
  \clone\arity(A\CatArrow B)\CatProduct\clone\arity(A)
    \stackrightarrow\iso
  \clone\arity((A\CatArrow B)\CatProduct A)
    \longstackrightarrow{\clone\arity(\eval)}
  \clone\arity(B)
  \quad.
  $$

  Explicitly, one may take $R(c)$ to be
  \begin{equation}\label{PShExponentialNaturality}
  \!\!\!\!\!\!\!\!\!\!\!\!\!\!\!\!\!\!
    \left\{\begin{array}{l|l}
      ( \, f: \arity(c)\to A\CatArrow B 
          \, , \, 
          \varphi : \yon(c)\times P\to Q
        \, )
      &
      \begin{array}{l}
        \forall\, \rho: c'\to c.\ \forall\, a\in P(c').\
        \\[.5mm]
        \quad
        q_{c'}(\varphi_{c'}(\rho,a))
        =
        \eval \comp \pair{ f\comp \arity(\rho) , p_{c'}(a) }
      \end{array}
      \end{array}\right\}
  \end{equation}
  with $r$ projecting pairs onto their first component.
\end{description}
\begin{proposition}\label{Gluecat_CCC}
Let $\catc$ be a small category and let $\Semcat$ be a cartesian closed
category.  For a functor \mbox{$\arity: \catc \rightarrow \Semcat$}, the 
glueing category \mbox{$\Set^{\catc^\op}\comma\clone\arity$} is cartesian
closed and the forgetful functor
$\Sproj
 : {\Set^{\catc^\op}\comma\clone\arity} \rightarrow \Semcat
 : (P,p,A)\mapsto A$
preserves the cartesian closed structure strictly.
\end{proposition}

\begin{remark}
The category of $\catc$-Kripke relations $\K\arity$ is a full subcategory of
the glueing category $\Set^{\catc^\op}\comma\clone{\arity}$ via the mapping
$(R,A) \mapsto (R,R\rightembedding\clone\arity(A),A)$.  
On the other hand, every glued object $(P,f,A)$ has an associated Kripke
relation given by the extension of the map~$f$ (as shown in the diagram below,
where $\im(f)$ denotes the image of $f$)
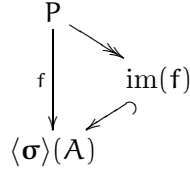
$$\xymatrix@C=5pt@R=10pt{
P \ar[dd]_-f \ar@{->>}[rd] \\
& \im(f) \ar@{^(->}[dl] \\
\clone\arity(A) 
}$$
and the mapping $\im: (P,f,A) \mapsto (\im(f),A)$ exhibits $\K\arity$ as a
reflective subcategory of $\Set^{\catc^\op}\comma\clone\arity$.  
For $\Semcat$ cartesian closed, as can be readily seen from the explicit
descriptions of finite products in $\K\arity$ and
$\Set^{\catc^\op}\comma\clone\arity$, the reflection 
$\im: \K\arity \to \Set^{\catc^\op}\comma\clone\arity$ preserves the cartesian
structure and, therefore, $\K\arity$ is an exponential ideal of
$\Set^{\catc^\op}\comma\clone\arity$ (as can also be readily seen from the
descriptions of exponentials in $\K\arity$ and
$\Set^{\catc^\op}\comma\clone\arity$).  Thus, for $(P,p,A)$ and $(Q,q,A)$ in 
$\Set^{\catc^\op}\comma\clone\arity$, there are inclusions
\[
  \im( \ (P,p,A)\CatArrow\!(Q,q,B) \ )(c)
  \ \subseteq \
  ( \ \im(P,p,A) \!\ArrowLogRel \im(Q,q,B) \ )(c)
  \qquad(\,c\in\obj{\catc}\,)
\]
where 
  \[\begin{array}{l}
  \im( \ (P,p,A)\CatArrow\!(Q,q,B) \ )(c)
  \\[2.5mm] \quad
    =
    \left\{\begin{array}{l|l}
      f: \arity(c)\to A\CatArrow B 
      &
      \begin{array}{l}
        \exists\, \varphi : \yon(c)\times P\to Q.\
        \forall\, \rho: c'\to c.\ \forall\, a\in P(c').\
        \\[1mm]
        \quad
        q_{c'}(\varphi_{c'}(\rho,a))
        =
        \eval \comp \pair{ f\comp \arity(\rho) , p_{c'}(a) }
      \end{array}
      \end{array}\right\}
  \end{array}\]
and
\[\begin{array}{l}
  ( \ \im(P,p,A) \!\ArrowLogRel \im(Q,q,B) \ ) (c) 
  \\[2.5mm] \quad
    =
    \left\{\begin{array}{l|l}
      f: \arity(c)\to A\CatArrow B 
      &
      \begin{array}{l}
        \forall\, \rho: c'\to c.\ \forall\, a\in P(c').\
        \exists\, b\in Q(c').\
        \\[1mm]
        \quad
        q_{c'}(b)
        =
        \eval \comp \pair{ f\comp \arity(\rho) , p_{c'}(a) }
      \end{array}
      \end{array}\right\}\ .
\end{array}\]
These inclusions may be strict; as it happens, for instance, when
$\catc^\op=\F{}$~(the category of untyped contexts and renamings), $\arity$ is
the unique functor to the trivial cartesian closed category, $Q=\yon(1)$ (for
a singleton context $1$), $P=\im(Q\to\clone\arity(\TerminalObj))$, and
$c=0$~(the empty context).  Indeed, in this situation,
$(P\CatArrow\!\!Q)(c)\iso\Set^{\F{}}(P,Q)=\emptyset$ whilst
$(\im(p)\ArrowLogRel\im(q))(c)=(P\ArrowLogRel P)(c)=\setof{\id}$.  
Thus, in general, the reflection 
$\im: \K\arity \to \Set^{\catc^\op}\comma\clone\arity$ does not preserve
exponentials.
\end{remark}

Now, note that~(\ref{Relative_Hom_Diagram}) induces the embedding 
$$\begin{array}{rcl}
\catc & \longstackrightembedding\commayon & \Set^{\catc^\op}\comma\clone\arity
\\[1mm]
\Gamma 
& \mapsto & 
( \ \yon(\Gamma)\ , 
  \ \yon(\Gamma) 
      \longstackrightarrow{\nt\arity_\Gamma} 
        \clone\arity(\arity \Gamma)\ ,
  \ \arity(\Gamma)\ )
\end{array}$$
extending both 
the Yoneda embedding $\yon:\catc\rightembedding\Set^{\catc^\op}$ and the
functor $\arity:\catc\rightarrow\Semcat$ 
$$\xymatrix@R=0pt{
& \catc \ar@{^(->}[ldddd]_-\yon \ar@{^(->}[dddd]^-\commayon \ar[rdddd]^-\arity 
\\ \\ \\ \\
\Set^{\catc^\op} & \ar[l] \Set^{\catc^\op}\comma\clone\arity \ar[r] & \Semcat
\\
P & \ar@{|->}[l]_-\Pproj (P,p,A) \ar@{|->}[r]^-\Sproj &  A
}$$
and satisfying the following extended form of the Yoneda Lemma (which we will
use in \S~\emph{Normalisation function} below).
\begin{lemma}\Rem{(Extended Yoneda Lemma)}
For a functor $\arity: \catc \rightarrow \Semcat$ where $\catc$ is a
small category, the natural transformation
$$
\Hom{\commayon(\fstarg),(P,p,A)} \rightarrow P(\fstarg):\
(\varphi,f) \mapsto \varphi(\id)
\quad,
$$
where $\Hom{\fstarg,\sndarg}$ denotes the hom-functor of the glueing
category $\Set^{\catc^\op}\comma\clone\arity$, is an isomorphism making
the following diagram  
$$\xymatrix@R=0pt@C=10pt{
\Hom{\commayon(\fstarg),(P,p,A)} \ar[rr]^-\iso \ar[ddddr]_-\Sproj &&
P(\simplefstarg) \ar[ddddl]^-p
\\ \\ \\ \\
& \Semcat(\arity(\fstarg),A) 
}$$
commute.
\end{lemma}
\begin{proof}
Follows from the fact that, for $\varphi: \yon(\Gamma) \rightarrow P$ in
$\Set^{\catc^\op}$ and $f: \arity(\Gamma) \rightarrow A$ in $\Semcat$,
the diagram
$$\xymatrix{
\yon(\Gamma) \ar[d]_{\nt\arity_\Gamma} \ar[r]^-\varphi &  P \ar[d]^-p
\\
\Semcat(\arity(\fstarg),\arity(\Gamma)) \ar[r]_-{f \comp \simplefstarg} &
\Semcat(\arity(\fstarg),A)
}$$
commutes if and only if $f = p_\Gamma(\varphi_\Gamma(\id_\Gamma))$.
\end{proof}

\begin{proposition} \label{Glued_Exponential_to_Representable}
For a functor $\arity: \catc \rightarrow \Semcat$ where $\catc$ is
small, $\catc$ and $\arity$ are cartesian, and $\Semcat$ is cartesian
closed, we have that 
$\commayon: \catc \rightembedding \Set^{\catc^\op}\comma\clone\arity$
preserves products and that the exponential
$(P,p,A)^{\commayon(\Gamma)}$ in $\Set^{\catc^\op}\comma\clone\arity$
can be described as $(P^{\yon(\Gamma)},p',\arity(\Gamma)\ArrowType A)$ where
$p'$ is the composite
$$
P^{\yon(\Gamma)} 
  \longstackrightarrow{p^{\yon(\Gamma)}}
    (\clone\arity A)^{\yon(\Gamma)}
      \stackrightarrow\iso 
        \clone\arity(\arity(\Gamma)\CatArrow A)
\quad.
$$
\end{proposition}
\begin{proof}
The first part follows from the commutativity of
$$\xymatrix{
\yon(\Gamma\CatProduct\Delta) \ar[r]^-\iso
\ar[dd]_-{\nt\arity_{\Gamma\CatProduct\Delta}} & 
\yon(\Gamma) \CatProduct \yon(\Delta)
\ar[d]^-{\nt\arity_\Gamma\CatProduct\nt\arity_\Delta}
\\
& 
\clone\arity(\arity\Gamma)\CatProduct\clone\arity(\arity\Delta)
\ar[d]^-\iso
\\
\clone\arity(\arity(\Gamma\CatProduct\Delta)) \ar[r]_-\iso
& 
\clone\arity(\arity(\Gamma)\CatProduct\arity(\Delta)) 
}$$
for all $\Gamma,\Delta\in\obj\catc$. 

For the second part, since the exponential $(P,p,A)^{\commayon(\Gamma)}$
is given by pulling back 
the map
$p^{\yon(\Gamma)}: 
   P^{\yon(\Gamma)} \rightarrow (\clone\arity A)^{\yon(\Gamma)}$
along the composite
$$
\clone\arity(\arity\Gamma \CatArrow A)
  \stackrightarrow f
(\clone\arity A)^{\clone\arity(\arity\Gamma)}
  \longlongstackrightarrow{(\clone\arity A)^{\nt\arity_\Gamma}}
(\clone\arity A)^{\yon(\Gamma)}
$$
(recall~(\ref{Exponential_By_Pullback})), where $f$ is the exponential
transpose of  
$$\clone\arity(\arity\Gamma \CatArrow A) \CatProduct \clone\arity(\arity\Gamma)
  \stackrightarrow\iso
 \clone\arity((\arity\Gamma \CatArrow A) \CatProduct \arity\Gamma)
  \longstackrightarrow{\clone\arity(\eval)}
 \clone\arity(A)\quad,$$
it will be enough to show that the composite
$$
(\clone\arity A)^{\yon(\Gamma)}
  \stackrightarrow\iso
\clone\arity(\arity\Gamma \CatArrow A)
  \stackrightarrow f
(\clone\arity A)^{\clone\arity(\arity\Gamma)}
  \longlongstackrightarrow{(\clone\arity A)^{\nt\arity_\Gamma}}
(\clone\arity A)^{\yon(\Gamma)}
$$
is the identity.  This is indeed the case as follows from the
commutativity of the diagram below
{\scriptsize$$
\!\!\!\!\!\!\!\!\!\!\!\!\!\!\!\!\!\!\!\!
\xymatrix{
& 
\clone\arity(\arity\Gamma \CatArrow A) \CatProduct \yon(\Gamma)
\ar[rr]^-{\mbox{$f \CatProduct \id$}} 
\ar[dr]^-{\mbox{$\id\CatProduct\nt\arity_\Gamma$}}
\ar[d]^-{\mbox{$\iso\CatProduct\id$}}
& & 
(\clone\arity A)^{\clone\arity(\arity\Gamma)} \CatProduct \yon(\Gamma)
\ar[rd]^-{\mbox{$\id^{\nt\arity_\Gamma}\CatProduct\id$}} 
\ar[d]^-{\mbox{$\id\CatProduct\nt\arity_\Gamma$}}
&
\\
(\clone\arity A)^{\yon(\Gamma)} \CatProduct \yon(\Gamma) 
\ar[ru]^-{\mbox{$\iso\CatProduct\id$}} 
\ar[dddrr]_-{\mbox{$\eval$}}
\ar[r]^-{\mbox{$\iso\CatProduct\id$}}
& 
(\clone\arity A)(\fstarg\CatProduct\Gamma) \CatProduct \yon(\Gamma)
\ar[r]^-{\mbox{$\iso\CatProduct\nt\arity_\Gamma$}}
\ar[dddr]_-{\mbox{$e$}}
& 
\clone\arity(\arity\Gamma \CatArrow A) \CatProduct \clone\arity(\arity\Gamma)
\ar[r]^-{\mbox{$f\CatProduct \id$}}
\ar[d]^-{\mbox{$\iso$}}
& 
(\clone\arity A)^{\clone\arity(\arity \Gamma)} 
  \CatProduct \clone\arity(\arity\Gamma)
\ar[dddl]^-{\mbox{$\eval$}}
& 
(\clone\arity A)^{\yon(\Gamma)} \CatProduct \yon(\Gamma)
\ar[dddll]^-{\mbox{$\eval$}}
\\
& & 
\clone\arity((\arity\Gamma\CatArrow A)\CatProduct \arity\Gamma)
\ar[dd]|-{\mbox{$\clone\arity(\eval)$}}
& & 
\\
\\
& & \clone\arity(A) & &
}$$}%
where 
\begin{equation} \label{RepresentableExponentialEvaluationMap}
e_P 
\, : \, 
P(\fstarg\CatProduct\Gamma)\CatProduct\yon(\Gamma) \rightarrow P(\fstarg)
\, : \, 
(x,\rho) \mapsto (P\pair{\id,\rho})(x)
\end{equation}
denotes the counit of the adjunction 
$\fstarg \CatProduct \yon(\Gamma) 
   \dashv \Set^{(\fstarg\CatProduct\Gamma)^\op}: 
     \Set^{\catc^\op} \rightarrow \Set^{\catc^\op}$.
\end{proof}

\paragraph{Glueing syntax and semantics.}
Let \mbox{$\semfun:\BaseTypeSet\rightarrow\Semcat$} be an interpretation of
base types in a cartesian closed category.  The embedding
\mbox{$\commayon: \F{}[\TypeClosure\BaseTypeSet] \rightembedding \Gluecat$}
restricted to types $\atype\in\TypeClosure\BaseTypeSet$ yields the glued 
object
$$
\Varobj\atype 
\ = \
\commayon\ctxtemb\atype
\ = \
(\ \VarPSh\atype\ ,\ 
   \VarPSh\atype
     \longstackrightarrow{\mbox{\scriptsize$\semfun\sembracket{\fstarg}$}}  
       \SemPSh\atype\ ,\
   \semfun\sembracket\atype\ )
\qquad\mbox{in $\Gluecat$}
$$
glueing the syntax and semantics of variables.  In the same
spirit, glueing the syntax and semantics of neutral and normal
terms~(see~(\ref{Initial_Algebra_Semantics})) we obtain the glued objects   
$$
\begin{array}{rcl}
\NEobj\tau & = & 
(\ \NEPSh\tau\ ,\ 
   \NEPSh\tau\longstackrightarrow{\NEsem\tau}\SemPSh\tau\ ,\
   \semfun\sembracket\tau\ )
\\[2mm]
\NFobj\tau & = & 
(\ \NFPSh\tau\ ,\ 
   \NFPSh\tau\longstackrightarrow{\NFsem\tau}\SemPSh\tau\ ,\
   \semfun\sembracket\tau\ )
\end{array}
$$
in $\Gluecat$.

Having constructed the $\pair{\SSSNE,\SSSNF}$-algebra structure on
$(\SemPSh{},\SemPSh{})$ by lifting the semantic operations in
$\Semcat$~(recall~(\ref{Semantic_Algebra_Carrier})
and~(\ref{Semantic_Algebra_Structure})), the homomorphism property of the
semantic interpretation 
$(\NEsem{},\NFsem{}):(\NEPSh{},\NFPSh{}) \rightarrow (\SemPSh{},\SemPSh{})$
(see Appendix~\ref{Appendix_(m,n)}) 
entails the two propositions below,
which show how the algebraic operations on the initial
$\pair{\SSSNE,\SSSNF}$-algebra $(\NEPSh{},\NFPSh{})$ and on the semantic
$\pair{\SSSNE,\SSSNF}$-algebra $(\SemPSh{},\SemPSh{})$ can be glued to
yield operations in $\Gluecat$ on the pair of families of glued objects 
$( \setof{ \NEobj\tau }_{\tau\in\TypeClosure\BaseTypeSet}\ ,\ 
   \setof{ \NFobj\tau }_{\tau\in\TypeClosure\BaseTypeSet} )$.
\begin{proposition} \label{mu_Operations}
Let $\semfun: \BaseTypeSet \rightarrow \Semcat$ be an interpretation of base
types in a cartesian closed category.
\begin{enumerate}
\item \label{mu_Operations_One}
For $\tau, \tau' \in \TypeClosure\BaseTypeSet$, the pair of maps
$$
( \
\varmap_\tau: \VarPSh\tau \rightarrow \NEPSh\tau
\ , \
\id_{\mbox{\scriptsize$\semfun\sembracket\tau$}}
\ )
$$
constitute a map $\Varobj\tau \rightarrow \NEobj\tau$ in $\Gluecat$.

\item \label{mu_Operations_Two}
For $\tau, \tau' \in \TypeClosure\BaseTypeSet$, the pair of maps
$$
( \
\fstmap^{(\tau')}_\tau: \NEPSh{\tau\ProductType\tau'} \rightarrow \NEPSh\tau
\ , \
\proj_1: 
  \semfun\sembracket{\tau}\CatProduct\semfun\sembracket{\tau'}
  \rightarrow
  \semfun\sembracket{\tau}
\ )
$$
constitute a map $\NEobj{\tau\ProductType\tau'} \rightarrow \NEobj\tau$ in
$\Gluecat$.

\item \label{mu_Operations_Three}
For $\tau, \tau' \in \TypeClosure\BaseTypeSet$, the pair of maps
$$
( \ 
\sndmap^{(\tau')}_\tau: \NEPSh{\tau'\ProductType\tau} \rightarrow \NEPSh\tau
\ , \
\proj_2: 
  \semfun\sembracket{\tau'}\CatProduct\semfun\sembracket{\tau}
  \rightarrow
  \semfun\sembracket{\tau}
\ )
$$
constitute a map $\NEobj{\tau'\ProductType\tau} \rightarrow \NEobj\tau$ in
$\Gluecat$.

\item \label{mu_Operations_Four}
For $\tau, \tau' \in \TypeClosure\BaseTypeSet$, the pair of maps
$$
( \ 
\appmap^{(\tau')}_\tau: 
  \NEPSh{\tau'\ArrowType\tau}\CatProduct\NFPSh{\tau'}
  \rightarrow
  \NEPSh{\tau}
\ , \
\eval: (\semfun\sembracket{\tau'}\CatArrow\semfun\sembracket\tau)
       \CatProduct
       \semfun\sembracket{\tau'}
       \rightarrow
       \semfun\sembracket\tau
\ ) 
$$
constitute a map 
$\NEobj{\tau'\ArrowType\tau} \CatProduct \NFobj{\tau'} 
 \rightarrow 
 \NEobj\tau$ 
in $\Gluecat$.
\end{enumerate}
\end{proposition}
\begin{proof}
Items~\ref{mu_Operations_One}, \ref{mu_Operations_Two},
\ref{mu_Operations_Three}, and \ref{mu_Operations_Four} respectively follow
from~(\ref{Proof_Of_mu_Operation_One}), (\ref{Proof_Of_mu_Operation_Two}),
(\ref{Proof_Of_mu_Operation_Three}), and (\ref{Proof_Of_mu_Operation_Four}) in
Appendix~\ref{Appendix_(m,n)}.
\end{proof}
\begin{proposition} \label{nu_Operations}
Let $\semfun: \BaseTypeSet \rightarrow \Semcat$ be an interpretation of base
types in a cartesian closed category.
\begin{enumerate}
\item \label{nu_Operation_One}
For a base type $\abasetype \in \BaseTypeSet$, the pair of isomorphisms
$$
( \
\NEPSh\abasetype 
  \iso \VarPSh\abasetype\CatCoproduct\EEE\abasetype(\NEPSh{},\NFPSh{})
    \iso \NFPSh\abasetype
\ ,\ 
\id_{\mbox{\scriptsize$\semfun(\abasetype)$}}
\ )
$$
constitute an isomorphism $\NEobj\abasetype \iso \NFobj\abasetype$ in
$\Gluecat$.

\item \label{nu_Operation_Two}
The pair of isomorphisms
$$ 
( \
\unitmap_\UnitType: \TerminalPSh \stackrightarrow\iso \NFPSh\UnitType
\ , \
\id_\TerminalObj
\ )
$$
constitute an isomorphism $\terminalobj \stackrightarrow\iso \NFobj\UnitType$
in $\Gluecat$.

\item \label{nu_Operation_Three}
For $\tau, \tau' \in \TypeClosure\BaseTypeSet$, the pair of isomorphisms
$$
( \ 
\pairmap_{\tau\ProductType\tau'}:
  \NFPSh\tau \CatProduct \NFPSh{\tau'}
    \stackrightarrow\iso
      \NFPSh{\tau\ProductType\tau'}
\ , \
\id_{\mbox{\scriptsize$\semfun\sembracket\tau\CatProduct\semfun\sembracket{\tau'}$}}
\ )
$$
constitute an isomorphism 
$\NFobj\tau\CatProduct\NFobj{\tau'} 
 \medstackrightarrow\iso
 \NFobj{\tau\ProductType\tau'}$ 
in $\Gluecat$.

\item \label{nu_Operation_Four}
For $\tau, \tau' \in \TypeClosure\BaseTypeSet$, the pair of isomorphisms
$$
( \ 
\absmap_{\tau\ArrowType\tau'}: 
  {\NFPSh{\tau'}}^{\VarPSh\tau} 
  \stackrightarrow\iso
  \NFPSh{\tau\ArrowType\tau'}
\ , \
\id_{\mbox{\scriptsize$\semfun\sembracket\tau\CatArrow\semfun\sembracket{\tau'}$}}
\ )
$$
constitute an isomorphism 
${\NFobj{\tau'}}^{\Varobj\tau} 
 \medstackrightarrow\iso 
 \NFobj{\tau\ArrowType\tau'}$
in $\Gluecat$.
\end{enumerate}
\end{proposition}
\begin{proof}
Items~\ref{nu_Operation_One}, \ref{nu_Operation_Two},
\ref{nu_Operation_Three}, and \ref{nu_Operation_Four}
respectively follow from 
~(\ref{Proof_Of_mu_Operation_One}--
\ref{Proof_Of_nu_Operation_One(iv)}),
(\ref{Proof_Of_nu_Operation_Two}),
(\ref{Proof_Of_nu_Operation_Three}), 
and
(\ref{Proof_Of_nu_Operation_Four}) (relying on
Proposition~\ref{Glued_Exponential_to_Representable})
in Appendix~\ref{Appendix_(m,n)}.
\end{proof}
Note that the above operations on glued objects are given by pairs of
syntactic operations together with their associated semantic meaning in the
case of 
neutral terms~(Proposition~\ref{mu_Operations}) 
and together with the identity in the case of 
normal terms~(Proposition~\ref{nu_Operations}). 

\paragraph{Normalisation by evaluation.}
Let $\semfun:\BaseTypeSet\rightarrow\Semcat$ be an interpretation of base
types in a cartesian closed category.  Consider the interpretation
\begin{equation} \label{mu_Interpretation}
\begin{array}{rcl}
\BaseTypeSet & \longstackrightarrow{\ext\semfun} & \Gluecat \\
\abasetype & \mapsto & \NEobj\abasetype
\end{array}
\end{equation}
By Proposition~\ref{Gluecat_CCC}, the semantics of terms induced by
$\ext\semfun$ in $\Gluecat$ extends the semantics induced by $\semfun$ in
$\Semcat$; that is, 
the denotation $\ext\semfun\sembracket{\Gamma\entails t:\tau}$ is a pair 
of the form 
\[
(\ \semfun'\sembracket{\Gamma\entails t:\tau} 
   \ , \
   \semfun\sembracket{\Gamma\entails t:\tau}) \ )
\]
such that, letting 
\[
\ext\semfun\sembracket\tau
\ = \
(\, \NSSemPSh\tau \,,\, \NSsem\tau \,,\, \semfun\sembracket\tau \,)
\enspace, 
\]
the diagram
$$\xymatrix@R=20pt@C=6pt{
\prod_{i=1,n} \NSSemPSh{\tau_i}
\ar[d]_-{\prod_{i=1,n}\NSsem{\tau_i}} 
\ar[r]^-{\mbox{\scriptsize$\semfun'\sembracket{\Gamma\entails t:\tau}$}} 
& \NSSemPSh\tau \ar[dd]^-{\NSsem\tau} 
\\
\prod_{i=1,n} 
  \Semcat(\semfun\sembracket\simplefstarg,\semfun\sembracket{\tau_i})
\ar[d]|-\iso
\\
\Semcat(\semfun\sembracket\simplefstarg,\semfun\sembracket{\Gamma}) 
\ar[r]_-{\raisebox{-4mm}{\scriptsize
         $\semfun\sembracket{\Gamma\entails t:\tau}\comp\simplefstarg$}}
& \Semcat(\semfun\sembracket\simplefstarg,\semfun\sembracket{\tau}) 
}$$
commutes for all $\Gamma=\pair{x_i:\tau_i}_{i=1,n}$. 

We now aim at defining maps 
$\xymatrix@C=15pt{\NEPSh\tau\ar@{.>}[r]&\NSSemPSh\tau\ar@{.>}[r]&\NFPSh\tau}$
($\tau\in\TypeClosure\BaseTypeSet$)
such that
\begin{equation} \label{uq_Maps}
\begin{array}{c}\xymatrix@R=15pt{
\NEPSh\tau \ar@{.>}[r] \ar[dr]_-{\NEsem\tau} & \NSSemPSh\tau \ar@{.>}[r]
\ar[d]_-{\NSsem\tau} & \ar[dl]^-{\NFsem\tau} \NFPSh\tau \\
& \Semcat(\semfun\sembracket\simplefstarg,\semfun\sembracket\tau) 
}\end{array}
\end{equation}
commutes; so that, for all terms 
$\Gamma\entails t:\tau$~($\Gamma=\pair{x_i:\tau_i}_{i=1,n}$), the diagram
below
$$\xymatrix@R=20pt@C=6pt{
\prod_{i=1,n} \NEPSh{\tau_i} \ar[rd]_(.35){\prod_{i=1,n}\NEsem{\tau_i}} 
\ar@{.>}[r] &
\ar[d]^-{\prod_{i=1,n}\NSsem{\tau_i}} \prod_{i=1,n} \NSSemPSh{\tau_i}
\ar[r]^-{\mbox{\scriptsize$\semfun'\sembracket{\Gamma\entails t:\tau}$}} & 
\NSSemPSh\tau \ar@{.>}[r] \ar[dd]_-{\NSsem\tau} & \NFPSh\tau 
\ar[ddl]^-{\NFsem\tau}
\\
& \prod_{i=1,n}
\Semcat(\semfun\sembracket\simplefstarg,\semfun\sembracket{\tau_i})
\ar[d]|-\iso
\\
& 
\Semcat(\semfun\sembracket\simplefstarg,\semfun\sembracket{\Gamma}) 
\ar[r]_-{\raisebox{-4mm}{\scriptsize
         $\semfun\sembracket{\Gamma\entails t:\tau}\comp\simplefstarg$}}
& \Semcat(\semfun\sembracket\simplefstarg,\semfun\sembracket{\tau}) 
}$$
will commute (\cf~diagram~(\ref{Basic_Lemma_Diagram}) of the Basic Lemma
(Lemma~\ref{Basic_Lemma}
)) 
and, hence, the evaluation of the horizontal top composite at the tuple
$\pair{\varmap_{\tau_i}(x_i)}_{i=1,n}$ of the variables in the
context~$\Gamma$ will yield a 
normal term in $\NFPSh\tau(\Gamma)$ with the
same semantics as the given term~$t$~(compare the Extensional Normalisation
Lemma (Lemma~\ref{Extensional_Normalisation_Lemma}
)
and see Corollary~\ref{Corollary_Correctness_Two} below).  
Moreover, as we will show below~(see
Corollary~\ref{Corollary_Correctness_Three}), the long $\beta\eta$-normal
forms associated to two $\beta\eta$-equal terms will be 
the same. 

\medskip
The abstract way to define the maps in~(\ref{uq_Maps}) ---which in the
literature on normalisation by evaluation are either referred to as
\Rem{unquote} and \Rem{quote} or as \Rem{reflect} and \Rem{reify}--- is by
defining maps 
$$
\NEobj\tau 
\longstackrightarrow{\unquotemap\tau}
\ext\semfun\sembracket\tau
\longstackrightarrow{\quotemap\tau}
\NFobj\tau 
\quad\mbox{in $\Gluecat$}
$$
that project in $\Semcat$ onto identities (see
Proposition~\ref{pi(u)=id=pi(q)} below).  
The definition of these maps is by induction on the structure of types relying
on Propositions~\ref{mu_Operations} and~\ref{nu_Operations} as follows:
\begin{enumerate}
\item
For a base type $\abasetype \in \BaseTypeSet$, we define
$\unquotemap\abasetype = \id_{\NEobj\abasetype}$ and 
\mbox{$\quotemap\abasetype 
         = (\NEobj\abasetype\stackrightarrow\iso\NFobj\abasetype)$}.

\item
We let 
$\unquotemap\TerminalType = (\NEobj\TerminalType \rightarrow \TerminalObj)$
and 
$\quotemap\TerminalType 
 = (\TerminalObj 
   \xymatrix@C=33pt{\ar[r]^-{(\unitmap_\TerminalType,\id)}_-\iso&}
   \NFobj\TerminalType)$.

\item
For types $\tau,\tau'\in\TypeClosure\BaseTypeSet$, we define 
$$
\unquotemap{\tau\ProductType\tau'}:
  \NEobj{\tau\ProductType\tau'} 
  \rightarrow
  \ext\semfun\sembracket{\tau}\CatProduct\ext\semfun\sembracket{\tau'}
$$
as the pairing of the maps
$$\begin{array}{ccc}
\NEobj{\tau\ProductType\tau'}
\xymatrix@C=50pt{\ar[r]^-{(\fstmap^{(\tau')}_{\tau},\proj_1)}&}
\NEobj{\tau}
\xymatrix{\ar[r]^-{\unquotemap{\tau}}&}
\ext\semfun\sembracket{\tau}
& \mbox{ and } & 
\NEobj{\tau\ProductType\tau'}
\xymatrix@C=50pt{\ar[r]^-{(\sndmap^{(\tau)}_{\tau'},\proj_2)}&}
\NEobj{\tau'}
\xymatrix{\ar[r]^-{\unquotemap{\tau'}}&}
\ext\semfun\sembracket{\tau'}
\quad,
\end{array}$$
and let 
$\quotemap{\tau\ProductType\tau'}:
 \ext\semfun\sembracket{\tau}\CatProduct\ext\semfun\sembracket{\tau'}
 \rightarrow
 \NFobj{\tau\ProductType\tau'}$
be the composite
$$
\ext\semfun\sembracket{\tau}\CatProduct\ext\semfun\sembracket{\tau'}
\xymatrix@C=30pt{\ar[r]^-{\quotemap{\tau}\CatProduct\quotemap{\tau'}}&}
\NFobj{\tau}\CatProduct\NFobj{\tau'}
\xymatrix@C=40pt{\ar[r]^-{(\pairmap_{\tau\ProductType\tau'},\id)}_-\iso&}
\NFobj{\tau\ProductType\tau'}\ .
$$ 

\item
For types $\tau,\tau'\in\TypeClosure\BaseTypeSet$, we define 
$$
\unquotemap{\tau\ArrowType\tau'}:
 \NEobj{\tau\ArrowType\tau'}
 \rightarrow
 \ext\semfun\sembracket{\tau'}^{\ext\semfun\sembracket\tau}
$$
as the exponential transpose of the map
$$
\NEobj{\tau\ArrowType\tau'}\CatProduct\ext\semfun\sembracket\tau
\longstackrightarrow{\id\times\quotemap\tau}
\NEobj{\tau\ArrowType\tau'}\CatProduct\NFobj\tau
\longlongstackrightarrow{(\appmap^{(\tau)}_{\tau'},\eval)}
\NEobj{\tau'}
\longstackrightarrow{\unquotemap{\tau'}}
\ext\semfun\sembracket{\tau'}
\ , 
$$
and let 
$\quotemap{\tau\ArrowType\tau'}:
 \ext\semfun\sembracket{\tau'}^{\ext\semfun\sembracket\tau}
 \rightarrow
 \NFobj{\tau\ArrowType\tau'}$
be the composite
$$
\ext\semfun\sembracket{\tau'}^{\ext\semfun\sembracket\tau}
\xymatrix@C=35pt{\ar[r]^-{{\quotemap{\tau'}}^{\unquotemap\tau\vmap_\tau}}&}
{\NFobj{\tau'}}^{\Varobj\tau}
\xymatrix@C=55pt{\ar[r]^-{(\absmap_{\tau\ArrowType\tau'},\id)}_\iso&}
\NFobj{\tau\ArrowType\tau'}
$$
where $\vmap_\tau=(\varmap_\tau,\id):\Varobj\tau\rightarrow\NEobj\tau$.
\end{enumerate}

Proposition~\ref{pi(u)=id=pi(q)} below yields~(\ref{uq_Maps}) as a corollary.
\begin{proposition}\label{pi(u)=id=pi(q)}
For every type $\tau\in\TypeClosure\BaseTypeSet$, we have the identities
$$
\Sproj(\unquotemap\tau) 
 \ = \ \id_{\mbox{\scriptsize$\semfun\sembracket\tau$}}
 \ = \ \Sproj(\quotemap\tau)
$$ 
for $\Sproj$ the forgetful functor $\Gluecat \rightarrow \Semcat$.
\end{proposition}
\begin{proof*}
The proof is by induction on the structure of types.
\begin{enumerate}
\item
For a base type $\abasetype \in \BaseTypeSet$, 
$\Sproj(\unquotemap\abasetype) 
   = \Sproj(\quotemap\abasetype) 
   = \id_{\mbox{\scriptsize$\semfun$}\sembracket\abasetype}$
by definition of $\unquotemap\abasetype$ and $\quotemap\abasetype$.

\item
$\Sproj(\unquotemap\TerminalType) 
   = \Sproj(\quotemap\TerminalType) 
   = \id_{\TerminalObj}$
by definition of $\unquotemap\TerminalType$ and $\quotemap\TerminalType$.

\item
For types $\tau,\tau'\in\TypeClosure\BaseTypeSet$, 
$$\begin{array}{rcll}
\Sproj(\unquotemap{\tau\ProductType\tau'})
& = & \pair{ \; \Sproj(\unquotemap\tau) \comp \proj_1 
             \, , 
             \, \Sproj(\unquotemap{\tau'}) \comp \proj_2
             \; }
    & \mbox{, by definition of $\unquotemap{\tau\ProductType\tau'}$}
\\[1.5mm]
& = & \pair{ \, \proj_1 \, , \, \proj_2 \, } 
    & \mbox{, by induction}
\\[1.5mm]
& = &
\id_{\mbox{\scriptsize$\semfun$}\sembracket{\tau}
     \CatProduct
     \mbox{\scriptsize$\semfun$}\sembracket{\tau'}}
\end{array}$$
and
$$\begin{array}{rcll}
\Sproj(\quotemap{\tau\ProductType\tau'})
& = & \Sproj(\quotemap\tau) \CatProduct \Sproj(\quotemap{\tau'})
    & \mbox{, by definition of $\quotemap{\tau\ProductType\tau'}$}
\\[1.5mm]
& = & \id_{\mbox{\scriptsize$\semfun$}\sembracket\tau} 
           \CatProduct 
           \id_{\mbox{\scriptsize$\semfun$}\sembracket{\tau'}}
    & \mbox{, by induction}
\\[1.5mm]
& = & \id_{\mbox{\scriptsize$\semfun$}\sembracket\tau 
           \CatProduct 
           \mbox{\scriptsize$\semfun$}\sembracket{\tau'}} 
\quad.
\end{array}$$

\item
For types $\tau,\tau'\in\TypeClosure\BaseTypeSet$, 
$$\begin{array}{rcll}
& & 
\!\!\!\!\!\!\!\!\!\!\!\!\!\!\!\!\!\!
\eval \comp 
  (\Sproj(\unquotemap{\tau\ArrowType\tau'}) 
     \CatProduct \id_{\mbox{\scriptsize$\semfun$}\sembracket\tau})
\\[1.5mm]
& = & \Sproj(\unquotemap{\tau'}) \comp \eval \comp
        (\id_{\mbox{\scriptsize$\semfun$}\sembracket\tau
              \CatArrow\mbox{\scriptsize$\semfun$}\sembracket{\tau'}}
              \CatProduct \Sproj(\quotemap\tau))
& \mbox{, by definition of $\unquotemap{\tau\ArrowType\tau'}$}
\\[1.5mm]
& = & \eval
    & \mbox{, by induction}
\end{array}$$
and hence
$$\begin{array}{rcl}
\Sproj(\unquotemap{\tau\ArrowType\tau'})
& = & \id_{\mbox{\scriptsize$\semfun$}\sembracket{\tau}
           \CatArrow
           \mbox{\scriptsize$\semfun$}\sembracket{\tau'}}
\quad ;
\end{array}$$
further
\begin{center}
\hfill
$\begin{array}[b]{rcll}
\Sproj(\quotemap{\tau\ArrowType\tau'})
& = & (\Sproj(\unquotemap\tau) \comp \Sproj(\vmap_\tau)) 
        \CatArrow (\Sproj(\quotemap{\tau'}))
    & \mbox{, by definition of $\quotemap{\tau\ArrowType\tau'}$}
\\[1.5mm]
& = & \id_{\mbox{\scriptsize$\semfun$}\sembracket\tau} 
      \CatArrow 
      \id_{\mbox{\scriptsize$\semfun$}\sembracket{\tau'}}
    & \mbox{, by induction and definition of $\vmap_\tau$}
\\[1.5mm]
& = & \id_{\mbox{\scriptsize$\semfun$}\sembracket\tau
           \;\CatArrow
           \mbox{\scriptsize$\semfun$}\sembracket{\tau'}}
\quad.
\end{array}$\qed
\end{center}
\end{enumerate}
\end{proof*}

\paragraph{Normalisation function.}
Every interpretation $\semfun: \BaseTypeSet \rightarrow \Semcat$ of base types
in a cartesian closed category, induces a \Rem{normalisation function}
$\nf\tau\semfun: \TermPSh\tau \rightarrow \NFPSh\tau$ in
$\Set^{\F{}\comma\TypeClosure\BaseTypeSet}$
defined as the composite
$$ 
\TermPSh\tau 
\stackrightarrow{\extTermsem\tau}
\Hom{\ext\semfun\sembracket\fstarg,\ext\semfun\sembracket\tau}
\longstackrightarrow{\arHom{\unquotemap{}\vmap{},\quotemap\tau}}
\Hom{\commayon(\fstarg),\NFobj\tau}
\stackrightarrow\iso
\NFPSh\tau
$$
where $\extTermsem{}$ denotes the semantics of terms induced by the
interpretation $\ext\semfun:\BaseTypeSet\rightarrow\Gluecat$ 
of~(\ref{mu_Interpretation}) and where
$$\begin{array}{lcl}
(\unquotemap{}\vmap)_\Gamma
& = & 
\commayon(\Gamma) 
  \stackrightarrow{\vmap_\Gamma}
\NEobj{}\sembracket\Gamma
  \stackrightarrow{\unquotemap\Gamma}
\ext\semfun\sembracket\Gamma
\end{array}$$
for
$$\begin{array}{lcl}
\NEobj{}\sembracket\Gamma 
& = & 
\prod_{(x:\tau)\in\Gamma}\NEobj\tau
\quad ,
\\[1mm]
\vmap_\Gamma 
& = & 
\commayon(\Gamma) 
  \stackrightarrow\iso
\prod_{(x:\tau)\in\Gamma}\Varobj\tau
  \xymatrix@C=50pt{\ar[r]^-{\prod_{(x:\tau)\in\Gamma} \vmap_\tau}&}
\NEobj{}\sembracket\Gamma
\quad ,
\\[2mm]
\unquotemap\Gamma 
& = & 
\prod_{(x:\tau)\in\Gamma} \unquotemap\tau
\quad .
\end{array}$$
Explicitly, 
$$
\nf{\tau,\Gamma}\semfun(t)
\ = \
( \quotemap\tau\ 
    \ext\semfun\sembracket{\Gamma\entails t:\tau}\
      (\unquotemap{}\vmap)_\Gamma
  )
    (\id_\Gamma)
\enspace\in \NFPSh\tau(\Gamma) 
$$
for all terms $t\in\TermPSh\tau(\Gamma)$. 

\medskip
Having the same denotation, $\beta\eta$-equal terms are identified
by the normalisation function.
\begin{corollary}\label{Corollary_Correctness_Three}
Let $\semfun:\BaseTypeSet\rightarrow\Semcat$ be an interpretation of base
types in a cartesian closed category.  
For every pair of terms $t,t'$ in $\TermPSh\tau(\Gamma)$, if $t \betaetaeq t'$
then $\nf{\tau,\Gamma}\semfun(t) = \nf{\tau,\Gamma}\semfun(t')$ in
$\NFPSh\tau(\Gamma)$.
\end{corollary}
Further, as a consequence of Proposition~\ref{pi(u)=id=pi(q)}
(see also~(\ref{uq_Maps})), we have that a term and its associated normal form
have the same semantics.
\begin{corollary}\label{Corollary_Correctness_Two}
For every interpretation 
$\semfun:\BaseTypeSet\rightarrow\Semcat$
of base types in a cartesian closed category, the diagram
$$\xymatrix{
\TermPSh\tau \ar[dr]_-{\Termsem\tau} 
\ar[rr]^-{\mbox{\scriptsize$\nf\tau\semfun$}} && 
\NFPSh\tau \ar[ld]^-{\NFsem\tau}\\
& \Semcat(\semfun\sembracket\simplefstarg,\semfun\sembracket\tau)
}$$
commutes for all types $\tau\in\TypeClosure\BaseTypeSet$.   
\end{corollary}

Considering the universal interpretation
$\freesemfun: \BaseTypeSet \rightarrow \FreeCCC[\BaseTypeSet]$ of the set of
base types $\BaseTypeSet$ into the free cartesian closed category
$\FreeCCC[\BaseTypeSet]$ over them, by
Corollary~\ref{Corollary_Correctness_Two}, we have that   
\begin{equation} \label{t_betaetaequal_f-nf(t)}
t \betaetaeq \nf{\tau,\Gamma}\freesemfun(t) 
\end{equation}
and hence, by Corollary~\ref{Corollary_Correctness_Three}, that 
\[
\nf{\tau,\Gamma}\semfun(t) 
= 
\nf{\tau,\Gamma}\semfun(\nf{\tau,\Gamma}\freesemfun(t))
\]
for all terms $t\in\TermPSh\tau(\Gamma)$.  
Thus, the normalisation function $\nf\tau\freesemfun$ is idempotent and
therefore fixes some normal terms.  
In fact, as we will see below (see~(\ref{NF_eqn}) in Theorem~\ref{THEOREM}),
all normalisation functions $\nf\tau\semfun$ fix all normal terms:  that is,
\begin{equation}\label{nf_Fixes_Normal_Forms}
\mbox{for all $N \in \NFPSh\tau(\Gamma)$, 
  $\nf{\tau,\Gamma}\semfun(N) = N$\quad.}
\end{equation}
This fixed-point property is important: from it and
Corollary~\ref{Corollary_Correctness_Three} it follows that 
\begin{itemize}
\item
  for all terms $t\in\TermPSh\tau(\Gamma)$ and normal terms
  \mbox{$N\in\NFPSh\tau(\Gamma)$}, if $t \betaetaeq N$ then
  $\nf{\tau,\Gamma}\semfun(t) = N$, and 

\item 
  for every pair of normal terms $N,N' \in \NFPSh\tau(\Gamma)$, if 
  $N \betaetaeq N'$ then $N = N'$;
\end{itemize}
so that, further using Corollary~\ref{Corollary_Correctness_Two} in the
form~(\ref{t_betaetaequal_f-nf(t)}), we have that 
\begin{itemize}
\item
  for all terms \mbox{$t\in\TermPSh\tau(\Gamma)$}, 
  $\nf{\tau,\Gamma}\semfun(t) = \nf{\tau,\Gamma}\freesemfun(t)$.
\end{itemize}
Thus, the fixed-point property~(\ref{nf_Fixes_Normal_Forms}) allows one to
conclude that: 
\begin{quote}
  all interpretations induce the same normalisation function
  $\simplenf\tau: \TermPSh\tau \rightarrow \NFPSh\tau$ 
such that, for every term
$t\in\TermPSh\tau(\Gamma)$, one has that
$\simplenf{\tau,\Gamma}(t)\in\NFPSh\tau(\Gamma)$ is the unique normal term
$\beta\eta$-equal to $t$.
\end{quote}

\medskip
We now 
establish~(\ref{nf_Fixes_Normal_Forms}).  
The appropriate induction hypothesis to proceed by induction on the structure
of neutral and normal terms is stated in the theorem below.
\begin{theorem} \label{THEOREM}
For every interpretation 
$\semfun:\BaseTypeSet\rightarrow\Semcat$
of base types in a cartesian closed category, the diagrams
\begin{equation}\label{NE_eqn}
\begin{array}{c}\xymatrix@C=-5pt{
\NEPSh\tau \ar[drr]_-{\extNEsem\tau} 
& \iso \Hom{\commayon(\simplefstarg),\NEobj\tau}
\ar[rr]^-{\arHom{\id,\unquotemap\tau}} 
&& 
\Hom{\commayon(\simplefstarg),\ext\semfun\sembracket\tau}
\\
& & \Hom{\ext\semfun\sembracket{\simplefstarg},\ext\semfun\sembracket\tau}
\ar[ru]_-{\arHom{\unquotemap{}\vmap{},\id}} & 
}\end{array}
\end{equation}
and
\begin{equation}\label{NF_eqn}
\begin{array}{c}\xymatrix{
\NFPSh\tau \ar[rr]|-\iso \ar[dr]_{\extNFsem\tau} & & 
\Hom{\commayon(\simplefstarg),\NFobj\tau} \\
& \Hom{\ext\semfun\sembracket\simplefstarg,\ext\semfun\sembracket\tau} 
\ar[ru]_-{\arHom{\unquotemap{}\vmap{},\quotemap\tau}}
}\end{array}
\end{equation}
commute for all types $\tau \in \TypeClosure\BaseTypeSet$.
\end{theorem}
\begin{proof}
The proof uses the induction principle associated to the initial 
$\pair{\SSSNE,\SSSNF}$-algebra $(\NEPSh{},\NFPSh{})$ 
(see~(\ref{Induction_Principle})) by considering the
equalisers 
\begin{center}
$\SubNEPSh\tau \stackrightmono{\SubNEmap\tau} \NEPSh\tau$
\ and \
$\SubNFPSh\tau \stackrightmono{\SubNFmap\tau} \NFPSh\tau$
\end{center}
of~(\ref{NE_eqn}) and~(\ref{NF_eqn}) respectively, and showing that the
family 
$$
(\SubNEmap\tau,\SubNFmap\tau):
  (\SubNEPSh\tau,\SubNFPSh\tau) \rightmono (\NEPSh\tau,\NFPSh\tau)
\qquad(\tau\in\TypeClosure\BaseTypeSet)
$$ 
is a sub~$\pair{\SSSNE,\SSSNF}$-algebra, from which it follows that
$\SubNEmap\tau$ and $\SubNFmap\tau$ are isomorphisms and hence
that~(\ref{NE_eqn}) and~(\ref{NF_eqn}) commute.  The details are spelled
out in Appendix~\ref{PROOF}.  
\end{proof}

\begin{remark}
In elementary terms, the above categorical proof amounts to establishing the
identities  
$$
\ext\semfun\sembracket{\Gamma\entails M:\tau}\ 
  (\unquotemap{}\vmap)_\Gamma
\ = \
\unquotemap\tau\ (M[\fstarg],\semfun\sembracket{\Gamma\entails M:\tau})
$$
and 
$$
\quotemap\tau\
  \ext\semfun\sembracket{\Gamma\entails N:\tau}\
    (\unquotemap{}\vmap)_\Gamma
\ = \
(N[\fstarg],\semfun\sembracket{\Gamma\entails N:\tau})
$$
for $M\in\NEPSh\tau(\Gamma)$ and $N\in\NFPSh\tau(\Gamma)$, by simultaneous 
induction on the derivation of neutral and normal
terms~(\cf~\cite{Reynolds}).  
\end{remark}

\medskip
The commutativity of diagram~(\ref{NF_eqn}) amounts to
property~(\ref{nf_Fixes_Normal_Forms}) and hence, as explained above, all
normalisation functions coincide. 
\begin{corollary} \label{Corollary_nfCoincidence}
For every interpretation 
$\semfun:\BaseTypeSet\rightarrow\Semcat$
of base types in a cartesian closed category and for the universal
interpretation 
$\freesemfun: \BaseTypeSet \rightarrow \FreeCCC[\BaseTypeSet]$ of base
types into the free cartesian closed category over them, the identity
$$\begin{array}{rclll}
\nf\tau\semfun & = & \nf\tau\freesemfun 
& : \TermPSh\tau \rightarrow \NFPSh\tau
&\mbox{ in $\Set^{\F{}\comma\TypeClosure\BaseTypeSet}$}
\end{array}$$
holds.
\end{corollary}

Summarising, we have obtained normalisation functions
$$
\simplenf{\tau,\Gamma}: 
  \TermPSh\tau(\Gamma) \rightarrow \NFPSh\tau(\Gamma)
\quad(\tau\in\TypeClosure\BaseTypeSet, 
      \Gamma\in\obj{\F{}\comma\TypeClosure\BaseTypeSet}) 
$$ satisfying the correctness properties below.
\begin{itemize}
\item
For all context renamings~$\rho: \Gamma \rightarrow \Gamma'$
in $\F{}\comma\TypeClosure\BaseTypeSet$, 
$$
(\simplenf{\tau,\Gamma}\ t)[\rho] = \simplenf{\tau,\Gamma'}(t[\rho])
$$
for every term $t \in \TermPSh\tau(\Gamma)$.

\item
For all normal terms $N\in\NFPSh\tau(\Gamma)$, 
$$\simplenf{\tau,\Gamma}(N) = N\quad.$$

\item
For all terms $t\in\TermPSh\tau(\Gamma)$, 
$$\simplenf{\tau,\Gamma}(t)\betaetaeq t\quad.$$

\item
For all terms $t,t'\in\TermPSh\tau(\Gamma)$, 
\begin{center}
if $t\betaetaeq t'$ then 
$\simplenf{\tau,\Gamma}(t) = \simplenf{\tau,\Gamma}(t')$\quad.
\end{center}
\end{itemize}

\paragraph{Normalisation algorithm.}
The simplest description of the normalisation function from which to extract
an algorithm is the one induced by the trivial interpretation 
$\trivsemfun$ 
of base types in the trivial cartesian closed category as, in this case, the
glueing category
$\Set^{\F{}
 \comma\TypeClosure\BaseTypeSet}\comma\clone{\trivsemfun\sembracket\fstarg}$ 
is simply (isomorphic to) the presheaf category
$\Set^{\F{}\comma\TypeClosure\BaseTypeSet}$ (recall
Example~\ref{Trivial_Glueing}).  In fact, previous categorical analysis of
normalisation by evaluation have centred around this
interpretation~\cite{AHS,Reynolds}.

Explicitly, the unquote and quote maps
\begin{equation} \label{UnquoteQuoteMaps}
\NEPSh\tau \longstackrightarrow{\unquotemap\tau} \semfun\sembracket\tau
\longstackrightarrow{\quotemap\tau} 
\NFPSh\tau 
\quad(\tau\in\TypeClosure\BaseTypeSet)
\end{equation}
in $\Set^{\F{}\comma\TypeClosure\BaseTypeSet}$, with respect to the 
interpretation of base types 
$\semfun: 
     \abasetype \mapsto \NEPSh\abasetype$,  
are (in the internal language of
$\Set^{\F{}\comma\TypeClosure\BaseTypeSet}$) as follows:
\begin{description}
\item[]1.%
\begin{tabular}[t]{l}
$\unquotemap\abasetype(M) = M$
\\[2mm]
$\quotemap\abasetype(M) = \norm(M)$, where 
$\norm:
\NEPSh\abasetype\stackrightarrow\iso\NFPSh\abasetype$~(see~(\ref{normISO}))
\end{tabular}

\item[]2.%
\begin{tabular}[t]{l}
$\unquotemap\TerminalType(M) = (\,)$
\\[2mm]
$\quotemap\TerminalType(\,) = \unitmap_\TerminalType(\,)$
\end{tabular}

\item[]3.%
\begin{tabular}[t]{l}
$\unquotemap{\tau\ProductType\tau'}(M)
 = (\ \unquotemap\tau(\fstmap^{(\tau')}_\tau\,M)\ ,\ 
     \unquotemap{\tau'}(\sndmap^{(\tau)}_{\tau'}\,M)\ )$
\\[2mm]
$\quotemap{\tau\ProductType\tau'}(x,x')
 = \pairmap_{\tau\ProductType\tau'}
     (\ \quotemap\tau(x)\ ,\
        \quotemap{\tau'}(x')\ )$
\end{tabular}

\item[]4.%
\begin{tabular}[t]{l}
$\unquotemap{\tau\ArrowType\tau'}(M)
 = \semlambdaop 
     {x^{\mbox{\scriptsize$\semfun\sembracket\tau$}}}
     {\unquotemap{\tau'}(\appmap^{(\tau)}_{\tau'}(M,\quotemap\tau\,x))}$
\\[2mm]
$\quotemap{\tau\ArrowType\tau'}(f)
 = \absmap_{\tau\ArrowType\tau'}
     ( \semlambdaop
         {v^{\VarPSh\tau}}
         {\quotemap{\tau'}(f(\unquotemap\tau(\varmap_\tau\,v)))}
       )$
\end{tabular}
\end{description}
and the normalisation function 
is given by
\begin{equation} \label{NormalisationFunction}
\simplenf{\tau,\Gamma}(t)
=
\quotemap\tau
  ( \semfun\sembracket{\Gamma\entails t:\tau}\,
      \pair{\unquotemap{\tau_i}(\varmap_{\tau_i}\,x_i)}_{i=1,n}
    )
\end{equation}
for all terms $t\in\TermPSh\tau(\Gamma)$ where 
$\Gamma=\pair{x_i:\tau_i}_{i=1,n}$. 

These functions can be directly implemented, for instance, in metalanguages
supporting abstract syntax with variable binding, like HOAS~\cite{HOAS},
\texttt{Fresh O'Caml}~\cite{FreshOCAML}, the Scope-and-Type Safe Universe of
Syntaxes with Binding~\cite{AllaisEtAl}, and the SOAS
Framework~\cite{FioreSzamozvancev}.  
Indeed, for concreteness, we here synthesise an elementary implementation in
Agda considered as a dependently-typed functional programming
language\footnote{The code is available from
  \url{www.cl.cam.ac.uk/~mpf23/Notes/Notes.html}.}.

\medskip

\begin{code}[hide]%
\>[0]\<%
\\
\>[0]\AgdaKeyword{open}\AgdaSpace{}%
\AgdaKeyword{import}\AgdaSpace{}%
\AgdaModule{Agda.Primitive}\<%
\\
\>[0]\AgdaKeyword{open}\AgdaSpace{}%
\AgdaKeyword{import}\AgdaSpace{}%
\AgdaModule{Agda.Builtin.Unit}\<%
\\
\>[0]\AgdaKeyword{open}\AgdaSpace{}%
\AgdaKeyword{import}\AgdaSpace{}%
\AgdaModule{Agda.Builtin.Nat}\<%
\\
\>[0]\AgdaKeyword{open}\AgdaSpace{}%
\AgdaKeyword{import}\AgdaSpace{}%
\AgdaModule{Agda.Builtin.Sigma}\<%
\\
\>[0]\AgdaKeyword{open}\AgdaSpace{}%
\AgdaKeyword{import}\AgdaSpace{}%
\AgdaModule{Agda.Builtin.Equality}\<%
\\
\\[\AgdaEmptyExtraSkip]%
\>[0]\AgdaComment{-- composition}\<%
\\
\>[0]\AgdaOperator{\AgdaFunction{\AgdaUnderscore{}∘\AgdaUnderscore{}}}\AgdaSpace{}%
\AgdaSymbol{:}%
\>[11I]\AgdaSymbol{\{}\AgdaBound{A}\AgdaSpace{}%
\AgdaSymbol{:}\AgdaSpace{}%
\AgdaPrimitiveType{Set}\AgdaSymbol{\}}\AgdaSpace{}%
\AgdaSymbol{\{}\AgdaBound{B}\AgdaSpace{}%
\AgdaSymbol{:}\AgdaSpace{}%
\AgdaBound{A}\AgdaSpace{}%
\AgdaSymbol{→}\AgdaSpace{}%
\AgdaPrimitiveType{Set}\AgdaSymbol{\}}\AgdaSpace{}%
\AgdaSymbol{\{}\AgdaBound{C}\AgdaSpace{}%
\AgdaSymbol{:}\AgdaSpace{}%
\AgdaSymbol{\{}\AgdaBound{x}\AgdaSpace{}%
\AgdaSymbol{:}\AgdaSpace{}%
\AgdaBound{A}\AgdaSymbol{\}}\AgdaSpace{}%
\AgdaSymbol{→}\AgdaSpace{}%
\AgdaBound{B}\AgdaSpace{}%
\AgdaBound{x}\AgdaSpace{}%
\AgdaSymbol{→}\AgdaSpace{}%
\AgdaPrimitiveType{Set}\AgdaSymbol{\}}\<%
\\
\>[.][@{}l@{}]\<[11I]%
\>[6]\AgdaSymbol{(}\AgdaBound{f}\AgdaSpace{}%
\AgdaSymbol{:}\AgdaSpace{}%
\AgdaSymbol{\{}\AgdaBound{x}\AgdaSpace{}%
\AgdaSymbol{:}\AgdaSpace{}%
\AgdaBound{A}\AgdaSymbol{\}(}\AgdaBound{y}\AgdaSpace{}%
\AgdaSymbol{:}\AgdaSpace{}%
\AgdaBound{B}\AgdaSpace{}%
\AgdaBound{x}\AgdaSymbol{)}\AgdaSpace{}%
\AgdaSymbol{→}\AgdaSpace{}%
\AgdaBound{C}\AgdaSpace{}%
\AgdaBound{y}\AgdaSymbol{)}\AgdaSpace{}%
\AgdaSymbol{(}\AgdaBound{g}\AgdaSpace{}%
\AgdaSymbol{:}\AgdaSpace{}%
\AgdaSymbol{(}\AgdaBound{x}\AgdaSpace{}%
\AgdaSymbol{:}\AgdaSpace{}%
\AgdaBound{A}\AgdaSymbol{)}\AgdaSpace{}%
\AgdaSymbol{→}\AgdaSpace{}%
\AgdaBound{B}\AgdaSpace{}%
\AgdaBound{x}\AgdaSymbol{)}\AgdaSpace{}%
\AgdaSymbol{(}\AgdaBound{x}\AgdaSpace{}%
\AgdaSymbol{:}\AgdaSpace{}%
\AgdaBound{A}\AgdaSymbol{)}\AgdaSpace{}%
\AgdaSymbol{→}\AgdaSpace{}%
\AgdaBound{C}\AgdaSpace{}%
\AgdaSymbol{(}\AgdaBound{g}\AgdaSpace{}%
\AgdaBound{x}\AgdaSymbol{)}\<%
\\
\>[0]\AgdaSymbol{(}\AgdaBound{f}\AgdaSpace{}%
\AgdaOperator{\AgdaFunction{∘}}\AgdaSpace{}%
\AgdaBound{g}\AgdaSymbol{)}\AgdaSpace{}%
\AgdaBound{x}\AgdaSpace{}%
\AgdaSymbol{=}\AgdaSpace{}%
\AgdaBound{f}\AgdaSpace{}%
\AgdaSymbol{(}\AgdaBound{g}\AgdaSpace{}%
\AgdaBound{x}\AgdaSymbol{)}\<%
\\
\\[\AgdaEmptyExtraSkip]%
\>[0]\AgdaComment{-- binary product}\<%
\\
\>[0]\AgdaOperator{\AgdaFunction{\AgdaUnderscore{}×\AgdaUnderscore{}}}\AgdaSpace{}%
\AgdaSymbol{:}\AgdaSpace{}%
\AgdaSymbol{\{}\AgdaBound{m}\AgdaSpace{}%
\AgdaBound{n}\AgdaSpace{}%
\AgdaSymbol{:}\AgdaSpace{}%
\AgdaPostulate{Level}\AgdaSymbol{\}}\AgdaSpace{}%
\AgdaSymbol{→}\AgdaSpace{}%
\AgdaPrimitiveType{Set}\AgdaSpace{}%
\AgdaBound{m}\AgdaSpace{}%
\AgdaSymbol{→}\AgdaSpace{}%
\AgdaPrimitiveType{Set}\AgdaSpace{}%
\AgdaBound{n}\AgdaSpace{}%
\AgdaSymbol{→}\AgdaSpace{}%
\AgdaPrimitiveType{Set}\AgdaSpace{}%
\AgdaSymbol{(}\AgdaBound{m}\AgdaSpace{}%
\AgdaOperator{\AgdaPrimitive{⊔}}\AgdaSpace{}%
\AgdaBound{n}\AgdaSymbol{)}\<%
\\
\>[0]\AgdaBound{A}\AgdaSpace{}%
\AgdaOperator{\AgdaFunction{×}}\AgdaSpace{}%
\AgdaBound{B}\AgdaSpace{}%
\AgdaSymbol{=}\AgdaSpace{}%
\AgdaRecord{Σ}\AgdaSpace{}%
\AgdaBound{A}\AgdaSpace{}%
\AgdaSymbol{(λ}\AgdaSpace{}%
\AgdaBound{\AgdaUnderscore{}}\AgdaSpace{}%
\AgdaSymbol{→}\AgdaSpace{}%
\AgdaBound{B}\AgdaSymbol{)}\<%
\\
\\[\AgdaEmptyExtraSkip]%
\>[0]\AgdaFunction{π₁}\AgdaSpace{}%
\AgdaSymbol{:}\AgdaSpace{}%
\AgdaSymbol{\{}\AgdaBound{A}\AgdaSpace{}%
\AgdaBound{B}\AgdaSpace{}%
\AgdaSymbol{:}\AgdaSpace{}%
\AgdaPrimitiveType{Set}\AgdaSymbol{\}}\AgdaSpace{}%
\AgdaSymbol{→}\AgdaSpace{}%
\AgdaBound{A}\AgdaSpace{}%
\AgdaOperator{\AgdaFunction{×}}\AgdaSpace{}%
\AgdaBound{B}\AgdaSpace{}%
\AgdaSymbol{→}\AgdaSpace{}%
\AgdaBound{A}\<%
\\
\>[0]\AgdaFunction{π₁}\AgdaSpace{}%
\AgdaSymbol{=}\AgdaSpace{}%
\AgdaField{fst}\<%
\\
\\[\AgdaEmptyExtraSkip]%
\>[0]\AgdaFunction{π₂}\AgdaSpace{}%
\AgdaSymbol{:}\AgdaSpace{}%
\AgdaSymbol{\{}\AgdaBound{A}\AgdaSpace{}%
\AgdaBound{B}\AgdaSpace{}%
\AgdaSymbol{:}\AgdaSpace{}%
\AgdaPrimitiveType{Set}\AgdaSymbol{\}}\AgdaSpace{}%
\AgdaSymbol{→}\AgdaSpace{}%
\AgdaBound{A}\AgdaSpace{}%
\AgdaOperator{\AgdaFunction{×}}\AgdaSpace{}%
\AgdaBound{B}\AgdaSpace{}%
\AgdaSymbol{→}\AgdaSpace{}%
\AgdaBound{B}\<%
\\
\>[0]\AgdaFunction{π₂}\AgdaSpace{}%
\AgdaSymbol{=}\AgdaSpace{}%
\AgdaField{snd}\<%
\end{code}

\mysubparagraph{Syntax.}
We consider simple types over a countably infinite set of base
types~(see~(\ref{SimpleTypes})): 
\begin{code}%
\>[0]\AgdaComment{-- base types}\<%
\\
\>[0]\AgdaFunction{T}\AgdaSpace{}%
\AgdaSymbol{:}\AgdaSpace{}%
\AgdaPrimitiveType{Set}\<%
\\
\>[0]\AgdaFunction{T}\AgdaSpace{}%
\AgdaSymbol{=}\AgdaSpace{}%
\AgdaDatatype{Nat}\<%
\\
\\[\AgdaEmptyExtraSkip]%
\>[0]\AgdaComment{-- simple types}\<%
\\
\>[0]\AgdaKeyword{data}\AgdaSpace{}%
\AgdaDatatype{𝓣}\AgdaSpace{}%
\AgdaSymbol{:}\AgdaSpace{}%
\AgdaPrimitiveType{Set}\AgdaSpace{}%
\AgdaKeyword{where}\<%
\\
\>[0][@{}l@{\AgdaIndent{0}}]%
\>[2]\AgdaInductiveConstructor{θ}\AgdaSpace{}%
\AgdaSymbol{:}\AgdaSpace{}%
\AgdaFunction{T}\AgdaSpace{}%
\AgdaSymbol{→}\AgdaSpace{}%
\AgdaDatatype{𝓣}\<%
\\
\>[2]\AgdaInductiveConstructor{υ}\AgdaSpace{}%
\AgdaSymbol{:}\AgdaSpace{}%
\AgdaDatatype{𝓣}\<%
\\
\>[2]\AgdaOperator{\AgdaInductiveConstructor{\AgdaUnderscore{}⋆\AgdaUnderscore{}}}\AgdaSpace{}%
\AgdaSymbol{:}\AgdaSpace{}%
\AgdaDatatype{𝓣}\AgdaSpace{}%
\AgdaSymbol{→}\AgdaSpace{}%
\AgdaDatatype{𝓣}\AgdaSpace{}%
\AgdaSymbol{→}\AgdaSpace{}%
\AgdaDatatype{𝓣}\<%
\\
\>[2]\AgdaOperator{\AgdaInductiveConstructor{\AgdaUnderscore{}⇒\AgdaUnderscore{}}}\AgdaSpace{}%
\AgdaSymbol{:}\AgdaSpace{}%
\AgdaDatatype{𝓣}\AgdaSpace{}%
\AgdaSymbol{→}\AgdaSpace{}%
\AgdaDatatype{𝓣}\AgdaSpace{}%
\AgdaSymbol{→}\AgdaSpace{}%
\AgdaDatatype{𝓣}\<%
\end{code}

Typing contexts~(Definition~\ref{ContextCat}) are inductively generated by
context extension from an empty context:
\begin{code}%
\>[0]\AgdaComment{-- typing contexts}\<%
\\
\>[0]\AgdaKeyword{data}\AgdaSpace{}%
\AgdaDatatype{𝔽↓}\AgdaSpace{}\AgdaSpace{}%
\AgdaSymbol{:}\AgdaSpace{}%
\AgdaPrimitiveType{Set}\AgdaSpace{}%
\AgdaSymbol{→}%
\>[17]\AgdaPrimitiveType{Set}\AgdaSpace{}%
\AgdaKeyword{where}\<%
\\
\>[0][@{}l@{\AgdaIndent{0}}]%
\>[2]\AgdaInductiveConstructor{·}\AgdaSpace{}%
\AgdaSymbol{:}\AgdaSpace{}%
\AgdaSymbol{\{}\AgdaBound{𝓧}\AgdaSpace{}%
\AgdaSymbol{:}\AgdaSpace{}%
\AgdaPrimitiveType{Set}\AgdaSymbol{\}}\AgdaSpace{}%
\AgdaSymbol{→}\AgdaSpace{}%
\AgdaDatatype{𝔽↓}\AgdaSpace{}%
\AgdaBound{𝓧}\<%
\\
\>[2]\AgdaOperator{\AgdaInductiveConstructor{\AgdaUnderscore{}::\AgdaUnderscore{}}}\AgdaSpace{}%
\AgdaSymbol{:}\AgdaSpace{}%
\AgdaSymbol{\{}\AgdaBound{𝓧}\AgdaSpace{}%
\AgdaSymbol{:}\AgdaSpace{}%
\AgdaPrimitiveType{Set}\AgdaSymbol{\}}\AgdaSpace{}%
\AgdaSymbol{→}\AgdaSpace{}%
\AgdaDatatype{𝔽↓}\AgdaSpace{}%
\AgdaBound{𝓧}\AgdaSpace{}%
\AgdaSymbol{→}\AgdaSpace{}%
\AgdaBound{𝓧}\AgdaSpace{}%
\AgdaSymbol{→}\AgdaSpace{}%
\AgdaDatatype{𝔽↓}\AgdaSpace{}%
\AgdaBound{𝓧}\<%
\end{code}
We then have a family of variable indices~(recall~(\ref{VPSh})) given as
follows:
\begin{code}%
\>[0]\AgdaComment{-- variable indices}\<%
\\
\>[0]\AgdaKeyword{data}\AgdaSpace{}%
\AgdaDatatype{V}\AgdaSpace{}%
\AgdaSymbol{:}\AgdaSpace{}%
\AgdaSymbol{\{}\AgdaBound{𝓧}\AgdaSpace{}%
\AgdaSymbol{:}\AgdaSpace{}%
\AgdaPrimitiveType{Set}\AgdaSymbol{\}}\AgdaSpace{}%
\AgdaSymbol{→}\AgdaSpace{}%
\AgdaBound{𝓧}\AgdaSpace{}%
\AgdaSymbol{→}\AgdaSpace{}%
\AgdaDatatype{𝔽↓}\AgdaSpace{}%
\AgdaBound{𝓧}\AgdaSpace{}%
\AgdaSymbol{→}\AgdaSpace{}%
\AgdaPrimitiveType{Set}\AgdaSpace{}%
\AgdaKeyword{where}\<%
\\
\>[0][@{}l@{\AgdaIndent{0}}]%
\>[2]\AgdaInductiveConstructor{•}\AgdaSpace{}%
\AgdaSymbol{:}\AgdaSpace{}%
\AgdaSymbol{\{}\AgdaBound{𝓧}\AgdaSpace{}%
\AgdaSymbol{:}\AgdaSpace{}%
\AgdaPrimitiveType{Set}\AgdaSymbol{\}}\AgdaSpace{}%
\AgdaSymbol{\{}\AgdaBound{α}\AgdaSpace{}%
\AgdaSymbol{:}\AgdaSpace{}%
\AgdaBound{𝓧}\AgdaSymbol{\}}\AgdaSpace{}%
\AgdaSymbol{\{}\AgdaBound{Γ}\AgdaSpace{}%
\AgdaSymbol{:}\AgdaSpace{}%
\AgdaDatatype{𝔽↓}\AgdaSpace{}%
\AgdaBound{𝓧}\AgdaSymbol{\}}\AgdaSpace{}%
\AgdaSymbol{→}\AgdaSpace{}%
\AgdaDatatype{V}\AgdaSpace{}%
\AgdaBound{α}\AgdaSpace{}%
\AgdaSymbol{(}\AgdaBound{Γ}\AgdaSpace{}%
\AgdaOperator{\AgdaInductiveConstructor{::}}\AgdaSpace{}%
\AgdaBound{α}\AgdaSymbol{)}\<%
\\
\>[2]\AgdaInductiveConstructor{↑}\AgdaSpace{}%
\AgdaSymbol{:}\AgdaSpace{}%
\AgdaSymbol{\{}\AgdaBound{𝓧}\AgdaSpace{}%
\AgdaSymbol{:}\AgdaSpace{}%
\AgdaPrimitiveType{Set}\AgdaSymbol{\}}\AgdaSpace{}%
\AgdaSymbol{\{}\AgdaBound{α}\AgdaSpace{}%
\AgdaBound{β}\AgdaSpace{}%
\AgdaSymbol{:}\AgdaSpace{}%
\AgdaBound{𝓧}\AgdaSymbol{\}}\AgdaSpace{}%
\AgdaSymbol{\{}\AgdaBound{Γ}\AgdaSpace{}%
\AgdaSymbol{:}\AgdaSpace{}%
\AgdaDatatype{𝔽↓}\AgdaSpace{}%
\AgdaBound{𝓧}\AgdaSymbol{\}}\AgdaSpace{}%
\AgdaSymbol{→}\AgdaSpace{}%
\AgdaDatatype{V}\AgdaSpace{}%
\AgdaBound{β}\AgdaSpace{}%
\AgdaBound{Γ}\AgdaSpace{}%
\AgdaSymbol{→}\AgdaSpace{}%
\AgdaDatatype{V}\AgdaSpace{}%
\AgdaBound{β}\AgdaSpace{}%
\AgdaSymbol{(}\AgdaBound{Γ}\AgdaSpace{}%
\AgdaOperator{\AgdaInductiveConstructor{::}}\AgdaSpace{}%
\AgdaBound{α}\AgdaSymbol{)}\<%
\end{code}
for which context renamings~(Definition~\ref{ContextCat}) are considered:
\begin{code}%
\>[0]\AgdaComment{-- context renamings}\<%
\\
\>[0]\AgdaFunction{𝔽↓₀}\AgdaSpace{}\AgdaSpace{}%
\AgdaSymbol{:}\AgdaSpace{}%
\AgdaSymbol{(}\AgdaBound{𝓣}\AgdaSpace{}%
\AgdaSymbol{:}\AgdaSpace{}%
\AgdaPrimitiveType{Set}\AgdaSymbol{)}\AgdaSpace{}%
\AgdaSymbol{→}\AgdaSpace{}%
\AgdaDatatype{𝔽↓}\AgdaSpace{}%
\AgdaBound{𝓣}\AgdaSpace{}%
\AgdaOperator{\AgdaFunction{×}}\AgdaSpace{}%
\AgdaDatatype{𝔽↓}\AgdaSpace{}%
\AgdaBound{𝓣}\AgdaSpace{}%
\AgdaSymbol{→}\AgdaSpace{}%
\AgdaPrimitiveType{Set}\<%
\\
\>[0]\AgdaFunction{𝔽↓₀}\AgdaSpace{}%
\AgdaBound{𝓣}\AgdaSpace{}%
\AgdaSymbol{(}\AgdaBound{Δ}\AgdaSpace{}%
\AgdaOperator{\AgdaInductiveConstructor{,}}\AgdaSpace{}%
\AgdaBound{Γ}\AgdaSpace{}%
\AgdaSymbol{)}\AgdaSpace{}%
\AgdaSymbol{=}\AgdaSpace{}%
\AgdaSymbol{\{}\AgdaBound{α}\AgdaSpace{}%
\AgdaSymbol{:}\AgdaSpace{}%
\AgdaBound{𝓣}\AgdaSymbol{\}}\AgdaSpace{}%
\AgdaSymbol{→}\AgdaSpace{}%
\AgdaDatatype{V}\AgdaSpace{}%
\AgdaBound{α}\AgdaSpace{}%
\AgdaBound{Δ}\AgdaSpace{}%
\AgdaSymbol{→}\AgdaSpace{}%
\AgdaDatatype{V}\AgdaSpace{}%
\AgdaBound{α}\AgdaSpace{}%
\AgdaBound{Γ}\<%
\end{code}

The abstract syntax of simply typed terms~(see~(\ref{LamInitialAlg})) is
implemented by the inductive family below:
\begin{code}%
\>[0]\AgdaComment{-- simply typed terms}\<%
\\
\>[0]\AgdaKeyword{data}\AgdaSpace{}%
\AgdaDatatype{𝓛}\AgdaSpace{}%
\AgdaSymbol{:}\AgdaSpace{}%
\AgdaDatatype{𝓣}\AgdaSpace{}%
\AgdaSymbol{→}\AgdaSpace{}%
\AgdaDatatype{𝔽↓}\AgdaSpace{}%
\AgdaDatatype{𝓣}\AgdaSpace{}%
\AgdaSymbol{→}\AgdaSpace{}%
\AgdaPrimitiveType{Set}\AgdaSpace{}%
\AgdaKeyword{where}\<%
\\
\>[0][@{}l@{\AgdaIndent{0}}]%
\>[2]\AgdaInductiveConstructor{var}\AgdaSpace{}%
\AgdaSymbol{:}\AgdaSpace{}%
\AgdaSymbol{\{}\AgdaBound{α}\AgdaSpace{}%
\AgdaSymbol{:}\AgdaSpace{}%
\AgdaDatatype{𝓣}\AgdaSymbol{\}}\AgdaSpace{}%
\AgdaSymbol{\{}\AgdaBound{Γ}\AgdaSpace{}%
\AgdaSymbol{:}\AgdaSpace{}%
\AgdaDatatype{𝔽↓}\AgdaSpace{}%
\AgdaDatatype{𝓣}\AgdaSymbol{\}}\AgdaSpace{}%
\AgdaSymbol{→}\AgdaSpace{}%
\AgdaDatatype{V}\AgdaSpace{}%
\AgdaBound{α}\AgdaSpace{}%
\AgdaBound{Γ}\AgdaSpace{}%
\AgdaSymbol{→}\AgdaSpace{}%
\AgdaDatatype{𝓛}\AgdaSpace{}%
\AgdaBound{α}\AgdaSpace{}%
\AgdaBound{Γ}\<%
\\
\>[2]\AgdaInductiveConstructor{unit}\AgdaSpace{}%
\AgdaSymbol{:}\AgdaSpace{}%
\AgdaSymbol{\{}\AgdaBound{Γ}\AgdaSpace{}%
\AgdaSymbol{:}\AgdaSpace{}%
\AgdaDatatype{𝔽↓}\AgdaSpace{}%
\AgdaDatatype{𝓣}\AgdaSymbol{\}}\AgdaSpace{}%
\AgdaSymbol{→}\AgdaSpace{}%
\AgdaDatatype{𝓛}\AgdaSpace{}%
\AgdaInductiveConstructor{υ}\AgdaSpace{}%
\AgdaBound{Γ}\<%
\\
\>[2]\AgdaInductiveConstructor{pair}\AgdaSpace{}%
\AgdaSymbol{:}\AgdaSpace{}%
\AgdaSymbol{\{}\AgdaBound{α}\AgdaSpace{}%
\AgdaBound{β}\AgdaSpace{}%
\AgdaSymbol{:}\AgdaSpace{}%
\AgdaDatatype{𝓣}\AgdaSymbol{\}}\AgdaSpace{}%
\AgdaSymbol{\{}\AgdaBound{Γ}\AgdaSpace{}%
\AgdaSymbol{:}\AgdaSpace{}%
\AgdaDatatype{𝔽↓}\AgdaSpace{}%
\AgdaDatatype{𝓣}\AgdaSymbol{\}}\AgdaSpace{}%
\AgdaSymbol{→}\AgdaSpace{}%
\AgdaDatatype{𝓛}\AgdaSpace{}%
\AgdaBound{α}\AgdaSpace{}%
\AgdaBound{Γ}\AgdaSpace{}%
\AgdaSymbol{→}\AgdaSpace{}%
\AgdaDatatype{𝓛}\AgdaSpace{}%
\AgdaBound{β}\AgdaSpace{}%
\AgdaBound{Γ}\AgdaSpace{}%
\AgdaSymbol{→}\AgdaSpace{}%
\AgdaDatatype{𝓛}\AgdaSpace{}%
\AgdaSymbol{(}\AgdaBound{α}\AgdaSpace{}%
\AgdaOperator{\AgdaInductiveConstructor{⋆}}\AgdaSpace{}%
\AgdaBound{β}\AgdaSymbol{)}\AgdaSpace{}%
\AgdaBound{Γ}\<%
\\
\>[2]\AgdaInductiveConstructor{fst₀}\AgdaSpace{}%
\AgdaSymbol{:}\AgdaSpace{}%
\AgdaSymbol{\{}\AgdaBound{α}\AgdaSpace{}%
\AgdaBound{β}\AgdaSpace{}%
\AgdaSymbol{:}\AgdaSpace{}%
\AgdaDatatype{𝓣}\AgdaSymbol{\}}%
\>[20]\AgdaSymbol{\{}\AgdaBound{Γ}\AgdaSpace{}%
\AgdaSymbol{:}\AgdaSpace{}%
\AgdaDatatype{𝔽↓}\AgdaSpace{}%
\AgdaDatatype{𝓣}\AgdaSymbol{\}}\AgdaSpace{}%
\AgdaSymbol{→}\AgdaSpace{}%
\AgdaDatatype{𝓛}\AgdaSpace{}%
\AgdaSymbol{(}\AgdaBound{α}\AgdaSpace{}%
\AgdaOperator{\AgdaInductiveConstructor{⋆}}\AgdaSpace{}%
\AgdaBound{β}\AgdaSymbol{)}\AgdaSpace{}%
\AgdaBound{Γ}\AgdaSpace{}%
\AgdaSymbol{→}\AgdaSpace{}%
\AgdaDatatype{𝓛}\AgdaSpace{}%
\AgdaBound{α}\AgdaSpace{}%
\AgdaBound{Γ}\<%
\\
\>[2]\AgdaInductiveConstructor{snd₀}\AgdaSpace{}%
\AgdaSymbol{:}\AgdaSpace{}%
\AgdaSymbol{\{}\AgdaBound{α}\AgdaSpace{}%
\AgdaBound{β}\AgdaSpace{}%
\AgdaSymbol{:}\AgdaSpace{}%
\AgdaDatatype{𝓣}\AgdaSymbol{\}}\AgdaSpace{}%
\AgdaSymbol{\{}\AgdaBound{Γ}\AgdaSpace{}%
\AgdaSymbol{:}\AgdaSpace{}%
\AgdaDatatype{𝔽↓}\AgdaSpace{}%
\AgdaDatatype{𝓣}\AgdaSymbol{\}}\AgdaSpace{}%
\AgdaSymbol{→}\AgdaSpace{}%
\AgdaDatatype{𝓛}\AgdaSpace{}%
\AgdaSymbol{(}\AgdaBound{α}\AgdaSpace{}%
\AgdaOperator{\AgdaInductiveConstructor{⋆}}\AgdaSpace{}%
\AgdaBound{β}\AgdaSymbol{)}\AgdaSpace{}%
\AgdaBound{Γ}\AgdaSpace{}%
\AgdaSymbol{→}\AgdaSpace{}%
\AgdaDatatype{𝓛}\AgdaSpace{}%
\AgdaBound{β}\AgdaSpace{}%
\AgdaBound{Γ}\<%
\\
\>[2]\AgdaInductiveConstructor{abs}\AgdaSpace{}%
\AgdaSymbol{:}\AgdaSpace{}%
\AgdaSymbol{\{}\AgdaBound{α}\AgdaSpace{}%
\AgdaBound{β}\AgdaSpace{}%
\AgdaSymbol{:}\AgdaSpace{}%
\AgdaDatatype{𝓣}\AgdaSymbol{\}}\AgdaSpace{}%
\AgdaSymbol{\{}\AgdaBound{Γ}\AgdaSpace{}%
\AgdaSymbol{:}\AgdaSpace{}%
\AgdaDatatype{𝔽↓}\AgdaSpace{}%
\AgdaDatatype{𝓣}\AgdaSymbol{\}}\AgdaSpace{}%
\AgdaSymbol{→}\AgdaSpace{}%
\AgdaDatatype{𝓛}\AgdaSpace{}%
\AgdaBound{β}\AgdaSpace{}%
\AgdaSymbol{(}\AgdaBound{Γ}\AgdaSpace{}%
\AgdaOperator{\AgdaInductiveConstructor{::}}\AgdaSpace{}%
\AgdaBound{α}\AgdaSpace{}%
\AgdaSymbol{)}\AgdaSpace{}%
\AgdaSymbol{→}\AgdaSpace{}%
\AgdaDatatype{𝓛}\AgdaSpace{}%
\AgdaSymbol{(}\AgdaBound{α}\AgdaSpace{}%
\AgdaOperator{\AgdaInductiveConstructor{⇒}}\AgdaSpace{}%
\AgdaBound{β}\AgdaSymbol{)}\AgdaSpace{}%
\AgdaBound{Γ}\<%
\\
\>[2]\AgdaInductiveConstructor{app}\AgdaSpace{}%
\AgdaSymbol{:}\AgdaSpace{}%
\AgdaSymbol{\{}\AgdaBound{α}\AgdaSpace{}%
\AgdaBound{β}\AgdaSpace{}%
\AgdaSymbol{:}\AgdaSpace{}%
\AgdaDatatype{𝓣}\AgdaSymbol{\}}\AgdaSpace{}%
\AgdaSymbol{\{}\AgdaBound{Γ}\AgdaSpace{}%
\AgdaSymbol{:}\AgdaSpace{}%
\AgdaDatatype{𝔽↓}\AgdaSpace{}%
\AgdaDatatype{𝓣}\AgdaSymbol{\}}\AgdaSpace{}%
\AgdaSymbol{→}\AgdaSpace{}%
\AgdaDatatype{𝓛}\AgdaSpace{}%
\AgdaSymbol{(}\AgdaBound{α}\AgdaSpace{}%
\AgdaOperator{\AgdaInductiveConstructor{⇒}}\AgdaSpace{}%
\AgdaBound{β}\AgdaSymbol{)}\AgdaSpace{}%
\AgdaBound{Γ}\AgdaSpace{}%
\AgdaSymbol{→}\AgdaSpace{}%
\AgdaDatatype{𝓛}\AgdaSpace{}%
\AgdaBound{α}\AgdaSpace{}%
\AgdaBound{Γ}\AgdaSpace{}%
\AgdaSymbol{→}\AgdaSpace{}%
\AgdaDatatype{𝓛}\AgdaSpace{}%
\AgdaBound{β}\AgdaSpace{}%
\AgdaBound{Γ}\<%
\end{code}
Analogously, the abstract syntax of neutral and normal
terms~(see~(\ref{NEandNFPShsOne}) and~(\ref{NEandNFPShsTwo})) is implemented
by the following mutually-inductive families:
\begin{code}%
\>[0]\AgdaComment{-- neutral and normal terms}\<%
\\
\>[0]\AgdaKeyword{mutual}\<%
\\
\>[0][@{}l@{\AgdaIndent{0}}]%
\>[2]\AgdaKeyword{data}\AgdaSpace{}%
\AgdaDatatype{𝓜}\AgdaSpace{}%
\AgdaSymbol{:}\AgdaSpace{}%
\AgdaDatatype{𝓣}\AgdaSpace{}%
\AgdaSymbol{→}\AgdaSpace{}%
\AgdaDatatype{𝔽↓}\AgdaSpace{}%
\AgdaDatatype{𝓣}\AgdaSpace{}%
\AgdaSymbol{→}\AgdaSpace{}%
\AgdaPrimitiveType{Set}\AgdaSpace{}%
\AgdaKeyword{where}\<%
\\
\>[2][@{}l@{\AgdaIndent{0}}]%
\>[4]\AgdaInductiveConstructor{varₘ}\AgdaSpace{}%
\AgdaSymbol{:}\AgdaSpace{}%
\AgdaSymbol{\{}\AgdaBound{α}\AgdaSpace{}%
\AgdaSymbol{:}\AgdaSpace{}%
\AgdaDatatype{𝓣}\AgdaSymbol{\}}\AgdaSpace{}%
\AgdaSymbol{\{}\AgdaBound{Γ}\AgdaSpace{}%
\AgdaSymbol{:}\AgdaSpace{}%
\AgdaDatatype{𝔽↓}\AgdaSpace{}%
\AgdaDatatype{𝓣}\AgdaSymbol{\}}\AgdaSpace{}%
\AgdaSymbol{→}\AgdaSpace{}%
\AgdaDatatype{V}\AgdaSpace{}%
\AgdaBound{α}\AgdaSpace{}%
\AgdaBound{Γ}\AgdaSpace{}%
\AgdaSymbol{→}\AgdaSpace{}%
\AgdaDatatype{𝓜}\AgdaSpace{}%
\AgdaBound{α}\AgdaSpace{}%
\AgdaBound{Γ}\<%
\\
\>[4]\AgdaInductiveConstructor{fstₘ}\AgdaSpace{}%
\AgdaSymbol{:}\AgdaSpace{}%
\AgdaSymbol{\{}\AgdaBound{α}\AgdaSpace{}%
\AgdaBound{β}\AgdaSpace{}%
\AgdaSymbol{:}\AgdaSpace{}%
\AgdaDatatype{𝓣}\AgdaSymbol{\}}\AgdaSpace{}%
\AgdaSymbol{\{}\AgdaBound{Γ}\AgdaSpace{}%
\AgdaSymbol{:}\AgdaSpace{}%
\AgdaDatatype{𝔽↓}\AgdaSpace{}%
\AgdaDatatype{𝓣}\AgdaSymbol{\}}\AgdaSpace{}%
\AgdaSymbol{→}\AgdaSpace{}%
\AgdaDatatype{𝓜}\AgdaSpace{}%
\AgdaSymbol{(}\AgdaBound{α}\AgdaSpace{}%
\AgdaOperator{\AgdaInductiveConstructor{⋆}}\AgdaSpace{}%
\AgdaBound{β}\AgdaSymbol{)}\AgdaSpace{}%
\AgdaBound{Γ}\AgdaSpace{}%
\AgdaSymbol{→}\AgdaSpace{}%
\AgdaDatatype{𝓜}\AgdaSpace{}%
\AgdaBound{α}\AgdaSpace{}%
\AgdaBound{Γ}\<%
\\
\>[4]\AgdaInductiveConstructor{sndₘ}\AgdaSpace{}%
\AgdaSymbol{:}\AgdaSpace{}%
\AgdaSymbol{\{}\AgdaBound{α}\AgdaSpace{}%
\AgdaBound{β}\AgdaSpace{}%
\AgdaSymbol{:}\AgdaSpace{}%
\AgdaDatatype{𝓣}\AgdaSymbol{\}}\AgdaSpace{}%
\AgdaSymbol{\{}\AgdaBound{Γ}\AgdaSpace{}%
\AgdaSymbol{:}\AgdaSpace{}%
\AgdaDatatype{𝔽↓}\AgdaSpace{}%
\AgdaDatatype{𝓣}\AgdaSymbol{\}}\AgdaSpace{}%
\AgdaSymbol{→}\AgdaSpace{}%
\AgdaDatatype{𝓜}\AgdaSpace{}%
\AgdaSymbol{(}\AgdaBound{α}\AgdaSpace{}%
\AgdaOperator{\AgdaInductiveConstructor{⋆}}\AgdaSpace{}%
\AgdaBound{β}\AgdaSymbol{)}\AgdaSpace{}%
\AgdaBound{Γ}\AgdaSpace{}%
\AgdaSymbol{→}\AgdaSpace{}%
\AgdaDatatype{𝓜}\AgdaSpace{}%
\AgdaBound{β}\AgdaSpace{}%
\AgdaBound{Γ}\<%
\\
\>[4]\AgdaInductiveConstructor{appₘ}\AgdaSpace{}%
\AgdaSymbol{:}\AgdaSpace{}%
\AgdaSymbol{\{}\AgdaBound{α}\AgdaSpace{}%
\AgdaBound{β}\AgdaSpace{}%
\AgdaSymbol{:}\AgdaSpace{}%
\AgdaDatatype{𝓣}\AgdaSymbol{\}}\AgdaSpace{}%
\AgdaSymbol{\{}\AgdaBound{Γ}\AgdaSpace{}%
\AgdaSymbol{:}\AgdaSpace{}%
\AgdaDatatype{𝔽↓}\AgdaSpace{}%
\AgdaDatatype{𝓣}\AgdaSymbol{\}}\AgdaSpace{}%
\AgdaSymbol{→}\AgdaSpace{}%
\AgdaDatatype{𝓜}\AgdaSpace{}%
\AgdaSymbol{(}\AgdaBound{α}\AgdaSpace{}%
\AgdaOperator{\AgdaInductiveConstructor{⇒}}\AgdaSpace{}%
\AgdaBound{β}\AgdaSymbol{)}\AgdaSpace{}%
\AgdaBound{Γ}\AgdaSpace{}%
\AgdaSymbol{→}\AgdaSpace{}%
\AgdaDatatype{𝓝}\AgdaSpace{}%
\AgdaBound{α}\AgdaSpace{}%
\AgdaBound{Γ}\AgdaSpace{}%
\AgdaSymbol{→}\AgdaSpace{}%
\AgdaDatatype{𝓜}\AgdaSpace{}%
\AgdaBound{β}\AgdaSpace{}%
\AgdaBound{Γ}\<%
\\
\>[0]\<%
\\
\>[2]\AgdaKeyword{data}\AgdaSpace{}%
\AgdaDatatype{𝓝}\AgdaSpace{}%
\AgdaSymbol{:}\AgdaSpace{}%
\AgdaDatatype{𝓣}\AgdaSpace{}%
\AgdaSymbol{→}\AgdaSpace{}%
\AgdaDatatype{𝔽↓}\AgdaSpace{}%
\AgdaDatatype{𝓣}\AgdaSpace{}%
\AgdaSymbol{→}\AgdaSpace{}%
\AgdaPrimitiveType{Set}\AgdaSpace{}%
\AgdaKeyword{where}\<%
\\
\>[2][@{}l@{\AgdaIndent{0}}]%
\>[4]\AgdaInductiveConstructor{varₙ}\AgdaSpace{}%
\AgdaSymbol{:}\AgdaSpace{}%
\AgdaSymbol{\{}\AgdaBound{i}\AgdaSpace{}%
\AgdaSymbol{:}\AgdaSpace{}%
\AgdaFunction{T}\AgdaSymbol{\}}\AgdaSpace{}%
\AgdaSymbol{\{}\AgdaBound{Γ}\AgdaSpace{}%
\AgdaSymbol{:}\AgdaSpace{}%
\AgdaDatatype{𝔽↓}\AgdaSpace{}%
\AgdaDatatype{𝓣}\AgdaSymbol{\}}\AgdaSpace{}%
\AgdaSymbol{→}\AgdaSpace{}%
\AgdaDatatype{V}\AgdaSpace{}%
\AgdaSymbol{(}\AgdaInductiveConstructor{θ}\AgdaSpace{}%
\AgdaBound{i}\AgdaSymbol{)}\AgdaSpace{}%
\AgdaBound{Γ}\AgdaSpace{}%
\AgdaSymbol{→}\AgdaSpace{}%
\AgdaDatatype{𝓝}\AgdaSpace{}%
\AgdaSymbol{(}\AgdaInductiveConstructor{θ}\AgdaSpace{}%
\AgdaBound{i}\AgdaSymbol{)}\AgdaSpace{}%
\AgdaBound{Γ}\<%
\\
\>[4]\AgdaInductiveConstructor{fstₙ}\AgdaSpace{}%
\AgdaSymbol{:}\AgdaSpace{}%
\AgdaSymbol{\{}\AgdaBound{i}\AgdaSpace{}%
\AgdaSymbol{:}\AgdaSpace{}%
\AgdaFunction{T}\AgdaSymbol{\}}\AgdaSpace{}%
\AgdaSymbol{\{}\AgdaBound{β}\AgdaSpace{}%
\AgdaSymbol{:}\AgdaSpace{}%
\AgdaDatatype{𝓣}\AgdaSymbol{\}}\AgdaSpace{}%
\AgdaSymbol{\{}\AgdaBound{Γ}\AgdaSpace{}%
\AgdaSymbol{:}\AgdaSpace{}%
\AgdaDatatype{𝔽↓}\AgdaSpace{}%
\AgdaDatatype{𝓣}\AgdaSymbol{\}}\AgdaSpace{}%
\AgdaSymbol{→}\AgdaSpace{}%
\AgdaDatatype{𝓜}\AgdaSpace{}%
\AgdaSymbol{(}\AgdaInductiveConstructor{θ}\AgdaSpace{}%
\AgdaBound{i}\AgdaSpace{}%
\AgdaOperator{\AgdaInductiveConstructor{⋆}}\AgdaSpace{}%
\AgdaBound{β}\AgdaSymbol{)}\AgdaSpace{}%
\AgdaBound{Γ}\AgdaSpace{}%
\AgdaSymbol{→}\AgdaSpace{}%
\AgdaDatatype{𝓝}\AgdaSpace{}%
\AgdaSymbol{(}\AgdaInductiveConstructor{θ}\AgdaSpace{}%
\AgdaBound{i}\AgdaSymbol{)}\AgdaSpace{}%
\AgdaBound{Γ}\<%
\\
\>[4]\AgdaInductiveConstructor{sndₙ}\AgdaSpace{}%
\AgdaSymbol{:}\AgdaSpace{}%
\AgdaSymbol{\{}\AgdaBound{i}\AgdaSpace{}%
\AgdaSymbol{:}\AgdaSpace{}%
\AgdaFunction{T}\AgdaSymbol{\}}\AgdaSpace{}%
\AgdaSymbol{\{}\AgdaBound{α}\AgdaSpace{}%
\AgdaSymbol{:}\AgdaSpace{}%
\AgdaDatatype{𝓣}\AgdaSymbol{\}}\AgdaSpace{}%
\AgdaSymbol{\{}\AgdaBound{Γ}\AgdaSpace{}%
\AgdaSymbol{:}\AgdaSpace{}%
\AgdaDatatype{𝔽↓}\AgdaSpace{}%
\AgdaDatatype{𝓣}\AgdaSymbol{\}}\AgdaSpace{}%
\AgdaSymbol{→}\AgdaSpace{}%
\AgdaDatatype{𝓜}\AgdaSpace{}%
\AgdaSymbol{(}\AgdaBound{α}\AgdaSpace{}%
\AgdaOperator{\AgdaInductiveConstructor{⋆}}\AgdaSpace{}%
\AgdaInductiveConstructor{θ}\AgdaSpace{}%
\AgdaBound{i}\AgdaSymbol{)}\AgdaSpace{}%
\AgdaBound{Γ}\AgdaSpace{}%
\AgdaSymbol{→}\AgdaSpace{}%
\AgdaDatatype{𝓝}\AgdaSpace{}%
\AgdaSymbol{(}\AgdaInductiveConstructor{θ}\AgdaSpace{}%
\AgdaBound{i}\AgdaSymbol{)}\AgdaSpace{}%
\AgdaBound{Γ}\<%
\\
\>[4]\AgdaInductiveConstructor{appₙ}\AgdaSpace{}%
\AgdaSymbol{:}\AgdaSpace{}%
\AgdaSymbol{\{}\AgdaBound{i}\AgdaSpace{}%
\AgdaSymbol{:}\AgdaSpace{}%
\AgdaFunction{T}\AgdaSymbol{\}}\AgdaSpace{}%
\AgdaSymbol{\{}\AgdaBound{α}\AgdaSpace{}%
\AgdaSymbol{:}\AgdaSpace{}%
\AgdaDatatype{𝓣}\AgdaSymbol{\}}\AgdaSpace{}%
\AgdaSymbol{\{}\AgdaBound{Γ}\AgdaSpace{}%
\AgdaSymbol{:}\AgdaSpace{}%
\AgdaDatatype{𝔽↓}\AgdaSpace{}%
\AgdaDatatype{𝓣}\AgdaSymbol{\}}\AgdaSpace{}%
\AgdaSymbol{→}\AgdaSpace{}%
\AgdaDatatype{𝓜}\AgdaSpace{}%
\AgdaSymbol{(}\AgdaBound{α}\AgdaSpace{}%
\AgdaOperator{\AgdaInductiveConstructor{⇒}}\AgdaSpace{}%
\AgdaInductiveConstructor{θ}\AgdaSpace{}%
\AgdaBound{i}\AgdaSymbol{)}\AgdaSpace{}%
\AgdaBound{Γ}\AgdaSpace{}%
\AgdaSymbol{→}\AgdaSpace{}%
\AgdaDatatype{𝓝}\AgdaSpace{}%
\AgdaBound{α}\AgdaSpace{}%
\AgdaBound{Γ}\AgdaSpace{}%
\AgdaSymbol{→}\AgdaSpace{}%
\AgdaDatatype{𝓝}\AgdaSpace{}%
\AgdaSymbol{(}\AgdaInductiveConstructor{θ}\AgdaSpace{}%
\AgdaBound{i}\AgdaSymbol{)}\AgdaSpace{}%
\AgdaBound{Γ}\<%
\\
\>[4]\AgdaInductiveConstructor{unitₙ}\AgdaSpace{}%
\AgdaSymbol{:}\AgdaSpace{}%
\AgdaSymbol{\{}\AgdaBound{Γ}\AgdaSpace{}%
\AgdaSymbol{:}\AgdaSpace{}%
\AgdaDatatype{𝔽↓}\AgdaSpace{}%
\AgdaDatatype{𝓣}\AgdaSymbol{\}}\AgdaSpace{}%
\AgdaSymbol{→}\AgdaSpace{}%
\AgdaDatatype{𝓝}\AgdaSpace{}%
\AgdaInductiveConstructor{υ}\AgdaSpace{}%
\AgdaBound{Γ}\<%
\\
\>[4]\AgdaInductiveConstructor{pairₙ}\AgdaSpace{}%
\AgdaSymbol{:}\AgdaSpace{}%
\AgdaSymbol{\{}\AgdaBound{α}\AgdaSpace{}%
\AgdaBound{β}\AgdaSpace{}%
\AgdaSymbol{:}\AgdaSpace{}%
\AgdaDatatype{𝓣}\AgdaSymbol{\}}\AgdaSpace{}%
\AgdaSymbol{\{}\AgdaBound{Γ}\AgdaSpace{}%
\AgdaSymbol{:}\AgdaSpace{}%
\AgdaDatatype{𝔽↓}\AgdaSpace{}%
\AgdaDatatype{𝓣}\AgdaSymbol{\}}\AgdaSpace{}%
\AgdaSymbol{→}\AgdaSpace{}%
\AgdaDatatype{𝓝}\AgdaSpace{}%
\AgdaBound{α}\AgdaSpace{}%
\AgdaBound{Γ}\AgdaSpace{}%
\AgdaSymbol{→}\AgdaSpace{}%
\AgdaDatatype{𝓝}\AgdaSpace{}%
\AgdaBound{β}\AgdaSpace{}%
\AgdaBound{Γ}\AgdaSpace{}%
\AgdaSymbol{→}\AgdaSpace{}%
\AgdaDatatype{𝓝}\AgdaSpace{}%
\AgdaSymbol{(}\AgdaBound{α}\AgdaSpace{}%
\AgdaOperator{\AgdaInductiveConstructor{⋆}}\AgdaSpace{}%
\AgdaBound{β}\AgdaSymbol{)}\AgdaSpace{}%
\AgdaBound{Γ}\<%
\\
\>[4]\AgdaInductiveConstructor{absₙ}\AgdaSpace{}%
\AgdaSymbol{:}\AgdaSpace{}%
\AgdaSymbol{\{}\AgdaBound{α}\AgdaSpace{}%
\AgdaBound{β}\AgdaSpace{}%
\AgdaSymbol{:}\AgdaSpace{}%
\AgdaDatatype{𝓣}\AgdaSymbol{\}}\AgdaSpace{}%
\AgdaSymbol{\{}\AgdaBound{Γ}\AgdaSpace{}%
\AgdaSymbol{:}\AgdaSpace{}%
\AgdaDatatype{𝔽↓}\AgdaSpace{}%
\AgdaDatatype{𝓣}\AgdaSymbol{\}}\AgdaSpace{}%
\AgdaSymbol{→}\AgdaSpace{}%
\AgdaDatatype{𝓝}\AgdaSpace{}%
\AgdaBound{β}\AgdaSpace{}%
\AgdaSymbol{(}\AgdaBound{Γ}\AgdaSpace{}%
\AgdaOperator{\AgdaInductiveConstructor{::}}\AgdaSpace{}%
\AgdaBound{α}\AgdaSymbol{)}\AgdaSpace{}%
\AgdaSymbol{→}\AgdaSpace{}%
\AgdaDatatype{𝓝}\AgdaSpace{}%
\AgdaSymbol{(}\AgdaBound{α}\AgdaSpace{}%
\AgdaOperator{\AgdaInductiveConstructor{⇒}}\AgdaSpace{}%
\AgdaBound{β}\AgdaSymbol{)}\AgdaSpace{}%
\AgdaBound{Γ}\<%
\end{code}
Their presheaf actions~(recall~(\ref{Neutral_Renaming_Invariance})
and~(\ref{Normal_Renaming_Invariance})) will be needed:
\begin{code}%
\>[0]\AgdaComment{-- neutral and normal presheaf actions}\<%
\\
\>[0]\AgdaKeyword{mutual}\<%
\\
\>[0][@{}l@{\AgdaIndent{0}}]%
\>[2]\AgdaOperator{\AgdaFunction{\AgdaUnderscore{}[\AgdaUnderscore{}]ₘ}}\AgdaSpace{}%
\AgdaSymbol{:}\AgdaSpace{}%
\AgdaSymbol{\{}\AgdaBound{α}\AgdaSpace{}%
\AgdaSymbol{:}\AgdaSpace{}%
\AgdaDatatype{𝓣}\AgdaSymbol{\}}\AgdaSpace{}%
\AgdaSymbol{\{}\AgdaBound{Δ}\AgdaSpace{}%
\AgdaBound{Γ}\AgdaSpace{}%
\AgdaSymbol{:}\AgdaSpace{}%
\AgdaDatatype{𝔽↓}\AgdaSpace{}%
\AgdaDatatype{𝓣}\AgdaSymbol{\}}\AgdaSpace{}%
\AgdaSymbol{→}\AgdaSpace{}%
\AgdaDatatype{𝓜}\AgdaSpace{}%
\AgdaBound{α}\AgdaSpace{}%
\AgdaBound{Δ}\AgdaSpace{}%
\AgdaSymbol{→}\AgdaSpace{}%
\AgdaFunction{𝔽↓₀}\AgdaSpace{}%
\AgdaDatatype{𝓣}\AgdaSymbol{(}\AgdaSpace{}%
\AgdaBound{Δ}\AgdaSpace{}%
\AgdaOperator{\AgdaInductiveConstructor{,}}\AgdaSpace{}%
\AgdaBound{Γ}\AgdaSpace{}%
\AgdaSymbol{)}\AgdaSpace{}%
\AgdaSymbol{→}\AgdaSpace{}%
\AgdaDatatype{𝓜}\AgdaSpace{}%
\AgdaBound{α}\AgdaSpace{}%
\AgdaBound{Γ}\<%
\\
\>[2]\AgdaInductiveConstructor{varₘ}\AgdaSpace{}%
\AgdaBound{x}\AgdaSpace{}%
\AgdaOperator{\AgdaFunction{[}}\AgdaSpace{}%
\AgdaBound{ρ}\AgdaSpace{}%
\AgdaOperator{\AgdaFunction{]ₘ}}\AgdaSpace{}%
\AgdaSymbol{=}\AgdaSpace{}%
\AgdaInductiveConstructor{varₘ}\AgdaSpace{}%
\AgdaSymbol{(}\AgdaSpace{}%
\AgdaBound{ρ}\AgdaSpace{}%
\AgdaBound{x}\AgdaSpace{}%
\AgdaSymbol{)}\<%
\\
\>[2]\AgdaInductiveConstructor{fstₘ}\AgdaSpace{}%
\AgdaBound{m}\AgdaSpace{}%
\AgdaOperator{\AgdaFunction{[}}\AgdaSpace{}%
\AgdaBound{ρ}\AgdaSpace{}%
\AgdaOperator{\AgdaFunction{]ₘ}}\AgdaSpace{}%
\AgdaSymbol{=}\AgdaSpace{}%
\AgdaInductiveConstructor{fstₘ}\AgdaSpace{}%
\AgdaSymbol{(}\AgdaSpace{}%
\AgdaBound{m}\AgdaSpace{}%
\AgdaOperator{\AgdaFunction{[}}\AgdaSpace{}%
\AgdaBound{ρ}\AgdaSpace{}%
\AgdaOperator{\AgdaFunction{]ₘ}}\AgdaSpace{}%
\AgdaSymbol{)}\<%
\\
\>[2]\AgdaInductiveConstructor{sndₘ}\AgdaSpace{}%
\AgdaBound{m}\AgdaSpace{}%
\AgdaOperator{\AgdaFunction{[}}\AgdaSpace{}%
\AgdaBound{ρ}\AgdaSpace{}%
\AgdaOperator{\AgdaFunction{]ₘ}}\AgdaSpace{}%
\AgdaSymbol{=}\AgdaSpace{}%
\AgdaInductiveConstructor{sndₘ}\AgdaSpace{}%
\AgdaSymbol{(}\AgdaSpace{}%
\AgdaBound{m}\AgdaSpace{}%
\AgdaOperator{\AgdaFunction{[}}\AgdaSpace{}%
\AgdaBound{ρ}\AgdaSpace{}%
\AgdaOperator{\AgdaFunction{]ₘ}}\AgdaSpace{}%
\AgdaSymbol{)}\<%
\\
\>[2]\AgdaInductiveConstructor{appₘ}\AgdaSpace{}%
\AgdaBound{m}\AgdaSpace{}%
\AgdaBound{n}\AgdaSpace{}%
\AgdaOperator{\AgdaFunction{[}}\AgdaSpace{}%
\AgdaBound{ρ}\AgdaSpace{}%
\AgdaOperator{\AgdaFunction{]ₘ}}\AgdaSpace{}%
\AgdaSymbol{=}\AgdaSpace{}%
\AgdaInductiveConstructor{appₘ}\AgdaSpace{}%
\AgdaSymbol{(}\AgdaSpace{}%
\AgdaBound{m}\AgdaSpace{}%
\AgdaOperator{\AgdaFunction{[}}\AgdaSpace{}%
\AgdaBound{ρ}\AgdaSpace{}%
\AgdaOperator{\AgdaFunction{]ₘ}}\AgdaSpace{}%
\AgdaSymbol{)}\AgdaSpace{}%
\AgdaSymbol{(}\AgdaSpace{}%
\AgdaBound{n}\AgdaSpace{}%
\AgdaOperator{\AgdaFunction{[}}\AgdaSpace{}%
\AgdaBound{ρ}\AgdaSpace{}%
\AgdaOperator{\AgdaFunction{]ₙ}}\AgdaSpace{}%
\AgdaSymbol{)}\<%
\\
\\[\AgdaEmptyExtraSkip]%
\>[2]\AgdaOperator{\AgdaFunction{\AgdaUnderscore{}[\AgdaUnderscore{}]ₙ}}\AgdaSpace{}%
\AgdaSymbol{:}\AgdaSpace{}%
\AgdaSymbol{\{}\AgdaBound{α}\AgdaSpace{}%
\AgdaSymbol{:}\AgdaSpace{}%
\AgdaDatatype{𝓣}\AgdaSymbol{\}}\AgdaSpace{}%
\AgdaSymbol{\{}\AgdaBound{Δ}\AgdaSpace{}%
\AgdaBound{Γ}\AgdaSpace{}%
\AgdaSymbol{:}\AgdaSpace{}%
\AgdaDatatype{𝔽↓}\AgdaSpace{}%
\AgdaDatatype{𝓣}\AgdaSymbol{\}}\AgdaSpace{}%
\AgdaSymbol{→}\AgdaSpace{}%
\AgdaDatatype{𝓝}\AgdaSpace{}%
\AgdaBound{α}\AgdaSpace{}%
\AgdaBound{Δ}\AgdaSpace{}%
\AgdaSymbol{→}\AgdaSpace{}%
\AgdaFunction{𝔽↓₀}\AgdaSpace{}%
\AgdaDatatype{𝓣}\AgdaSymbol{(}\AgdaSpace{}%
\AgdaBound{Δ}\AgdaSpace{}%
\AgdaOperator{\AgdaInductiveConstructor{,}}\AgdaSpace{}%
\AgdaBound{Γ}\AgdaSpace{}%
\AgdaSymbol{)}\AgdaSpace{}%
\AgdaSymbol{→}\AgdaSpace{}%
\AgdaDatatype{𝓝}\AgdaSpace{}%
\AgdaBound{α}\AgdaSpace{}%
\AgdaBound{Γ}\<%
\\
\>[2]\AgdaInductiveConstructor{varₙ}\AgdaSpace{}%
\AgdaBound{x}\AgdaSpace{}%
\AgdaOperator{\AgdaFunction{[}}\AgdaSpace{}%
\AgdaBound{ρ}\AgdaSpace{}%
\AgdaOperator{\AgdaFunction{]ₙ}}\AgdaSpace{}%
\AgdaSymbol{=}\AgdaSpace{}%
\AgdaInductiveConstructor{varₙ}\AgdaSpace{}%
\AgdaSymbol{(}\AgdaSpace{}%
\AgdaBound{ρ}\AgdaSpace{}%
\AgdaBound{x}\AgdaSpace{}%
\AgdaSymbol{)}\<%
\\
\>[2]\AgdaInductiveConstructor{fstₙ}\AgdaSpace{}%
\AgdaBound{m}\AgdaSpace{}%
\AgdaOperator{\AgdaFunction{[}}\AgdaSpace{}%
\AgdaBound{ρ}\AgdaSpace{}%
\AgdaOperator{\AgdaFunction{]ₙ}}\AgdaSpace{}%
\AgdaSymbol{=}\AgdaSpace{}%
\AgdaInductiveConstructor{fstₙ}\AgdaSpace{}%
\AgdaSymbol{(}\AgdaSpace{}%
\AgdaBound{m}\AgdaSpace{}%
\AgdaOperator{\AgdaFunction{[}}\AgdaSpace{}%
\AgdaBound{ρ}\AgdaSpace{}%
\AgdaOperator{\AgdaFunction{]ₘ}}\AgdaSpace{}%
\AgdaSymbol{)}\<%
\\
\>[2]\AgdaInductiveConstructor{sndₙ}\AgdaSpace{}%
\AgdaBound{m}\AgdaSpace{}%
\AgdaOperator{\AgdaFunction{[}}\AgdaSpace{}%
\AgdaBound{ρ}\AgdaSpace{}%
\AgdaOperator{\AgdaFunction{]ₙ}}\AgdaSpace{}%
\AgdaSymbol{=}\AgdaSpace{}%
\AgdaInductiveConstructor{sndₙ}\AgdaSpace{}%
\AgdaSymbol{(}\AgdaSpace{}%
\AgdaBound{m}\AgdaSpace{}%
\AgdaOperator{\AgdaFunction{[}}\AgdaSpace{}%
\AgdaBound{ρ}\AgdaSpace{}%
\AgdaOperator{\AgdaFunction{]ₘ}}\AgdaSpace{}%
\AgdaSymbol{)}\<%
\\
\>[2]\AgdaInductiveConstructor{appₙ}\AgdaSpace{}%
\AgdaBound{m}\AgdaSpace{}%
\AgdaBound{n}\AgdaSpace{}%
\AgdaOperator{\AgdaFunction{[}}\AgdaSpace{}%
\AgdaBound{ρ}\AgdaSpace{}%
\AgdaOperator{\AgdaFunction{]ₙ}}\AgdaSpace{}%
\AgdaSymbol{=}\AgdaSpace{}%
\AgdaInductiveConstructor{appₙ}\AgdaSpace{}%
\AgdaSymbol{(}\AgdaSpace{}%
\AgdaBound{m}\AgdaSpace{}%
\AgdaOperator{\AgdaFunction{[}}\AgdaSpace{}%
\AgdaBound{ρ}\AgdaSpace{}%
\AgdaOperator{\AgdaFunction{]ₘ}}\AgdaSpace{}%
\AgdaSymbol{)}\AgdaSpace{}%
\AgdaSymbol{(}\AgdaSpace{}%
\AgdaBound{n}\AgdaSpace{}%
\AgdaOperator{\AgdaFunction{[}}\AgdaSpace{}%
\AgdaBound{ρ}\AgdaSpace{}%
\AgdaOperator{\AgdaFunction{]ₙ}}\AgdaSpace{}%
\AgdaSymbol{)}\<%
\\
\>[2]\AgdaInductiveConstructor{unitₙ}\AgdaSpace{}%
\AgdaOperator{\AgdaFunction{[}}\AgdaSpace{}%
\AgdaBound{ρ}\AgdaSpace{}%
\AgdaOperator{\AgdaFunction{]ₙ}}\AgdaSpace{}%
\AgdaSymbol{=}\AgdaSpace{}%
\AgdaInductiveConstructor{unitₙ}\<%
\\
\>[2]\AgdaInductiveConstructor{pairₙ}\AgdaSpace{}%
\AgdaBound{n₁}\AgdaSpace{}%
\AgdaBound{n₂}\AgdaSpace{}%
\AgdaOperator{\AgdaFunction{[}}\AgdaSpace{}%
\AgdaBound{ρ}\AgdaSpace{}%
\AgdaOperator{\AgdaFunction{]ₙ}}\AgdaSpace{}%
\AgdaSymbol{=}\AgdaSpace{}%
\AgdaInductiveConstructor{pairₙ}\AgdaSpace{}%
\AgdaSymbol{(}\AgdaSpace{}%
\AgdaBound{n₁}\AgdaSpace{}%
\AgdaOperator{\AgdaFunction{[}}\AgdaSpace{}%
\AgdaBound{ρ}\AgdaSpace{}%
\AgdaOperator{\AgdaFunction{]ₙ}}\AgdaSpace{}%
\AgdaSymbol{)}\AgdaSpace{}%
\AgdaSymbol{(}\AgdaSpace{}%
\AgdaBound{n₂}\AgdaSpace{}%
\AgdaOperator{\AgdaFunction{[}}\AgdaSpace{}%
\AgdaBound{ρ}\AgdaSpace{}%
\AgdaOperator{\AgdaFunction{]ₙ}}\AgdaSpace{}%
\AgdaSymbol{)}\<%
\\
\>[2]\AgdaInductiveConstructor{absₙ}\AgdaSpace{}%
\AgdaBound{n}\AgdaSpace{}%
\AgdaOperator{\AgdaFunction{[}}\AgdaSpace{}%
\AgdaBound{ρ}\AgdaSpace{}%
\AgdaOperator{\AgdaFunction{]ₙ}}\AgdaSpace{}%
\AgdaSymbol{=}%
\>[804I]\AgdaInductiveConstructor{absₙ}\AgdaSpace{}%
\AgdaSymbol{(}\AgdaSpace{}%
\AgdaBound{n}\AgdaSpace{}%
\AgdaOperator{\AgdaFunction{[}}\AgdaSpace{}%
\AgdaFunction{lift}\AgdaSpace{}%
\AgdaBound{ρ}\AgdaSpace{}%
\AgdaOperator{\AgdaFunction{]ₙ}}\AgdaSpace{}%
\AgdaSymbol{)}\<%
\\
\>[.][@{}l@{}]\<[804I]%
\>[18]\AgdaKeyword{where}\<%
\\
\>[18][@{}l@{\AgdaIndent{0}}]%
\>[20]\AgdaFunction{lift}\AgdaSpace{}%
\AgdaSymbol{:}\AgdaSpace{}%
\AgdaSymbol{\{}\AgdaBound{α}\AgdaSpace{}%
\AgdaSymbol{:}\AgdaSpace{}%
\AgdaDatatype{𝓣}\AgdaSymbol{\}}\AgdaSpace{}%
\AgdaSymbol{\{}\AgdaBound{Δ}\AgdaSpace{}%
\AgdaBound{Γ}\AgdaSpace{}%
\AgdaSymbol{:}\AgdaSpace{}%
\AgdaDatatype{𝔽↓}\AgdaSpace{}%
\AgdaDatatype{𝓣}\AgdaSymbol{\}}\AgdaSpace{}%
\AgdaSymbol{→}\AgdaSpace{}%
\AgdaFunction{𝔽↓₀}\AgdaSpace{}%
\AgdaDatatype{𝓣}\AgdaSymbol{(}\AgdaSpace{}%
\AgdaBound{Δ}\AgdaSpace{}%
\AgdaOperator{\AgdaInductiveConstructor{,}}\AgdaSpace{}%
\AgdaBound{Γ}\AgdaSpace{}%
\AgdaSymbol{)}\AgdaSpace{}%
\AgdaSymbol{→}\AgdaSpace{}%
\AgdaFunction{𝔽↓₀}\AgdaSpace{}%
\AgdaDatatype{𝓣}\AgdaSymbol{(}\AgdaSpace{}%
\AgdaBound{Δ}\AgdaSpace{}%
\AgdaOperator{\AgdaInductiveConstructor{::}}\AgdaSpace{}%
\AgdaBound{α}\AgdaSpace{}%
\AgdaOperator{\AgdaInductiveConstructor{,}}\AgdaSpace{}%
\AgdaBound{Γ}\AgdaSpace{}%
\AgdaOperator{\AgdaInductiveConstructor{::}}\AgdaSpace{}%
\AgdaBound{α}\AgdaSpace{}%
\AgdaSymbol{)}\<%
\\
\>[20]\AgdaFunction{lift}\AgdaSpace{}%
\AgdaSymbol{\AgdaUnderscore{}}\AgdaSpace{}%
\AgdaInductiveConstructor{•}\AgdaSpace{}%
\AgdaSymbol{=}\AgdaSpace{}%
\AgdaInductiveConstructor{•}\<%
\\
\>[20]\AgdaFunction{lift}\AgdaSpace{}%
\AgdaBound{ρ}\AgdaSpace{}%
\AgdaSymbol{(}\AgdaInductiveConstructor{↑}\AgdaSpace{}%
\AgdaBound{x}\AgdaSymbol{)}\AgdaSpace{}%
\AgdaSymbol{=}\AgdaSpace{}%
\AgdaInductiveConstructor{↑}\AgdaSymbol{(}\AgdaBound{ρ}\AgdaSpace{}%
\AgdaBound{x}\AgdaSymbol{)}\<%
\end{code}
Note the treatment of abstraction guaranteeing fresh bindings.

\medskip
\mysubparagraph{Semantics.}
We implement the presheaf semantics of types induced by the interpretation of
base types as neutral terms~(see~(\ref{mu_Interpretation})).  Note that for
higher types the implementation glosses over the naturality condition required
of presheaf exponentials~(see~(\ref{PShExponentialNaturality})).
\begin{code}%
\>[0]\AgdaComment{-- type semantics}\<%
\\
\>[0]\AgdaOperator{\AgdaFunction{⟦\AgdaUnderscore{}⟧}}\AgdaSpace{}%
\AgdaSymbol{:}\AgdaSpace{}%
\AgdaDatatype{𝓣}\AgdaSpace{}%
\AgdaSymbol{→}\AgdaSpace{}%
\AgdaDatatype{𝔽↓}\AgdaSpace{}%
\AgdaDatatype{𝓣}\AgdaSpace{}%
\AgdaSymbol{→}\AgdaSpace{}%
\AgdaPrimitiveType{Set}\<%
\\
\>[0]\AgdaOperator{\AgdaFunction{⟦}}\AgdaSpace{}%
\AgdaInductiveConstructor{θ}\AgdaSpace{}%
\AgdaBound{i}\AgdaSpace{}%
\AgdaOperator{\AgdaFunction{⟧}}\AgdaSpace{}%
\AgdaBound{Γ}\AgdaSpace{}%
\AgdaSymbol{=}\AgdaSpace{}%
\AgdaDatatype{𝓜}\AgdaSpace{}%
\AgdaSymbol{(}\AgdaInductiveConstructor{θ}\AgdaSpace{}%
\AgdaBound{i}\AgdaSymbol{)}\AgdaSpace{}%
\AgdaBound{Γ}\<%
\\
\>[0]\AgdaOperator{\AgdaFunction{⟦}}\AgdaSpace{}%
\AgdaInductiveConstructor{υ}\AgdaSpace{}%
\AgdaOperator{\AgdaFunction{⟧}}\AgdaSpace{}%
\AgdaSymbol{\AgdaUnderscore{}}\AgdaSpace{}%
\AgdaSymbol{=}\AgdaSpace{}%
\AgdaRecord{⊤}\<%
\\
\>[0]\AgdaOperator{\AgdaFunction{⟦}}\AgdaSpace{}%
\AgdaBound{α}\AgdaSpace{}%
\AgdaOperator{\AgdaInductiveConstructor{⋆}}\AgdaSpace{}%
\AgdaBound{β}\AgdaSpace{}%
\AgdaOperator{\AgdaFunction{⟧}}\AgdaSpace{}%
\AgdaBound{Γ}\AgdaSpace{}%
\AgdaSymbol{=}\AgdaSpace{}%
\AgdaOperator{\AgdaFunction{⟦}}\AgdaSpace{}%
\AgdaBound{α}\AgdaSpace{}%
\AgdaOperator{\AgdaFunction{⟧}}\AgdaSpace{}%
\AgdaBound{Γ}\AgdaSpace{}%
\AgdaOperator{\AgdaFunction{×}}\AgdaSpace{}%
\AgdaOperator{\AgdaFunction{⟦}}\AgdaSpace{}%
\AgdaBound{β}\AgdaSpace{}%
\AgdaOperator{\AgdaFunction{⟧}}\AgdaSpace{}%
\AgdaBound{Γ}\<%
\\
\>[0]\AgdaOperator{\AgdaFunction{⟦}}\AgdaSpace{}%
\AgdaBound{α}\AgdaSpace{}%
\AgdaOperator{\AgdaInductiveConstructor{⇒}}\AgdaSpace{}%
\AgdaBound{β}\AgdaSpace{}%
\AgdaOperator{\AgdaFunction{⟧}}\AgdaSpace{}%
\AgdaBound{Γ}\AgdaSpace{}%
\AgdaSymbol{=}\AgdaSpace{}%
\AgdaSymbol{\{}\AgdaBound{Δ}\AgdaSpace{}%
\AgdaSymbol{:}\AgdaSpace{}%
\AgdaDatatype{𝔽↓}\AgdaSpace{}%
\AgdaDatatype{𝓣}\AgdaSymbol{\}}\AgdaSpace{}%
\AgdaSymbol{→}\AgdaSpace{}%
\AgdaFunction{𝔽↓₀}\AgdaSpace{}%
\AgdaDatatype{𝓣}\AgdaSymbol{(}\AgdaSpace{}%
\AgdaBound{Γ}\AgdaSpace{}%
\AgdaOperator{\AgdaInductiveConstructor{,}}\AgdaSpace{}%
\AgdaBound{Δ}\AgdaSpace{}%
\AgdaSymbol{)}\AgdaSpace{}%
\AgdaSymbol{→}\AgdaSpace{}%
\AgdaOperator{\AgdaFunction{⟦}}\AgdaSpace{}%
\AgdaBound{α}\AgdaSpace{}%
\AgdaOperator{\AgdaFunction{⟧}}\AgdaSpace{}%
\AgdaBound{Δ}\AgdaSpace{}%
\AgdaSymbol{→}\AgdaSpace{}%
\AgdaOperator{\AgdaFunction{⟦}}\AgdaSpace{}%
\AgdaBound{β}\AgdaSpace{}%
\AgdaOperator{\AgdaFunction{⟧}}\AgdaSpace{}%
\AgdaBound{Δ}\<%
\\
\\[\AgdaEmptyExtraSkip]%
\>[0]\AgdaOperator{\AgdaFunction{\AgdaUnderscore{}[\AgdaUnderscore{}]}}\AgdaSpace{}%
\AgdaSymbol{:}\AgdaSpace{}%
\AgdaSymbol{\{}\AgdaBound{α}\AgdaSpace{}%
\AgdaSymbol{:}\AgdaSpace{}%
\AgdaDatatype{𝓣}\AgdaSymbol{\}}\AgdaSpace{}%
\AgdaSymbol{\{}\AgdaBound{Δ}\AgdaSpace{}%
\AgdaBound{Γ}\AgdaSpace{}%
\AgdaSymbol{:}\AgdaSpace{}%
\AgdaDatatype{𝔽↓}\AgdaSpace{}%
\AgdaDatatype{𝓣}\AgdaSymbol{\}}\AgdaSpace{}%
\AgdaSymbol{→}\AgdaSpace{}%
\AgdaOperator{\AgdaFunction{⟦}}\AgdaSpace{}%
\AgdaBound{α}\AgdaSpace{}%
\AgdaOperator{\AgdaFunction{⟧}}\AgdaSpace{}%
\AgdaBound{Δ}\AgdaSpace{}%
\AgdaSymbol{→}\AgdaSpace{}%
\AgdaFunction{𝔽↓₀}\AgdaSpace{}%
\AgdaDatatype{𝓣}\AgdaSymbol{(}\AgdaSpace{}%
\AgdaBound{Δ}\AgdaSpace{}%
\AgdaOperator{\AgdaInductiveConstructor{,}}\AgdaSpace{}%
\AgdaBound{Γ}\AgdaSpace{}%
\AgdaSymbol{)}\AgdaSpace{}%
\AgdaSymbol{→}\AgdaSpace{}%
\AgdaOperator{\AgdaFunction{⟦}}\AgdaSpace{}%
\AgdaBound{α}\AgdaSpace{}%
\AgdaOperator{\AgdaFunction{⟧}}\AgdaSpace{}%
\AgdaBound{Γ}\<%
\\
\>[0]\AgdaOperator{\AgdaFunction{\AgdaUnderscore{}[\AgdaUnderscore{}]}}\AgdaSpace{}%
\AgdaSymbol{\{}\AgdaInductiveConstructor{θ}\AgdaSpace{}%
\AgdaSymbol{\AgdaUnderscore{}\}}\AgdaSpace{}%
\AgdaSymbol{=}\AgdaSpace{}%
\AgdaOperator{\AgdaFunction{\AgdaUnderscore{}[\AgdaUnderscore{}]ₘ}}\<%
\\
\>[0]\AgdaOperator{\AgdaFunction{\AgdaUnderscore{}[\AgdaUnderscore{}]}}\AgdaSpace{}%
\AgdaSymbol{\{}\AgdaInductiveConstructor{υ}\AgdaSymbol{\}}\AgdaSpace{}%
\AgdaSymbol{=}\AgdaSpace{}%
\AgdaSymbol{\AgdaUnderscore{}}\<%
\\
\>[0]\AgdaOperator{\AgdaFunction{\AgdaUnderscore{}[\AgdaUnderscore{}]}}\AgdaSpace{}%
\AgdaSymbol{\{\AgdaUnderscore{}}\AgdaSpace{}%
\AgdaOperator{\AgdaInductiveConstructor{⋆}}\AgdaSpace{}%
\AgdaSymbol{\AgdaUnderscore{}\}}\AgdaSpace{}%
\AgdaSymbol{(}\AgdaSpace{}%
\AgdaBound{x₁}\AgdaSpace{}%
\AgdaOperator{\AgdaInductiveConstructor{,}}\AgdaSpace{}%
\AgdaBound{x₂}\AgdaSpace{}%
\AgdaSymbol{)}\AgdaSpace{}%
\AgdaBound{ρ}\AgdaSpace{}%
\AgdaSymbol{=}\AgdaSpace{}%
\AgdaSymbol{(}\AgdaSpace{}%
\AgdaBound{x₁}\AgdaSpace{}%
\AgdaOperator{\AgdaFunction{[}}\AgdaSpace{}%
\AgdaBound{ρ}\AgdaSpace{}%
\AgdaOperator{\AgdaFunction{]}}\AgdaSpace{}%
\AgdaOperator{\AgdaInductiveConstructor{,}}\AgdaSpace{}%
\AgdaBound{x₂}\AgdaSpace{}%
\AgdaOperator{\AgdaFunction{[}}\AgdaSpace{}%
\AgdaBound{ρ}\AgdaSpace{}%
\AgdaOperator{\AgdaFunction{]}}\AgdaSpace{}%
\AgdaSymbol{)}\<%
\\
\>[0]\AgdaOperator{\AgdaFunction{\AgdaUnderscore{}[\AgdaUnderscore{}]}}\AgdaSpace{}%
\AgdaSymbol{\{\AgdaUnderscore{}}\AgdaSpace{}%
\AgdaOperator{\AgdaInductiveConstructor{⇒}}\AgdaSpace{}%
\AgdaSymbol{\AgdaUnderscore{}\}}\AgdaSpace{}%
\AgdaBound{f}\AgdaSpace{}%
\AgdaBound{ρ}\AgdaSpace{}%
\AgdaBound{ρ′}\AgdaSpace{}%
\AgdaSymbol{=}\AgdaSpace{}%
\AgdaBound{f}\AgdaSpace{}%
\AgdaSymbol{(}\AgdaBound{ρ′}\AgdaSpace{}%
\AgdaOperator{\AgdaFunction{∘}}\AgdaSpace{}%
\AgdaBound{ρ}\AgdaSymbol{)}\<%
\end{code}

The semantic interpretation of
terms~(see~(\ref{StandardSemanticInterpretation})) follows:
\begin{code}%
\>[0]\AgdaComment{-- term semantics}\<%
\\
\>[0]\AgdaFunction{Π}\AgdaSpace{}%
\AgdaSymbol{:}\AgdaSpace{}%
\AgdaDatatype{𝔽↓}\AgdaSpace{}%
\AgdaDatatype{𝓣}\AgdaSpace{}%
\AgdaSymbol{→}\AgdaSpace{}%
\AgdaSymbol{(}\AgdaDatatype{𝓣}\AgdaSpace{}%
\AgdaSymbol{→}\AgdaSpace{}%
\AgdaDatatype{𝔽↓}\AgdaSpace{}%
\AgdaDatatype{𝓣}\AgdaSpace{}%
\AgdaSymbol{→}\AgdaSpace{}%
\AgdaPrimitiveType{Set}\AgdaSymbol{)}\AgdaSpace{}%
\AgdaSymbol{→}\AgdaSpace{}%
\AgdaDatatype{𝔽↓}\AgdaSpace{}%
\AgdaDatatype{𝓣}\AgdaSpace{}%
\AgdaSymbol{→}\AgdaSpace{}%
\AgdaPrimitiveType{Set}\<%
\\
\>[0]\AgdaFunction{Π}\AgdaSpace{}%
\AgdaInductiveConstructor{·}\AgdaSpace{}%
\AgdaSymbol{\AgdaUnderscore{}}\AgdaSpace{}%
\AgdaSymbol{\AgdaUnderscore{}}\AgdaSpace{}%
\AgdaSymbol{=}\AgdaSpace{}%
\AgdaRecord{⊤}\<%
\\
\>[0]\AgdaFunction{Π}\AgdaSpace{}%
\AgdaSymbol{(}\AgdaBound{Γ}\AgdaSpace{}%
\AgdaOperator{\AgdaInductiveConstructor{::}}\AgdaSpace{}%
\AgdaBound{α}\AgdaSymbol{)}\AgdaSpace{}%
\AgdaBound{P}\AgdaSpace{}%
\AgdaBound{Δ}\AgdaSpace{}%
\AgdaSymbol{=}\AgdaSpace{}%
\AgdaSymbol{(}\AgdaFunction{Π}\AgdaSpace{}%
\AgdaBound{Γ}\AgdaSpace{}%
\AgdaBound{P}\AgdaSpace{}%
\AgdaBound{Δ}\AgdaSymbol{)}\AgdaSpace{}%
\AgdaOperator{\AgdaFunction{×}}\AgdaSpace{}%
\AgdaSymbol{(}\AgdaBound{P}\AgdaSpace{}%
\AgdaBound{α}\AgdaSpace{}%
\AgdaBound{Δ}\AgdaSymbol{)}\<%
\\
\\[\AgdaEmptyExtraSkip]%
\>[0]\AgdaOperator{\AgdaFunction{⟦\AgdaUnderscore{}⟧₀}}\AgdaSpace{}%
\AgdaSymbol{:}\AgdaSpace{}%
\AgdaSymbol{\{}\AgdaBound{α}\AgdaSpace{}%
\AgdaSymbol{:}\AgdaSpace{}%
\AgdaDatatype{𝓣}\AgdaSymbol{\}}\AgdaSpace{}%
\AgdaSymbol{\{}\AgdaBound{Γ}\AgdaSpace{}%
\AgdaSymbol{:}\AgdaSpace{}%
\AgdaDatatype{𝔽↓}\AgdaSpace{}%
\AgdaDatatype{𝓣}\AgdaSymbol{\}}\AgdaSpace{}%
\AgdaSymbol{→}\AgdaSpace{}%
\AgdaDatatype{𝓛}\AgdaSpace{}%
\AgdaBound{α}\AgdaSpace{}%
\AgdaBound{Γ}\AgdaSpace{}%
\AgdaSymbol{→}\AgdaSpace{}%
\AgdaSymbol{\{}\AgdaBound{Δ}\AgdaSpace{}%
\AgdaSymbol{:}\AgdaSpace{}%
\AgdaDatatype{𝔽↓}\AgdaSpace{}%
\AgdaDatatype{𝓣}\AgdaSymbol{\}}\AgdaSpace{}%
\AgdaSymbol{→}\AgdaSpace{}%
\AgdaFunction{Π}\AgdaSpace{}%
\AgdaBound{Γ}\AgdaSpace{}%
\AgdaOperator{\AgdaFunction{⟦\AgdaUnderscore{}⟧}}\AgdaSpace{}%
\AgdaBound{Δ}\AgdaSpace{}%
\AgdaSymbol{→}\AgdaSpace{}%
\AgdaOperator{\AgdaFunction{⟦}}\AgdaSpace{}%
\AgdaBound{α}\AgdaSpace{}%
\AgdaOperator{\AgdaFunction{⟧}}\AgdaSpace{}%
\AgdaBound{Δ}\<%
\\
\>[0]\AgdaOperator{\AgdaFunction{⟦}}\AgdaSpace{}%
\AgdaInductiveConstructor{var}\AgdaSpace{}%
\AgdaBound{x}\AgdaSpace{}%
\AgdaOperator{\AgdaFunction{⟧₀}}\AgdaSpace{}%
\AgdaBound{ε}%
\>[1042I]\AgdaSymbol{=}\AgdaSpace{}%
\AgdaBound{ε}\AgdaSpace{}%
\AgdaOperator{\AgdaFunction{≪}}\AgdaSpace{}%
\AgdaBound{x}\AgdaSpace{}%
\AgdaOperator{\AgdaFunction{≫}}\<%
\\
\>[.][@{}l@{}]\<[1042I]%
\>[13]\AgdaKeyword{where}\<%
\\
\>[13][@{}l@{\AgdaIndent{0}}]%
\>[15]\AgdaOperator{\AgdaFunction{\AgdaUnderscore{}≪\AgdaUnderscore{}≫}}\AgdaSpace{}%
\AgdaSymbol{:}\AgdaSpace{}%
\AgdaSymbol{\{}\AgdaBound{α}\AgdaSpace{}%
\AgdaSymbol{:}\AgdaSpace{}%
\AgdaDatatype{𝓣}\AgdaSymbol{\}}\AgdaSpace{}%
\AgdaSymbol{\{}\AgdaBound{Γ}\AgdaSpace{}%
\AgdaBound{Δ}\AgdaSpace{}%
\AgdaSymbol{:}\AgdaSpace{}%
\AgdaDatatype{𝔽↓}\AgdaSpace{}%
\AgdaDatatype{𝓣}\AgdaSymbol{\}}\AgdaSpace{}%
\AgdaSymbol{→}\AgdaSpace{}%
\AgdaFunction{Π}\AgdaSpace{}%
\AgdaBound{Γ}\AgdaSpace{}%
\AgdaOperator{\AgdaFunction{⟦\AgdaUnderscore{}⟧}}\AgdaSpace{}%
\AgdaBound{Δ}\AgdaSpace{}%
\AgdaSymbol{→}\AgdaSpace{}%
\AgdaDatatype{V}\AgdaSpace{}%
\AgdaBound{α}\AgdaSpace{}%
\AgdaBound{Γ}\AgdaSpace{}%
\AgdaSymbol{→}\AgdaSpace{}%
\AgdaOperator{\AgdaFunction{⟦}}\AgdaSpace{}%
\AgdaBound{α}\AgdaSpace{}%
\AgdaOperator{\AgdaFunction{⟧}}\AgdaSpace{}%
\AgdaBound{Δ}\<%
\\
\>[15]\AgdaBound{ε}\AgdaSpace{}%
\AgdaOperator{\AgdaFunction{≪}}\AgdaSpace{}%
\AgdaInductiveConstructor{•}\AgdaSpace{}%
\AgdaOperator{\AgdaFunction{≫}}\AgdaSpace{}%
\AgdaSymbol{=}\AgdaSpace{}%
\AgdaFunction{π₂}\AgdaSpace{}%
\AgdaBound{ε}\<%
\\
\>[15]\AgdaBound{ε}\AgdaSpace{}%
\AgdaOperator{\AgdaFunction{≪}}\AgdaSpace{}%
\AgdaInductiveConstructor{↑}\AgdaSpace{}%
\AgdaBound{x}\AgdaSpace{}%
\AgdaOperator{\AgdaFunction{≫}}\AgdaSpace{}%
\AgdaSymbol{=}\AgdaSpace{}%
\AgdaFunction{π₁}\AgdaSpace{}%
\AgdaBound{ε}\AgdaSpace{}%
\AgdaOperator{\AgdaFunction{≪}}\AgdaSpace{}%
\AgdaBound{x}\AgdaSpace{}%
\AgdaOperator{\AgdaFunction{≫}}\<%
\\
\>[0]\AgdaOperator{\AgdaFunction{⟦}}\AgdaSpace{}%
\AgdaInductiveConstructor{unit}\AgdaSpace{}%
\AgdaOperator{\AgdaFunction{⟧₀}}\AgdaSpace{}%
\AgdaSymbol{\AgdaUnderscore{}}\AgdaSpace{}%
\AgdaSymbol{=}\AgdaSpace{}%
\AgdaSymbol{\AgdaUnderscore{}}\<%
\\
\>[0]\AgdaOperator{\AgdaFunction{⟦}}\AgdaSpace{}%
\AgdaInductiveConstructor{pair}\AgdaSpace{}%
\AgdaBound{t₁}\AgdaSpace{}%
\AgdaBound{t₂}\AgdaSpace{}%
\AgdaOperator{\AgdaFunction{⟧₀}}\AgdaSpace{}%
\AgdaBound{ε}\AgdaSpace{}%
\AgdaSymbol{=}\AgdaSpace{}%
\AgdaSymbol{(}\AgdaSpace{}%
\AgdaOperator{\AgdaFunction{⟦}}\AgdaSpace{}%
\AgdaBound{t₁}\AgdaSpace{}%
\AgdaOperator{\AgdaFunction{⟧₀}}\AgdaSpace{}%
\AgdaBound{ε}\AgdaSpace{}%
\AgdaOperator{\AgdaInductiveConstructor{,}}\AgdaSpace{}%
\AgdaOperator{\AgdaFunction{⟦}}\AgdaSpace{}%
\AgdaBound{t₂}\AgdaSpace{}%
\AgdaOperator{\AgdaFunction{⟧₀}}\AgdaSpace{}%
\AgdaBound{ε}\AgdaSpace{}%
\AgdaSymbol{)}\<%
\\
\>[0]\AgdaOperator{\AgdaFunction{⟦}}\AgdaSpace{}%
\AgdaInductiveConstructor{fst₀}\AgdaSpace{}%
\AgdaBound{t}\AgdaSpace{}%
\AgdaOperator{\AgdaFunction{⟧₀}}\AgdaSpace{}%
\AgdaBound{ε}\AgdaSpace{}%
\AgdaSymbol{=}\AgdaSpace{}%
\AgdaFunction{π₁}\AgdaSpace{}%
\AgdaSymbol{(}\AgdaSpace{}%
\AgdaOperator{\AgdaFunction{⟦}}\AgdaSpace{}%
\AgdaBound{t}\AgdaSpace{}%
\AgdaOperator{\AgdaFunction{⟧₀}}\AgdaSpace{}%
\AgdaBound{ε}\AgdaSpace{}%
\AgdaSymbol{)}\<%
\\
\>[0]\AgdaOperator{\AgdaFunction{⟦}}\AgdaSpace{}%
\AgdaInductiveConstructor{snd₀}\AgdaSpace{}%
\AgdaBound{t}\AgdaSpace{}%
\AgdaOperator{\AgdaFunction{⟧₀}}\AgdaSpace{}%
\AgdaBound{ε}\AgdaSpace{}%
\AgdaSymbol{=}\AgdaSpace{}%
\AgdaFunction{π₂}\AgdaSpace{}%
\AgdaSymbol{(}\AgdaSpace{}%
\AgdaOperator{\AgdaFunction{⟦}}\AgdaSpace{}%
\AgdaBound{t}\AgdaSpace{}%
\AgdaOperator{\AgdaFunction{⟧₀}}\AgdaSpace{}%
\AgdaBound{ε}\AgdaSpace{}%
\AgdaSymbol{)}\<%
\\
\>[0]\AgdaOperator{\AgdaFunction{⟦}}\AgdaSpace{}%
\AgdaInductiveConstructor{abs}\AgdaSpace{}%
\AgdaBound{t}\AgdaSpace{}%
\AgdaOperator{\AgdaFunction{⟧₀}}\AgdaSpace{}%
\AgdaBound{ε}\AgdaSpace{}%
\AgdaBound{f}\AgdaSpace{}%
\AgdaBound{x}%
\>[1138I]\AgdaSymbol{=}\AgdaSpace{}%
\AgdaOperator{\AgdaFunction{⟦}}\AgdaSpace{}%
\AgdaBound{t}\AgdaSpace{}%
\AgdaOperator{\AgdaFunction{⟧₀}}\AgdaSpace{}%
\AgdaSymbol{(}\AgdaSpace{}%
\AgdaBound{ε}\AgdaSpace{}%
\AgdaOperator{\AgdaFunction{[}}\AgdaSpace{}%
\AgdaBound{f}\AgdaSpace{}%
\AgdaOperator{\AgdaFunction{]ₚ}}\AgdaSpace{}%
\AgdaOperator{\AgdaInductiveConstructor{,}}\AgdaSpace{}%
\AgdaBound{x}\AgdaSpace{}%
\AgdaSymbol{)}\<%
\\
\>[.][@{}l@{}]\<[1138I]%
\>[17]\AgdaKeyword{where}\<%
\\
\>[17][@{}l@{\AgdaIndent{0}}]%
\>[19]\AgdaOperator{\AgdaFunction{\AgdaUnderscore{}[\AgdaUnderscore{}]ₚ}}\AgdaSpace{}%
\AgdaSymbol{:}\AgdaSpace{}%
\AgdaSymbol{\{}\AgdaBound{Ξ}\AgdaSpace{}%
\AgdaBound{Δ}\AgdaSpace{}%
\AgdaBound{Γ}\AgdaSpace{}%
\AgdaSymbol{:}\AgdaSpace{}%
\AgdaDatatype{𝔽↓}\AgdaSpace{}%
\AgdaDatatype{𝓣}\AgdaSymbol{\}}\AgdaSpace{}%
\AgdaSymbol{→}\AgdaSpace{}%
\AgdaFunction{Π}\AgdaSpace{}%
\AgdaBound{Ξ}\AgdaSpace{}%
\AgdaOperator{\AgdaFunction{⟦\AgdaUnderscore{}⟧}}\AgdaSpace{}%
\AgdaBound{Δ}\AgdaSpace{}%
\AgdaSymbol{→}\AgdaSpace{}%
\AgdaFunction{𝔽↓₀}\AgdaSpace{}%
\AgdaDatatype{𝓣}\AgdaSymbol{(}\AgdaSpace{}%
\AgdaBound{Δ}\AgdaSpace{}%
\AgdaOperator{\AgdaInductiveConstructor{,}}\AgdaSpace{}%
\AgdaBound{Γ}\AgdaSpace{}%
\AgdaSymbol{)}\AgdaSpace{}%
\AgdaSymbol{→}\AgdaSpace{}%
\AgdaFunction{Π}\AgdaSpace{}%
\AgdaBound{Ξ}\AgdaSpace{}%
\AgdaOperator{\AgdaFunction{⟦\AgdaUnderscore{}⟧}}\AgdaSpace{}%
\AgdaBound{Γ}\<%
\\
\>[19]\AgdaOperator{\AgdaFunction{\AgdaUnderscore{}[\AgdaUnderscore{}]ₚ}}\AgdaSpace{}%
\AgdaSymbol{\{}\AgdaInductiveConstructor{·}\AgdaSymbol{\}}\AgdaSpace{}%
\AgdaSymbol{=}\AgdaSpace{}%
\AgdaSymbol{\AgdaUnderscore{}}\<%
\\
\>[19]\AgdaOperator{\AgdaFunction{\AgdaUnderscore{}[\AgdaUnderscore{}]ₚ}}\AgdaSpace{}%
\AgdaSymbol{\{\AgdaUnderscore{}}\AgdaSpace{}%
\AgdaOperator{\AgdaInductiveConstructor{::}}\AgdaSpace{}%
\AgdaSymbol{\AgdaUnderscore{}\}}\AgdaSpace{}%
\AgdaSymbol{(}\AgdaSpace{}%
\AgdaBound{ε}\AgdaSpace{}%
\AgdaOperator{\AgdaInductiveConstructor{,}}\AgdaSpace{}%
\AgdaBound{x}\AgdaSpace{}%
\AgdaSymbol{)}\AgdaSpace{}%
\AgdaBound{ρ}\AgdaSpace{}%
\AgdaSymbol{=}\AgdaSpace{}%
\AgdaSymbol{(}\AgdaSpace{}%
\AgdaBound{ε}\AgdaSpace{}%
\AgdaOperator{\AgdaFunction{[}}\AgdaSpace{}%
\AgdaBound{ρ}\AgdaSpace{}%
\AgdaOperator{\AgdaFunction{]ₚ}}\AgdaSpace{}%
\AgdaOperator{\AgdaInductiveConstructor{,}}\AgdaSpace{}%
\AgdaBound{x}\AgdaSpace{}%
\AgdaOperator{\AgdaFunction{[}}\AgdaSpace{}%
\AgdaBound{ρ}\AgdaSpace{}%
\AgdaOperator{\AgdaFunction{]}}\AgdaSpace{}%
\AgdaSymbol{)}\<%
\\
\>[0]\AgdaOperator{\AgdaFunction{⟦}}\AgdaSpace{}%
\AgdaInductiveConstructor{app}\AgdaSpace{}%
\AgdaBound{t₁}\AgdaSpace{}%
\AgdaBound{t₂}\AgdaSpace{}%
\AgdaOperator{\AgdaFunction{⟧₀}}\AgdaSpace{}%
\AgdaBound{ε}\AgdaSpace{}%
\AgdaSymbol{=}\AgdaSpace{}%
\AgdaOperator{\AgdaFunction{⟦}}\AgdaSpace{}%
\AgdaBound{t₁}\AgdaSpace{}%
\AgdaOperator{\AgdaFunction{⟧₀}}\AgdaSpace{}%
\AgdaBound{ε}\AgdaSpace{}%
\AgdaSymbol{(λ}\AgdaSpace{}%
\AgdaBound{x}\AgdaSpace{}%
\AgdaSymbol{→}\AgdaSpace{}%
\AgdaBound{x}\AgdaSymbol{)}\AgdaSpace{}%
\AgdaSymbol{(}\AgdaSpace{}%
\AgdaOperator{\AgdaFunction{⟦}}\AgdaSpace{}%
\AgdaBound{t₂}\AgdaSpace{}%
\AgdaOperator{\AgdaFunction{⟧₀}}\AgdaSpace{}%
\AgdaBound{ε}\AgdaSpace{}%
\AgdaSymbol{)}\<%
\end{code}

\medskip
\mysubparagraph{Normalisation by evaluation.}
The unquote and quote functions~(see~(\ref{UnquoteQuoteMaps})) are
implemented:
\begin{code}%
\>[0]\AgdaComment{-- unquote and quote}\<%
\\
\>[0]\AgdaKeyword{mutual}\<%
\\
\>[0][@{}l@{\AgdaIndent{0}}]%
\>[2]\AgdaFunction{u}\AgdaSpace{}%
\AgdaSymbol{:}\AgdaSpace{}%
\AgdaSymbol{\{}\AgdaBound{α}\AgdaSpace{}%
\AgdaSymbol{:}\AgdaSpace{}%
\AgdaDatatype{𝓣}\AgdaSymbol{\}}\AgdaSpace{}%
\AgdaSymbol{\{}\AgdaBound{Γ}\AgdaSpace{}%
\AgdaSymbol{:}\AgdaSpace{}%
\AgdaDatatype{𝔽↓}\AgdaSpace{}%
\AgdaDatatype{𝓣}\AgdaSymbol{\}}\AgdaSpace{}%
\AgdaSymbol{→}\AgdaSpace{}%
\AgdaDatatype{𝓜}\AgdaSpace{}%
\AgdaBound{α}\AgdaSpace{}%
\AgdaBound{Γ}\AgdaSpace{}%
\AgdaSymbol{→}\AgdaSpace{}%
\AgdaOperator{\AgdaFunction{⟦}}\AgdaSpace{}%
\AgdaBound{α}\AgdaSpace{}%
\AgdaOperator{\AgdaFunction{⟧}}\AgdaSpace{}%
\AgdaBound{Γ}\<%
\\
\>[2]\AgdaFunction{u}\AgdaSpace{}%
\AgdaSymbol{\{}\AgdaInductiveConstructor{θ}\AgdaSpace{}%
\AgdaSymbol{\AgdaUnderscore{}\}}\AgdaSpace{}%
\AgdaBound{m}\AgdaSpace{}%
\AgdaSymbol{=}\AgdaSpace{}%
\AgdaBound{m}\<%
\\
\>[2]\AgdaFunction{u}\AgdaSpace{}%
\AgdaSymbol{\{}\AgdaInductiveConstructor{υ}\AgdaSymbol{\}}\AgdaSpace{}%
\AgdaSymbol{\AgdaUnderscore{}}\AgdaSpace{}%
\AgdaSymbol{=}\AgdaSpace{}%
\AgdaSymbol{\AgdaUnderscore{}}\<%
\\
\>[2]\AgdaFunction{u}\AgdaSpace{}%
\AgdaSymbol{\{\AgdaUnderscore{}}\AgdaSpace{}%
\AgdaOperator{\AgdaInductiveConstructor{⋆}}\AgdaSpace{}%
\AgdaSymbol{\AgdaUnderscore{}\}}\AgdaSpace{}%
\AgdaBound{m}\AgdaSpace{}%
\AgdaSymbol{=}\AgdaSpace{}%
\AgdaSymbol{(}\AgdaSpace{}%
\AgdaFunction{u}\AgdaSpace{}%
\AgdaSymbol{(}\AgdaInductiveConstructor{fstₘ}\AgdaSpace{}%
\AgdaBound{m}\AgdaSymbol{)}\AgdaSpace{}%
\AgdaOperator{\AgdaInductiveConstructor{,}}\AgdaSpace{}%
\AgdaFunction{u}\AgdaSpace{}%
\AgdaSymbol{(}\AgdaInductiveConstructor{sndₘ}\AgdaSpace{}%
\AgdaBound{m}\AgdaSymbol{)}\AgdaSpace{}%
\AgdaSymbol{)}\<%
\\
\>[2]\AgdaFunction{u}\AgdaSpace{}%
\AgdaSymbol{\{\AgdaUnderscore{}}\AgdaSpace{}%
\AgdaOperator{\AgdaInductiveConstructor{⇒}}\AgdaSpace{}%
\AgdaSymbol{\AgdaUnderscore{}\}}\AgdaSpace{}%
\AgdaBound{m}\AgdaSpace{}%
\AgdaBound{ρ}\AgdaSpace{}%
\AgdaBound{x}\AgdaSpace{}%
\AgdaSymbol{=}\AgdaSpace{}%
\AgdaFunction{u}\AgdaSpace{}%
\AgdaSymbol{(}\AgdaSpace{}%
\AgdaInductiveConstructor{appₘ}\AgdaSpace{}%
\AgdaSymbol{(}\AgdaBound{m}\AgdaSpace{}%
\AgdaOperator{\AgdaFunction{[}}\AgdaSpace{}%
\AgdaBound{ρ}\AgdaSpace{}%
\AgdaOperator{\AgdaFunction{]ₘ}}\AgdaSymbol{)}\AgdaSpace{}%
\AgdaSymbol{(}\AgdaFunction{q}\AgdaSpace{}%
\AgdaBound{x}\AgdaSymbol{)}\AgdaSpace{}%
\AgdaSymbol{)}\<%
\\
\>[0]\<%
\\
\>[2]\AgdaFunction{q}\AgdaSpace{}%
\AgdaSymbol{:}\AgdaSpace{}%
\AgdaSymbol{\{}\AgdaBound{α}\AgdaSpace{}%
\AgdaSymbol{:}\AgdaSpace{}%
\AgdaDatatype{𝓣}\AgdaSymbol{\}}\AgdaSpace{}%
\AgdaSymbol{\{}\AgdaBound{Γ}\AgdaSpace{}%
\AgdaSymbol{:}\AgdaSpace{}%
\AgdaDatatype{𝔽↓}\AgdaSpace{}%
\AgdaDatatype{𝓣}\AgdaSymbol{\}}\AgdaSpace{}%
\AgdaSymbol{→}\AgdaSpace{}%
\AgdaOperator{\AgdaFunction{⟦}}\AgdaSpace{}%
\AgdaBound{α}\AgdaSpace{}%
\AgdaOperator{\AgdaFunction{⟧}}\AgdaSpace{}%
\AgdaBound{Γ}\AgdaSpace{}%
\AgdaSymbol{→}\AgdaSpace{}%
\AgdaDatatype{𝓝}\AgdaSpace{}%
\AgdaBound{α}\AgdaSpace{}%
\AgdaBound{Γ}\<%
\\
\>[2]\AgdaFunction{q}\AgdaSpace{}%
\AgdaSymbol{\{}\AgdaInductiveConstructor{θ}\AgdaSpace{}%
\AgdaSymbol{\AgdaUnderscore{}\}}\AgdaSpace{}%
\AgdaSymbol{(}\AgdaInductiveConstructor{varₘ}\AgdaSpace{}%
\AgdaBound{x}\AgdaSymbol{)}\AgdaSpace{}%
\AgdaSymbol{=}\AgdaSpace{}%
\AgdaInductiveConstructor{varₙ}\AgdaSpace{}%
\AgdaBound{x}\<%
\\
\>[2]\AgdaFunction{q}\AgdaSpace{}%
\AgdaSymbol{\{}\AgdaInductiveConstructor{θ}\AgdaSpace{}%
\AgdaSymbol{\AgdaUnderscore{}\}}\AgdaSpace{}%
\AgdaSymbol{(}\AgdaInductiveConstructor{fstₘ}\AgdaSpace{}%
\AgdaBound{m}\AgdaSymbol{)}\AgdaSpace{}%
\AgdaSymbol{=}\AgdaSpace{}%
\AgdaInductiveConstructor{fstₙ}\AgdaSpace{}%
\AgdaBound{m}\<%
\\
\>[2]\AgdaFunction{q}\AgdaSpace{}%
\AgdaSymbol{\{}\AgdaInductiveConstructor{θ}\AgdaSpace{}%
\AgdaSymbol{\AgdaUnderscore{}\}}\AgdaSpace{}%
\AgdaSymbol{(}\AgdaInductiveConstructor{sndₘ}\AgdaSpace{}%
\AgdaBound{m}\AgdaSymbol{)}\AgdaSpace{}%
\AgdaSymbol{=}\AgdaSpace{}%
\AgdaInductiveConstructor{sndₙ}\AgdaSpace{}%
\AgdaBound{m}\<%
\\
\>[2]\AgdaFunction{q}\AgdaSpace{}%
\AgdaSymbol{\{}\AgdaInductiveConstructor{θ}\AgdaSpace{}%
\AgdaSymbol{\AgdaUnderscore{}\}}\AgdaSpace{}%
\AgdaSymbol{(}\AgdaInductiveConstructor{appₘ}\AgdaSpace{}%
\AgdaBound{m}\AgdaSpace{}%
\AgdaBound{n}\AgdaSymbol{)}\AgdaSpace{}%
\AgdaSymbol{=}\AgdaSpace{}%
\AgdaInductiveConstructor{appₙ}\AgdaSpace{}%
\AgdaBound{m}\AgdaSpace{}%
\AgdaBound{n}\<%
\\
\>[2]\AgdaFunction{q}\AgdaSpace{}%
\AgdaSymbol{\{}\AgdaInductiveConstructor{υ}\AgdaSymbol{\}}\AgdaSpace{}%
\AgdaSymbol{\AgdaUnderscore{}}\AgdaSpace{}%
\AgdaSymbol{=}\AgdaSpace{}%
\AgdaInductiveConstructor{unitₙ}\<%
\\
\>[2]\AgdaFunction{q}\AgdaSpace{}%
\AgdaSymbol{\{\AgdaUnderscore{}}\AgdaSpace{}%
\AgdaOperator{\AgdaInductiveConstructor{⋆}}\AgdaSpace{}%
\AgdaSymbol{\AgdaUnderscore{}\}}\AgdaSpace{}%
\AgdaSymbol{(}\AgdaBound{x₁}\AgdaSpace{}%
\AgdaOperator{\AgdaInductiveConstructor{,}}\AgdaSpace{}%
\AgdaBound{x₂}\AgdaSymbol{)}\AgdaSpace{}%
\AgdaSymbol{=}\AgdaSpace{}%
\AgdaInductiveConstructor{pairₙ}\AgdaSpace{}%
\AgdaSymbol{(}\AgdaFunction{q}\AgdaSpace{}%
\AgdaBound{x₁}\AgdaSymbol{)}\AgdaSpace{}%
\AgdaSymbol{(}\AgdaFunction{q}\AgdaSpace{}%
\AgdaBound{x₂}\AgdaSymbol{)}\<%
\\
\>[2]\AgdaFunction{q}\AgdaSpace{}%
\AgdaSymbol{\{\AgdaUnderscore{}}\AgdaSpace{}%
\AgdaOperator{\AgdaInductiveConstructor{⇒}}\AgdaSpace{}%
\AgdaSymbol{\AgdaUnderscore{}\}}\AgdaSpace{}%
\AgdaBound{f}\AgdaSpace{}%
\AgdaSymbol{=}\AgdaSpace{}%
\AgdaInductiveConstructor{absₙ}\AgdaSymbol{(}\AgdaSpace{}%
\AgdaFunction{q}\AgdaSymbol{(}\AgdaSpace{}%
\AgdaBound{f}\AgdaSpace{}%
\AgdaInductiveConstructor{↑}\AgdaSpace{}%
\AgdaSymbol{(}\AgdaFunction{u}\AgdaSpace{}%
\AgdaSymbol{(}\AgdaInductiveConstructor{varₘ}\AgdaSpace{}%
\AgdaInductiveConstructor{•}\AgdaSymbol{))}\AgdaSpace{}%
\AgdaSymbol{)}\AgdaSpace{}%
\AgdaSymbol{)}\<%
\end{code}
\begin{remark}
A technical point to note is that the implementation of
\AgdaFunction{q}\AgdaSpace{}\AgdaSymbol{\{\AgdaUnderscore{}}\AgdaSpace{}%
\AgdaOperator{\AgdaInductiveConstructor{⇒}}\AgdaSpace{}%
\AgdaSymbol{\AgdaUnderscore{}\}}\AgdaSpace{}\AgdaBound{f}\AgdaSpace{}
arises from the functorial action of presheaf
exponentiation~(in particular with respect to the evaluation
map~(\ref{RepresentableExponentialEvaluationMap})),
which in this case instatiates to the equivalent expression
\mbox{%
\AgdaInductiveConstructor{absₙ}\AgdaSymbol{(}\AgdaSpace{}%
\AgdaFunction{q}\AgdaSymbol{(}\AgdaSpace{}%
\AgdaBound{f}\AgdaSpace{}\AgdaFunction{[}\AgdaSpace{}%
\AgdaInductiveConstructor{↑}\AgdaSpace{}\AgdaFunction{]}\AgdaSpace{}%
\AgdaSymbol{(λ}\AgdaSpace{}\AgdaBound{x}\AgdaSpace{}%
\AgdaSymbol{→}\AgdaSpace{}\AgdaBound{x}\AgdaSymbol{)}\AgdaSpace{}%
\AgdaSymbol{(}\AgdaFunction{u}\AgdaSpace{}%
\AgdaSymbol{(}\AgdaInductiveConstructor{varₘ}\AgdaSpace{}%
\AgdaInductiveConstructor{•}\AgdaSymbol{))}\AgdaSpace{}%
\AgdaSymbol{)}\AgdaSpace{}\AgdaSymbol{)}%
}.
\end{remark}

Finally, the normalisation function~(see~(\ref{NormalisationFunction})) is:
\begin{code}%
\>[0]\AgdaComment{-- nbe}\<%
\\
\>[0]\AgdaFunction{nf}\AgdaSpace{}%
\AgdaSymbol{:}\AgdaSpace{}%
\AgdaSymbol{\{}\AgdaBound{α}\AgdaSpace{}%
\AgdaSymbol{:}\AgdaSpace{}%
\AgdaDatatype{𝓣}\AgdaSymbol{\}}\AgdaSpace{}%
\AgdaSymbol{\{}\AgdaBound{Γ}\AgdaSpace{}%
\AgdaSymbol{:}\AgdaSpace{}%
\AgdaDatatype{𝔽↓}\AgdaSpace{}%
\AgdaDatatype{𝓣}\AgdaSymbol{\}}\AgdaSpace{}%
\AgdaSymbol{→}\AgdaSpace{}%
\AgdaDatatype{𝓛}\AgdaSpace{}%
\AgdaBound{α}\AgdaSpace{}%
\AgdaBound{Γ}\AgdaSpace{}%
\AgdaSymbol{→}\AgdaSpace{}%
\AgdaDatatype{𝓝}\AgdaSpace{}%
\AgdaBound{α}\AgdaSpace{}%
\AgdaBound{Γ}\<%
\\
\>[0]\AgdaFunction{nf}\AgdaSpace{}%
\AgdaBound{t}%
\>[1369I]\AgdaSymbol{=}\AgdaSpace{}%
\AgdaFunction{q}\AgdaSpace{}%
\AgdaSymbol{(}\AgdaSpace{}%
\AgdaOperator{\AgdaFunction{⟦}}\AgdaSpace{}%
\AgdaBound{t}\AgdaSpace{}%
\AgdaOperator{\AgdaFunction{⟧₀}}\AgdaSpace{}%
\AgdaSymbol{(}\AgdaSpace{}%
\AgdaFunction{Π₀}\AgdaSpace{}%
\AgdaSymbol{(}\AgdaFunction{u}\AgdaSpace{}%
\AgdaOperator{\AgdaFunction{∘}}\AgdaSpace{}%
\AgdaInductiveConstructor{varₘ}\AgdaSymbol{)}\AgdaSpace{}%
\AgdaFunction{xs}\AgdaSpace{}%
\AgdaSymbol{)}\AgdaSpace{}%
\AgdaSymbol{)}\<%
\\
\>[.][@{}l@{}]\<[1369I]%
\>[5]\AgdaKeyword{where}\<%
\\
\>[5][@{}l@{\AgdaIndent{0}}]%
\>[7]\AgdaFunction{Π₀}\AgdaSpace{}%
\AgdaSymbol{:}\AgdaSpace{}%
\AgdaSymbol{(}\AgdaBound{f}\AgdaSpace{}%
\AgdaSymbol{:}\AgdaSpace{}%
\AgdaSymbol{\{}\AgdaBound{α}\AgdaSpace{}%
\AgdaSymbol{:}\AgdaSpace{}%
\AgdaDatatype{𝓣}\AgdaSymbol{\}}\AgdaSpace{}%
\AgdaSymbol{\{}\AgdaBound{Δ}\AgdaSpace{}%
\AgdaSymbol{:}\AgdaSpace{}%
\AgdaDatatype{𝔽↓}\AgdaSpace{}%
\AgdaDatatype{𝓣}\AgdaSymbol{\}}\AgdaSpace{}%
\AgdaSymbol{→}\AgdaSpace{}%
\AgdaDatatype{V}\AgdaSpace{}%
\AgdaBound{α}\AgdaSpace{}%
\AgdaBound{Δ}\AgdaSpace{}%
\AgdaSymbol{→}\AgdaSpace{}%
\AgdaOperator{\AgdaFunction{⟦\AgdaUnderscore{}⟧}}\AgdaSpace{}%
\AgdaBound{α}\AgdaSpace{}%
\AgdaBound{Δ}\AgdaSymbol{)}\AgdaSpace{}%
\AgdaSymbol{\{}\AgdaBound{Γ}\AgdaSpace{}%
\AgdaBound{Δ}\AgdaSpace{}%
\AgdaSymbol{:}\AgdaSpace{}%
\AgdaDatatype{𝔽↓}\AgdaSpace{}%
\AgdaDatatype{𝓣}\AgdaSymbol{\}}\AgdaSpace{}%
\AgdaSymbol{→}\AgdaSpace{}%
\AgdaFunction{Π}\AgdaSpace{}%
\AgdaBound{Γ}\AgdaSpace{}%
\AgdaDatatype{V}\AgdaSpace{}%
\AgdaBound{Δ}\AgdaSpace{}%
\AgdaSymbol{→}\AgdaSpace{}%
\AgdaFunction{Π}\AgdaSpace{}%
\AgdaBound{Γ}\AgdaSpace{}%
\AgdaOperator{\AgdaFunction{⟦\AgdaUnderscore{}⟧}}\AgdaSpace{}%
\AgdaBound{Δ}\<%
\\
\>[7]\AgdaFunction{Π₀}\AgdaSpace{}%
\AgdaSymbol{\AgdaUnderscore{}}\AgdaSpace{}%
\AgdaSymbol{\{}\AgdaInductiveConstructor{·}\AgdaSymbol{\}}\AgdaSpace{}%
\AgdaSymbol{\AgdaUnderscore{}}\AgdaSpace{}%
\AgdaSymbol{=}\AgdaSpace{}%
\AgdaSymbol{\AgdaUnderscore{}}\<%
\\
\>[7]\AgdaFunction{Π₀}\AgdaSpace{}%
\AgdaBound{f}\AgdaSpace{}%
\AgdaSymbol{\{\AgdaUnderscore{}}\AgdaSpace{}%
\AgdaOperator{\AgdaInductiveConstructor{::}}\AgdaSpace{}%
\AgdaSymbol{\AgdaUnderscore{}\}}\AgdaSpace{}%
\AgdaSymbol{(}\AgdaSpace{}%
\AgdaBound{xs}\AgdaSpace{}%
\AgdaOperator{\AgdaInductiveConstructor{,}}\AgdaSpace{}%
\AgdaBound{x}\AgdaSpace{}%
\AgdaSymbol{)}\AgdaSpace{}%
\AgdaSymbol{=}\AgdaSpace{}%
\AgdaSymbol{(}\AgdaSpace{}%
\AgdaFunction{Π₀}\AgdaSpace{}%
\AgdaBound{f}\AgdaSpace{}%
\AgdaBound{xs}\AgdaSpace{}%
\AgdaOperator{\AgdaInductiveConstructor{,}}\AgdaSpace{}%
\AgdaBound{f}\AgdaSpace{}%
\AgdaBound{x}\AgdaSpace{}%
\AgdaSymbol{)}\<%
\\
\>[0]\<%
\\
\>[7]\AgdaFunction{xs}\AgdaSpace{}%
\AgdaSymbol{:}\AgdaSpace{}%
\AgdaSymbol{\{}\AgdaBound{Γ}\AgdaSpace{}%
\AgdaSymbol{:}\AgdaSpace{}%
\AgdaDatatype{𝔽↓}\AgdaSpace{}%
\AgdaDatatype{𝓣}\AgdaSymbol{\}}\AgdaSpace{}%
\AgdaSymbol{→}\AgdaSpace{}%
\AgdaFunction{Π}\AgdaSpace{}%
\AgdaBound{Γ}\AgdaSpace{}%
\AgdaDatatype{V}\AgdaSpace{}%
\AgdaBound{Γ}\<%
\\
\>[7]\AgdaFunction{xs}\AgdaSpace{}%
\AgdaSymbol{\{}\AgdaInductiveConstructor{·}\AgdaSymbol{\}}\AgdaSpace{}%
\AgdaSymbol{=}\AgdaSpace{}%
\AgdaSymbol{\AgdaUnderscore{}}\<%
\\
\>[7]\AgdaFunction{xs}\AgdaSpace{}%
\AgdaSymbol{\{\AgdaUnderscore{}}\AgdaSpace{}%
\AgdaOperator{\AgdaInductiveConstructor{::}}%
\>[1454I]\AgdaSymbol{\AgdaUnderscore{}\}}\AgdaSpace{}%
\AgdaSymbol{=}\AgdaSpace{}%
\AgdaSymbol{(}%
\>[24]\AgdaFunction{xs}\AgdaSpace{}%
\AgdaOperator{\AgdaFunction{[}}\AgdaSpace{}%
\AgdaInductiveConstructor{↑}\AgdaSpace{}%
\AgdaOperator{\AgdaFunction{]ᵥ}}\AgdaSpace{}%
\AgdaOperator{\AgdaInductiveConstructor{,}}\AgdaSpace{}%
\AgdaInductiveConstructor{•}\AgdaSpace{}%
\AgdaSymbol{)}\<%
\\
\>[1454I][@{}l@{\AgdaIndent{0}}]%
\>[17]\AgdaKeyword{where}\<%
\\
\>[17][@{}l@{\AgdaIndent{0}}]%
\>[19]\AgdaOperator{\AgdaFunction{\AgdaUnderscore{}[\AgdaUnderscore{}]ᵥ}}\AgdaSpace{}%
\AgdaSymbol{:}\AgdaSpace{}%
\AgdaSymbol{\{}\AgdaBound{Ξ}\AgdaSpace{}%
\AgdaBound{Δ}\AgdaSpace{}%
\AgdaBound{Γ}\AgdaSpace{}%
\AgdaSymbol{:}\AgdaSpace{}%
\AgdaDatatype{𝔽↓}\AgdaSpace{}%
\AgdaDatatype{𝓣}\AgdaSymbol{\}}\AgdaSpace{}%
\AgdaSymbol{→}\AgdaSpace{}%
\AgdaFunction{Π}\AgdaSpace{}%
\AgdaBound{Ξ}\AgdaSpace{}%
\AgdaDatatype{V}\AgdaSpace{}%
\AgdaBound{Δ}\AgdaSpace{}%
\AgdaSymbol{→}\AgdaSpace{}%
\AgdaFunction{𝔽↓₀}\AgdaSpace{}%
\AgdaDatatype{𝓣}\AgdaSymbol{(}\AgdaSpace{}%
\AgdaBound{Δ}\AgdaSpace{}%
\AgdaOperator{\AgdaInductiveConstructor{,}}\AgdaSpace{}%
\AgdaBound{Γ}\AgdaSpace{}%
\AgdaSymbol{)}\AgdaSpace{}%
\AgdaSymbol{→}\AgdaSpace{}%
\AgdaFunction{Π}\AgdaSpace{}%
\AgdaBound{Ξ}\AgdaSpace{}%
\AgdaDatatype{V}\AgdaSpace{}%
\AgdaBound{Γ}\<%
\\
\>[19]\AgdaOperator{\AgdaFunction{\AgdaUnderscore{}[\AgdaUnderscore{}]ᵥ}}\AgdaSpace{}%
\AgdaSymbol{\{}\AgdaInductiveConstructor{·}\AgdaSymbol{\}}\AgdaSpace{}%
\AgdaSymbol{=}\AgdaSpace{}%
\AgdaSymbol{\AgdaUnderscore{}}\<%
\\
\>[19]\AgdaOperator{\AgdaFunction{\AgdaUnderscore{}[\AgdaUnderscore{}]ᵥ}}\AgdaSpace{}%
\AgdaSymbol{\{\AgdaUnderscore{}}\AgdaSpace{}%
\AgdaOperator{\AgdaInductiveConstructor{::}}\AgdaSpace{}%
\AgdaSymbol{\AgdaUnderscore{}\}}\AgdaSpace{}%
\AgdaSymbol{(}\AgdaSpace{}%
\AgdaBound{xs}\AgdaSpace{}%
\AgdaOperator{\AgdaInductiveConstructor{,}}\AgdaSpace{}%
\AgdaBound{x}\AgdaSpace{}%
\AgdaSymbol{)}\AgdaSpace{}%
\AgdaBound{ρ}\AgdaSpace{}%
\AgdaSymbol{=}\AgdaSpace{}%
\AgdaSymbol{(}\AgdaSpace{}%
\AgdaBound{xs}\AgdaSpace{}%
\AgdaOperator{\AgdaFunction{[}}\AgdaSpace{}%
\AgdaBound{ρ}\AgdaSpace{}%
\AgdaOperator{\AgdaFunction{]ᵥ}}\AgdaSpace{}%
\AgdaOperator{\AgdaInductiveConstructor{,}}\AgdaSpace{}%
\AgdaBound{ρ}\AgdaSpace{}%
\AgdaBound{x}\AgdaSpace{}%
\AgdaSymbol{)}\<%
\end{code}

\begin{code}[hide]%
\>[0]\AgdaComment{-- examples}\<%
\\
\>[0]\AgdaFunction{A}\AgdaSpace{}%
\AgdaSymbol{:}\AgdaSpace{}%
\AgdaDatatype{𝓣}\<%
\\
\>[0]\AgdaFunction{A}\AgdaSpace{}%
\AgdaSymbol{=}\AgdaSpace{}%
\AgdaInductiveConstructor{θ}\AgdaSpace{}%
\AgdaInductiveConstructor{zero}\AgdaSpace{}%
\AgdaOperator{\AgdaInductiveConstructor{⋆}}\AgdaSpace{}%
\AgdaSymbol{(}\AgdaInductiveConstructor{θ}\AgdaSpace{}%
\AgdaInductiveConstructor{zero}\AgdaSpace{}%
\AgdaOperator{\AgdaInductiveConstructor{⋆}}\AgdaSpace{}%
\AgdaInductiveConstructor{υ}\AgdaSpace{}%
\AgdaSymbol{)}\<%
\\
\\[\AgdaEmptyExtraSkip]%
\>[0]\AgdaFunction{s}\AgdaSpace{}%
\AgdaSymbol{:}\AgdaSpace{}%
\AgdaDatatype{𝓛}\AgdaSpace{}%
\AgdaSymbol{(}\AgdaSpace{}%
\AgdaFunction{A}\AgdaSpace{}%
\AgdaOperator{\AgdaInductiveConstructor{⇒}}\AgdaSpace{}%
\AgdaFunction{A}\AgdaSpace{}%
\AgdaSymbol{)}\AgdaSpace{}%
\AgdaInductiveConstructor{·}\<%
\\
\>[0]\AgdaFunction{s}\AgdaSpace{}%
\AgdaSymbol{=}\AgdaSpace{}%
\AgdaInductiveConstructor{abs}\AgdaSpace{}%
\AgdaSymbol{(}\AgdaInductiveConstructor{var}\AgdaSpace{}%
\AgdaInductiveConstructor{•}\AgdaSymbol{)}\<%
\\
\\[\AgdaEmptyExtraSkip]%
\>[0]\AgdaFunction{B}\AgdaSpace{}%
\AgdaSymbol{:}\AgdaSpace{}%
\AgdaDatatype{𝓣}\<%
\\
\>[0]\AgdaFunction{B}\AgdaSpace{}%
\AgdaSymbol{=}\AgdaSpace{}%
\AgdaInductiveConstructor{θ}\AgdaSpace{}%
\AgdaInductiveConstructor{zero}\AgdaSpace{}%
\AgdaOperator{\AgdaInductiveConstructor{⋆}}\AgdaSpace{}%
\AgdaSymbol{(}\AgdaInductiveConstructor{θ}\AgdaSpace{}%
\AgdaInductiveConstructor{zero}\AgdaSpace{}%
\AgdaOperator{\AgdaInductiveConstructor{⇒}}\AgdaSpace{}%
\AgdaInductiveConstructor{θ}\AgdaSpace{}%
\AgdaInductiveConstructor{zero}\AgdaSpace{}%
\AgdaSymbol{)}\<%
\\
\\[\AgdaEmptyExtraSkip]%
\>[0]\AgdaFunction{t1}\AgdaSpace{}%
\AgdaSymbol{:}\AgdaSpace{}%
\AgdaDatatype{𝓛}\AgdaSpace{}%
\AgdaSymbol{(}\AgdaSpace{}%
\AgdaFunction{B}\AgdaSpace{}%
\AgdaOperator{\AgdaInductiveConstructor{⇒}}\AgdaSpace{}%
\AgdaFunction{B}\AgdaSpace{}%
\AgdaSymbol{)}\AgdaSpace{}%
\AgdaInductiveConstructor{·}\<%
\\
\>[0]\AgdaFunction{t1}\AgdaSpace{}%
\AgdaSymbol{=}\AgdaSpace{}%
\AgdaInductiveConstructor{abs}\AgdaSpace{}%
\AgdaSymbol{(}\AgdaInductiveConstructor{var}\AgdaSpace{}%
\AgdaInductiveConstructor{•}\AgdaSymbol{)}\<%
\\
\\[\AgdaEmptyExtraSkip]%
\>[0]\AgdaFunction{t2}\AgdaSpace{}%
\AgdaSymbol{:}\AgdaSpace{}%
\AgdaDatatype{𝓛}\AgdaSpace{}%
\AgdaSymbol{(}\AgdaSpace{}%
\AgdaSymbol{(}\AgdaSpace{}%
\AgdaSymbol{(}\AgdaSpace{}%
\AgdaSymbol{(}\AgdaFunction{B}\AgdaSpace{}%
\AgdaOperator{\AgdaInductiveConstructor{⇒}}\AgdaSpace{}%
\AgdaFunction{B}\AgdaSymbol{)}\AgdaSpace{}%
\AgdaOperator{\AgdaInductiveConstructor{⇒}}\AgdaSpace{}%
\AgdaSymbol{(}\AgdaFunction{B}\AgdaSpace{}%
\AgdaOperator{\AgdaInductiveConstructor{⇒}}\AgdaSpace{}%
\AgdaFunction{B}\AgdaSymbol{)}\AgdaSpace{}%
\AgdaSymbol{)}\AgdaSpace{}%
\AgdaOperator{\AgdaInductiveConstructor{⇒}}\AgdaSpace{}%
\AgdaSymbol{(}\AgdaSpace{}%
\AgdaSymbol{(}\AgdaFunction{B}\AgdaSpace{}%
\AgdaOperator{\AgdaInductiveConstructor{⇒}}\AgdaSpace{}%
\AgdaFunction{B}\AgdaSymbol{)}\AgdaSpace{}%
\AgdaOperator{\AgdaInductiveConstructor{⇒}}\AgdaSpace{}%
\AgdaSymbol{(}\AgdaFunction{B}\AgdaSpace{}%
\AgdaOperator{\AgdaInductiveConstructor{⇒}}\AgdaSpace{}%
\AgdaFunction{B}\AgdaSymbol{)}\AgdaSpace{}%
\AgdaSymbol{)}\AgdaSpace{}%
\AgdaSymbol{)}\AgdaSpace{}%
\AgdaOperator{\AgdaInductiveConstructor{⇒}}\AgdaSpace{}%
\AgdaSymbol{(}\AgdaSpace{}%
\AgdaSymbol{(}\AgdaSpace{}%
\AgdaSymbol{(}\AgdaFunction{B}\AgdaSpace{}%
\AgdaOperator{\AgdaInductiveConstructor{⇒}}\AgdaSpace{}%
\AgdaFunction{B}\AgdaSymbol{)}\AgdaSpace{}%
\AgdaOperator{\AgdaInductiveConstructor{⇒}}\AgdaSpace{}%
\AgdaSymbol{(}\AgdaFunction{B}\AgdaSpace{}%
\AgdaOperator{\AgdaInductiveConstructor{⇒}}\AgdaSpace{}%
\AgdaFunction{B}\AgdaSymbol{)}\AgdaSpace{}%
\AgdaSymbol{)}\AgdaSpace{}%
\AgdaOperator{\AgdaInductiveConstructor{⇒}}\AgdaSpace{}%
\AgdaSymbol{(}\AgdaSpace{}%
\AgdaSymbol{(}\AgdaFunction{B}\AgdaSpace{}%
\AgdaOperator{\AgdaInductiveConstructor{⇒}}\AgdaSpace{}%
\AgdaFunction{B}\AgdaSymbol{)}\AgdaSpace{}%
\AgdaOperator{\AgdaInductiveConstructor{⇒}}\AgdaSpace{}%
\AgdaSymbol{(}\AgdaFunction{B}\AgdaSpace{}%
\AgdaOperator{\AgdaInductiveConstructor{⇒}}\AgdaSpace{}%
\AgdaFunction{B}\AgdaSymbol{)}\AgdaSpace{}%
\AgdaSymbol{)}\AgdaSpace{}%
\AgdaSymbol{)}\AgdaSpace{}%
\AgdaSymbol{)}\AgdaSpace{}%
\AgdaInductiveConstructor{·}\<%
\\
\>[0]\AgdaFunction{t2}\AgdaSpace{}%
\AgdaSymbol{=}\AgdaSpace{}%
\AgdaInductiveConstructor{abs}\AgdaSpace{}%
\AgdaSymbol{(}\AgdaInductiveConstructor{var}\AgdaSpace{}%
\AgdaInductiveConstructor{•}\AgdaSymbol{)}\<%
\end{code}

\section*{\Large Conclusion}

We have given a new categorical view of normalisation by evaluation for typed
lambda calculus, both for extensional 
and intensional 
normalisation problems.

\medskip
Extensional normalisation was obtained from a basic lemma unifying
definability and normalisation.  Our analysis has the important methodological
consequence of providing guidance when looking for normal forms.  Indeed, a
basic lemma based on the definability result of Fiore and
Simpson~\cite{FioreSimpson} via Grothendiek logical relations led to syntactic
counterparts of the normal forms of Altenkirch, Dybjer, Hofmann, and
Scott~\cite{ADHS} and has been applied to establish extensional normalisation
for the typed lambda calculus with empty and sum
types~\cite{FioreDiCosmoBalat}.  Along this line of research, one can study
normalisation for other calculi for which definability results based on
Kripke relations have been obtained ---as classical linear
logic~\cite{Streicher}, for instance. 

\medskip
The approach to normalisation by evaluation presented in the paper is novel,
chiefly, in the following respects.  
\begin{itemize}
  \item
    The refinement from the extensional setting to the intensional one
    leading to the formalisation of normalisation by evaluation via
    categorical glueing.

  \item
    The use of an algebraic framework to structure both the development and
    proofs culminating in the definition of the normalisation function within
    a simply typed metatheory. 

  \item
    The synthesis of a normalisation-by-evaluation program in a
    dependently-typed functional programming language.
\end{itemize}

The obtained abstract normalisation algorithm synthesises various concrete
implementations.  Its specialisation to particular implementations of abstract
syntax directly yields normalisation programs for concrete syntactic 
representations.
In particular, we have provided a normalisation-by-evaluation program for the
type-and-scope safe, intrinsically-typed encoding of typed lambda
terms~\cite{AltenkirchReus,BentonEtAl,AllaisEtAl}.
How the abstract setting is related to representations of binding based on
generating globally unique identifiers, say as in~\cite{Filinski01}, needs to
be investigated.  

The role of categorical glueing in our analysis is reminiscent of
realisability.  It would be interesting to understand whether there are
connections to the modified realisability approach of Berger~\cite{Berger}. 

\paragraph{Acknowledgements.}
The basis for this work, which was motivated by a question of Roberto
Di~Cosmo, was done during a visit to PPS, Universit\'e Paris~7 in July~2001
organised by Paul-Andr\'e Melli\`es and supported by the CNRS.
Discussions with Pierre-Louis Curien and Vincent Danos are gratefully
acknowledged.

\appendix

\section{Homomorphism property of $\Termsem{}:\TermPSh{}\to\SemPSh{}$}
\label{Appendix_l}

$$\xymatrix{
& \ar[dl]_-{\varmap_\atype} \VarPSh\atype
\ar[dr]^-{\mbox{\scriptsize$\semfun\sembracket\simplefstarg$}} & \\
\TermPSh\atype \ar[rr]_-{\Termsem\atype} & & 
\Semcat(\semfun\sembracket\fstarg,\semfun\sembracket\atype)
}$$
$$\xymatrix{
& \ar[dl]_-{\unitmap_\UnitType} \TerminalPSh \ar[dr]|-{\iso} & \\
\TermPSh\UnitType \ar[rr]_-{\Termsem\UnitType} & & 
\Semcat(\semfun\sembracket\fstarg,\semfun\sembracket\UnitType)
}$$
$$\begin{array}{ccc}
\xymatrix{
\TermPSh{\atype\ProductType\atype'} \ar[d]_-{\fstmap^{(\atype')}_\atype}
\ar[r]^-{\Termsem{\atype\ProductType\atype'}} & 
\Semcat(\semfun\sembracket\fstarg,
        \semfun\sembracket{\atype\ProductType\atype'})
\ar[d]^-{\proj_1 \comp \simplefstarg}
\\
\TermPSh{\atype} \ar[r]_-{\Termsem\atype} & 
\Semcat(\semfun\sembracket\fstarg,\semfun\sembracket{\atype})
}
& \qquad &
\xymatrix{
\TermPSh{\atype'\ProductType\atype} \ar[d]_-{\sndmap^{(\atype')}_\atype}
\ar[r]^-{\Termsem{\atype'\ProductType\atype}} & 
\Semcat(\semfun\sembracket\fstarg,
        \semfun\sembracket{\atype'\ProductType\atype})
\ar[d]^-{\proj_2 \comp \simplefstarg}
\\
\TermPSh{\atype} \ar[r]_-{\Termsem\atype} & 
\Semcat(\semfun\sembracket\fstarg,\semfun\sembracket{\atype})
}
\end{array}$$
$$\xymatrix{
\TermPSh\tau \CatProduct \TermPSh{\tau'} 
\ar[d]_-{\pairmap_{\tau\ProductType\tau'}} 
\ar[r]^-{\Termsem\tau\CatProduct\Termsem{\tau'}} & 
\Semcat(\semfun\sembracket\fstarg,\semfun\sembracket\tau)
\CatProduct
\Semcat(\semfun\sembracket\fstarg,\semfun\sembracket{\tau'})
\ar[d]|-{\iso}
\\
\TermPSh{\tau\ProductType\tau'}
\ar[r]_-{\Termsem{\tau\ProductType\tau'}} & 
\Semcat(\semfun\sembracket\fstarg,\semfun\sembracket{\tau\ProductType\tau'})
}$$
$$\xymatrix{
\TermPSh{\tau'\ArrowType\tau} \CatProduct \TermPSh{\tau'}
\ar[dd]_-{\appmap^{(\tau')}_\tau} 
\ar[rr]^-{\Termsem{\tau'\ArrowType\tau}\times\Termsem{\tau'}}
&& 
\Semcat(\semfun\sembracket\fstarg,
        \semfun\sembracket{\tau'\ArrowType\tau})
\times
\Semcat(\semfun\sembracket\fstarg,
        \semfun\sembracket{\tau'})
\ar[d]|-{\iso}
\\
&& 
\Semcat(\semfun\sembracket\fstarg,
        (\semfun\sembracket{\tau'} \CatArrow \semfun\sembracket{\tau})
          \CatProduct
        \semfun\sembracket{\tau'}
        )
\ar[d]^-{\eval \comp \simplefstarg}
\\
\TermPSh\tau 
\ar[rr]_-{\Termsem\tau}
&& \Semcat(\semfun\sembracket\fstarg,\semfun\sembracket\tau)
}$$
$$\xymatrix{
(\TermPSh{\tau'})^{\VarPSh\tau} 
\ar[d]_-{\absmap_{\tau\ArrowType\tau'}}
\ar[rr]^-{(\Termsem{\tau'})^{\VarPSh\tau}}
&& 
(\Semcat(\semfun\sembracket\fstarg,\semfun\sembracket{\tau'}))^{\VarPSh\tau}
\ar[d]|-\iso
\\
\TermPSh{\tau\ArrowType\tau'}
\ar[rr]_-{\Termsem{\tau\ArrowType\tau'}}
&&
\Semcat(\semfun\sembracket\fstarg,\semfun\sembracket{\tau\ArrowType\tau'})
}$$

\section{Homomorphism property of
$(\NEsem{},\NFsem{}):(\NEPSh{},\NFPSh{})\to(\SemPSh{},\SemPSh{})$}
\label{Appendix_(m,n)}

\begin{equation} \label{Proof_Of_mu_Operation_One}
\begin{array}{c}\xymatrix{
& \VarPSh\tau 
\ar[ld]_-{\varmap_\tau}
\ar[rd]^-{\mbox{\scriptsize$\semfun\sembracket\simplefstarg$}}
&
\\
\NEPSh\tau 
\ar[rr]_-{\NEsem\tau} 
& & 
\Semcat(\semfun\sembracket\fstarg,\semfun\sembracket\tau)
}\end{array}
\end{equation}
\begin{equation} \label{Proof_Of_mu_Operation_Two}
\begin{array}{c}\xymatrix{
\NEPSh{\tau\ProductType\tau'} 
\ar[d]_-{\fstmap^{(\tau')}_\tau}
\ar[rr]^-{\NEsem{\tau\ProductType\tau'}}
&&
\Semcat(\semfun\sembracket\fstarg,\semfun\sembracket{\tau\ProductType\tau'})
\ar[d]^-{\proj_1 \comp \simplefstarg}
\\
\NEPSh{\tau} 
\ar[rr]_-{\NEsem\tau} 
&&
\Semcat(\semfun\sembracket\fstarg,\semfun\sembracket\tau)
}\end{array}
\end{equation}
\begin{equation} \label{Proof_Of_mu_Operation_Three}
\begin{array}{c}\xymatrix{
\NEPSh{\tau'\ProductType\tau} 
\ar[d]_-{\sndmap^{(\tau')}_\tau}
\ar[rr]^-{\NEsem{\tau'\ProductType\tau}} 
&&
\Semcat(\semfun\sembracket\fstarg,\semfun\sembracket{\tau'\ProductType\tau})
\ar[d]^-{\proj_2 \comp \fstarg}
\\
\NEPSh{\tau}
\ar[rr]_-{\NEsem\tau} 
&&
\Semcat(\semfun\sembracket\fstarg,\semfun\sembracket\tau)
}\end{array}
\end{equation}
\begin{equation} \label{Proof_Of_mu_Operation_Four} 
\begin{array}{c}\xymatrix{
\NEPSh{\tau'\ArrowType\tau} \CatProduct \NFPSh{\tau'}
\ar[dd]_-{\appmap^{(\tau')}_\tau}
\ar[rrr]^-{\NEsem{\tau'\ArrowType\tau}\CatProduct\NFsem{\tau'}}
&&&
\Semcat(\semfun\sembracket\fstarg,
        \semfun\sembracket{\tau'\ArrowType\tau})
\CatProduct
\Semcat(\semfun\sembracket\fstarg,
        \semfun\sembracket{\tau'})
\ar[d]|-{\iso}
\\
&&&
\Semcat(\semfun\sembracket\fstarg,
        (\semfun\sembracket{\tau'}\CatArrow\semfun\sembracket{\tau})
        \CatProduct
        \semfun\sembracket{\tau'})
\ar[d]^-{\eval \comp \fstarg}
\\
\NEPSh\tau
\ar[rrr]_-{\NEsem\tau}
&&&
\Semcat(\semfun\sembracket\fstarg,\semfun\sembracket\tau)
}\end{array}
\end{equation}
\begin{equation} \label{Proof_Of_nu_Operation_One(i)}
\begin{array}{c}\xymatrix{
& \VarPSh\abasetype 
\ar[dl]_-{\varmap_\abasetype}
\ar[rd]^-{\mbox{\scriptsize$\semfun\sembracket\simplefstarg$}}
& 
\\
\NFPSh\abasetype \ar[rr]_-{\NFsem\abasetype} & & 
\Semcat(\semfun\sembracket\fstarg,\semfun\sembracket\abasetype)
}\end{array}
\end{equation}
\begin{equation} \label{Proof_Of_nu_Operation_One(ii)}
\begin{array}{c}\xymatrix{
\NEPSh{\abasetype\ProductType\tau'} 
\ar[d]_-{\fstmap^{(\tau')}_\abasetype}
\ar[rr]^-{\NEsem{\abasetype\ProductType\tau'}}
&&
\Semcat(\semfun\sembracket\fstarg,\semfun\sembracket{\abasetype\ProductType\tau'})
\ar[d]^-{\proj_1 \comp \simplefstarg}
\\
\NFPSh{\abasetype}
\ar[rr]_-{\NFsem\abasetype} 
&&
\Semcat(\semfun\sembracket\fstarg,\semfun\sembracket\abasetype)
}\end{array}
\end{equation}
\begin{equation} \label{Proof_Of_nu_Operation_One(iii)}
\begin{array}{c}\xymatrix{
\NEPSh{\tau'\ProductType\abasetype} 
\ar[d]_-{\sndmap^{(\tau')}_\abasetype}
\ar[rr]^-{\NEsem{\tau'\ProductType\abasetype}} 
&&
\Semcat(\semfun\sembracket\fstarg,
        \semfun\sembracket{\tau'\ProductType\abasetype})
\ar[d]^-{\proj_2 \comp \fstarg}
\\
\NFPSh{\abasetype}
\ar[rr]_-{\NFsem\abasetype} 
&&
\Semcat(\semfun\sembracket\fstarg,\semfun\sembracket\abasetype)
}\end{array}
\end{equation}
\begin{equation} \label{Proof_Of_nu_Operation_One(iv)}
\begin{array}{c}\xymatrix{
\NEPSh{\tau'\ArrowType\abasetype} \CatProduct \NFPSh{\tau'}
\ar[dd]_-{\appmap^{(\tau')}_\abasetype }
\ar[rrr]^-{\NEsem{\tau'\ArrowType\abasetype}\CatProduct\NFsem{\tau'}}
&&&
\Semcat(\semfun\sembracket\fstarg,
        \semfun\sembracket{\tau'\ArrowType\abasetype})
\CatProduct
\Semcat(\semfun\sembracket\fstarg,
        \semfun\sembracket{\tau'})
\ar[d]|-{\iso}
\\
&&&
\Semcat(\semfun\sembracket\fstarg,
        (\semfun\sembracket{\tau'}\CatArrow\semfun\sembracket{\abasetype})
        \CatProduct
        \semfun\sembracket{\tau'})
\ar[d]^-{\eval \comp \fstarg}
\\
\NFPSh\abasetype
\ar[rrr]_-{\NFsem\abasetype}
&&&
\Semcat(\semfun\sembracket\fstarg,\semfun\sembracket\abasetype)
}\end{array}
\end{equation}
\begin{equation} \label{Proof_Of_nu_Operation_Two}
\begin{array}{c}\xymatrix{
& \TerminalPSh 
\ar[dl]_-{\!\!\!\!\!\unitmap_\UnitType}|-{\iso}
\ar[dr]|-{\iso}
& 
\\
\NFPSh\UnitType 
\ar[rr]_-{\NFsem\UnitType}
&&
\Semcat(\semfun\sembracket\fstarg,\semfun\sembracket\UnitType)
}\end{array}
\end{equation}
\begin{equation} \label{Proof_Of_nu_Operation_Three} 
\begin{array}{c}\xymatrix{
\NFPSh\tau \CatProduct \NFPSh{\tau'} 
\ar[d]_-{\pairmap_{\tau\ProductType\tau'}}|-{\iso}
\ar[rr]^-{\NFsem\tau \CatProduct \NFsem{\tau'}}
&&
\Semcat(\semfun\sembracket\fstarg,\semfun\sembracket\tau)
\CatProduct
\Semcat(\semfun\sembracket\fstarg,\semfun\sembracket{\tau'})
\ar[d]|-{\iso}
\\
\NFPSh{\tau\ProductType\tau'}
\ar[rr]_-{\NFsem{\tau\ProductType\tau'}}
&& 
\Semcat(\semfun\sembracket\fstarg,\semfun\sembracket{\tau\ProductType\tau'})
}\end{array}
\end{equation}
\begin{equation} \label{Proof_Of_nu_Operation_Four} 
\begin{array}{c}\xymatrix{
(\NFPSh{\tau'})^{\VarPSh\tau} 
\ar[d]_-{\absmap_{\tau\ArrowType\tau'}}|-{\iso}
\ar[rr]^-{(\NFsem{\tau'})^{\VarPSh\tau}}
&&
(\Semcat(\semfun\sembracket\fstarg,\semfun\sembracket{\tau'}))^{\VarPSh\tau}
\ar[d]|-{\iso}
\\
\NFPSh{\tau\ArrowType\tau'}
\ar[rr]_-{\NFsem{\tau\ArrowType\tau'}}
&&
\Semcat(\semfun\sembracket\fstarg,\semfun\sembracket{\tau\ArrowType\tau'})
}\end{array}
\end{equation}

\newcommand{\boldpar}[1]{\mbox{\boldmath$($}#1\mbox{\boldmath$)$}}
\newcommand{\mymark}[1]{\mathbf{#1}}
\newcommand{\Hmark}
  {\mbox{\scriptsize\boldmath$($}\mymark{H}\mbox{\scriptsize\boldmath$)$}}
\newcommand{\Imark}
  {\mbox{\scriptsize\boldmath$($}\mymark{I}\mbox{\scriptsize\boldmath$)$}}
\newcommand{\Jmark}
{\mbox{\scriptsize\boldmath$($}\mymark{J}\mbox{\scriptsize\boldmath$)$}}
\newcommand{\Qmark}
{\mbox{\scriptsize\boldmath$($}\mymark{Q}\mbox{\scriptsize\boldmath$)$}}
\newcommand{\Umark}
{\mbox{\scriptsize\boldmath$($}\mymark{U}\mbox{\scriptsize\boldmath$)$}}

\section{Proof of Theorem~\protect\ref{THEOREM}}\label{PROOF}

For
$\SubNEmap\atype: \SubNEPSh\atype \rightmono \NEPSh\atype$ and
$\SubNFmap\atype: \SubNFPSh\atype \rightmono \NFPSh\atype$ 
the equalisers of~(\ref{NE_eqn}) and~(\ref{NF_eqn}) 
respectively, we 
show that  
$( \setof{ \SubNEPSh\atype }_{\atype \in \TypeClosure\BaseTypeSet} 
   ,
   \setof{ \SubNFPSh\atype }_{\atype \in \TypeClosure\BaseTypeSet}
   )$ 
is a sub~$\pair{\SSSNE,\SSSNF}$-algebra of $(\NEPSh{},\NFPSh{})$. 
That is, that we have the following situation
$$\begin{array}{ccc}
\xymatrix{
\raisebox{1.5mm}{$\SSSNE(\SubNEPSh{},\SubNFPSh{})$} \ar@{-->}[r] 
\ar@{>->}[d]_-{\SSSNE(\SubNEmap{},\SubNFmap{})}
& \raisebox{1.5mm}{$\SubNEPSh{}$} \ar@{>->}[d]^-{\SubNEmap{}}
\\
\SSSNE(\NEPSh{},\NFPSh{}) \ar[r]|-\iso 
& \NEPSh{}
}
& \qquad\qquad &
\xymatrix{
\raisebox{1.5mm}{$\SSSNF(\SubNEPSh{},\SubNFPSh{})$} \ar@{-->}[r]
\ar@{>->}[d]_-{\SSSNF(\SubNEmap{},\SubNFmap{})}
& \raisebox{1.5mm}{$\SubNFPSh{}$} \ar@{>->}[d]^-{\SubNFmap{}}
\\
\SSSNF(\NEPSh{},\NFPSh{}) \ar[r]|-\iso
& \NFPSh{}
}
\end{array}$$

\medskip
Below, we will use the following conventions: $\boldpar{\mymark{H}}$
indicates commutativity by the homomorphism property;
$\boldpar{\mymark{I}}$ and $\boldpar{\mymark{J}}$ respectively indicate
commutativity by the definition of $\SubNEmap{}$ and $\SubNFmap{}$; and
$\boldpar{\mymark{Q}}$ and $\boldpar{\mymark{U}}$ respectively indicate
commutativity by the definition of $\quotemap{}$ and $\unquotemap{}$. 
\begin{enumerate}
  \item
    For $\atype \in \TypeClosure\BaseTypeSet$, the map 
    $\VarPSh\atype \longstackrightarrow{\varmap_\atype} \NEPSh\atype$ 
    equalises diagram~(\ref{NE_eqn}), and hence factors through 
    ${\SubNEPSh\atype \stackrightmono{\SubNEmap\atype} \NEPSh\atype}$, because
    the diagram 
\begin{center}$\begin{array}{c}\xymatrix@C=85pt{
\terminalobj 
\ar[rd]|-{\varmap_{\atype,\ctxtemb\atype}(\id_{\ctxtemb\atype})}
\ar@/^5pt/[rrd]^-{\vmap_\atype}
\ar@/_5pt/[ddr]_-{\id_{\mbox{\scriptsize$\gsemfun\sembracket\atype$}}}
\ar@{}[ddr]|(.525){\quad\qquad\Hmark}
\\
& \ar@{}[rd]|-\puncture
\NEPSh\atype\ctxtemb\atype \ar[r]^-\iso
\ar[d]_-{\extNEsem{\atype,\ctxtemb\atype}} &
\Hom{\commayon\ctxtemb\atype,\NEobj\atype}
\ar[d]^-{\arHom{\id,\unquotemap\atype}} 
\\
& \Hom{\gsemfun\sembracket\atype,\gsemfun\sembracket\atype}
\ar[r]_-{\arHom{\unquotemap\atype \vmap_\atype,\id}} &
\Hom{\commayon\ctxtemb\atype,\gsemfun\sembracket\atype}
}\end{array}$
\quad in $\Set$
\end{center}
    commutes.

  \item \label{Third_Case} 
    For $\atype, \atype' \in \TypeClosure\BaseTypeSet$, the map
    $\SubNEPSh{\atype\producttype\atype'}
    \longstackrightmono{\SubNEmap{\atype\producttype\atype'}} 
     \NEPSh{\atype\producttype\atype'}
       \longstackrightarrow{\fstmap_\atype^{(\atype')}}
     \NEPSh{\atype}$
    equalises diagram~(\ref{NE_eqn}), and hence factors through 
    $\SubNEPSh\atype \stackrightmono{\SubNEmap\atype} \NEPSh\atype$, as shown
    by the diagram below.
$$\xymatrix{
\SubNEPSh{\atype\producttype\atype'}
\ar[rrr]^-{\SubNEmap{\atype\producttype\atype'}} 
\ar[dd]_- {\SubNEmap{\atype\producttype\atype'}} 
\ar@{}[rrrdd]|(.5){\!\!\!\!\!\!\!\!\!\!\!\!\!\!\!\Imark}
&&& 
\NEPSh{\atype\producttype\atype'} \ar[r]^{\fstmap_\atype^{(\atype')}}
\ar[d]_-{\extNEsem{\atype\producttype\atype}} 
\ar@{}[rd]|(.5){\Hmark}
& \NEPSh\atype
\ar[d]^-{\extNEsem\atype}
\\
&&& \Hom{\gsemfun\sembracket\simplefstarg,
        \gsemfun\sembracket\atype\times\gsemfun\sembracket{\atype'}}
\ar[r]_-{\arHom{\id,\proj_1}} \ar[d]_-{\arHom{\unquotemap{}\vmap,\id}} & 
\Hom{\gsemfun\sembracket\simplefstarg,\gsemfun\sembracket\atype} 
\ar[dd]^-{\arHom{\unquotemap{}\vmap,\id}}
\\
\NEPSh{\atype\producttype\atype'} \ar[d]_-{\fstmap_\atype^{(\atype')}} 
\ar[r]|-\iso
& \Hom{\commayon(\fstarg),\NEobj{\atype\producttype\atype'}} 
\ar[d]|-{\arHom{\id,(\fstmap_\atype^{(\atype')},\proj_1)}}
\ar[rr]^-{\arHom{\id,\unquotemap{\atype\producttype\atype'}}} 
\ar@{}[drr]|(.5){\qquad\qquad\Umark}
&& 
\Hom{\commayon(\fstarg),
     \gsemfun\sembracket\atype\times\gsemfun\sembracket{\atype'}} 
\ar[rd]|-{\arHom{\id,\proj_1}} &
\\
\NEPSh\atype \ar[r]|-\iso
& \Hom{\commayon(\fstarg), \NEobj\atype} 
\ar[rrr]_-{\arHom{\id,\unquotemap\atype}} && & 
\Hom{\commayon(\fstarg),\gsemfun\sembracket\atype}
}$$
    Analogously, for $\atype, \atype' \in \BaseTypeSet$, the map 
    $\SubNEPSh{\atype'\producttype\atype} 
       \longstackrightmono{\SubNEmap{\atype'\producttype\atype}}
     \NEPSh{\atype'\producttype\atype}
       \longstackrightarrow{\sndmap_\atype^{(\atype')}}
     \NEPSh{\atype}$
    equalises diagram~(\ref{NE_eqn}), and hence also factors through 
    $\SubNEPSh\atype \stackrightmono{\SubNEmap\atype} \NEPSh\atype$. 

  \item \label{Second_Case}
    For $\atype, \atype' \in \TypeClosure\BaseTypeSet$, the map 
    $\SubNEPSh{\atype' \arrowtype \atype} \times \SubNFPSh{\atype'}
       \longlongstackrightmono
         {\SubNEmap{\atype' \arrowtype \atype} \times \SubNFmap{\atype'}}
     \NEPSh{\atype' \arrowtype \atype} \times \NFPSh\atype'
       \longlongstackrightarrow{\appmap_\atype^{(\atype')}}
     \NEPSh\atype$
    equalises diagram~(\ref{NE_eqn}), and hence factors through 
    $\SubNEPSh\atype \stackrightmono{\SubNEmap\atype} \NEPSh\atype$, as shown
    by the diagram below.
{\scriptsize$$
\!\!\!\!\!\!\!\!\!\!\!\!\!\!\!\!\!\!\!\!\!\!\!\!\!\!\!\!\!\!\!\!\!\!\!\!
\xymatrix{
\SubNEPSh{\atype' \arrowtype \atype} \times \SubNFPSh{\atype'}
\ar[rr]
  ^-{\SubNEmap{\atype' \arrowtype \atype} \times \SubNFmap{\atype'}} 
\ar[dr]|-{\id \times (\extNFsem{\atype'} \, \SubNFmap{\atype'})}
\ar[dd]|-{\SubNEmap{\atype'\arrowtype\atype} \times \id}
& & 
\NEPSh{\atype' \arrowtype \atype} \times \NFPSh{\atype'}
\ar[rr]^-{\appmap_\atype^{(\atype')}}
\ar[d]|-{\extNEsem{\atype'\arrowtype\atype} \times \extNFsem{\atype'}}
\ar@{}[rrd]|-{\Hmark}
& 
& 
\NEPSh\atype \ar[d]|-{\extNEsem{\atype}} 
&
\\ 
& 
\SubNEPSh{\atype' \arrowtype \atype} \times
\arHom{\gsemfun\sembracket\simplefstarg,\gsemfun\sembracket{\atype'}}
\ar[r]
_-{\raisebox{-3mm}{\scriptsize$(\extNEsem{\atype'\arrowtype\atype} \SubNEmap{\atype'\arrowtype\atype}) \times \id$}} 
\ar[d]|-{\SubNEmap{\atype'\arrowtype\atype} \times \id} 
& 
\arHom{\gsemfun\sembracket\simplefstarg,\gsemfun\sembracket\atype^{\gsemfun\sembracket{\atype'}}}
\times \arHom{\gsemfun\sembracket\simplefstarg,\gsemfun\sembracket{\atype'}}
\ar[dd]|-{\arHom{\unquotemap{}\vmap,\id} \times \id} \ar[r]|-\iso 
&
\arHom{\gsemfun\sembracket\simplefstarg,
     \gsemfun\sembracket\atype^{\gsemfun\sembracket{\atype'}}
       \times \gsemfun\sembracket{\atype'}}
\ar[r]^-{\arHom{\id,\ev}}
\ar[ddd]|-{\arHom{\unquotemap{}\vmap,\id}}
& 
\arHom{\gsemfun\sembracket\simplefstarg,\gsemfun\sembracket\atype} 
\ar[dddddd]|-{\arHom{\unquotemap{}\vmap,\id}}
\\ 
\NEPSh{\atype' \arrowtype \atype} \times \SubNFPSh{\atype'}
\ar[r]^-{\id \times (\extNFsem{\atype'} \SubNFmap{\atype'})}
\ar[dddd]|-{\id \times \SubNFmap{\atype'}}
& 
\NEPSh{\atype' \arrowtype \atype} \times 
\arHom{\gsemfun\sembracket\simplefstarg,\gsemfun\sembracket{\atype'}}
\ar[d]|-{\iso\times\id}
\ar@{}[r]|(.55){\Imark}
&
& & 
&
\\ 
\ar@{}[dr]|-{\Jmark}
&
\arHom{\commayon(\fstarg),\NEobj{\atype' \arrowtype \atype}} \times
\arHom{\gsemfun\sembracket\simplefstarg,\gsemfun\sembracket{\atype'}}
\ar[r]^-{\arHom{\id,\unquotemap{\atype'\arrowtype\atype}} \times \id}
\ar[d]|-{\id \times \arHom{\unquotemap{}\vmap,\id}}
&  
\arHom{\commayon(\fstarg),
       \gsemfun\sembracket{\atype}^{\gsemfun\sembracket{\atype'}}}
\times \arHom{\gsemfun\sembracket\simplefstarg,\gsemfun\sembracket{\atype'}}
\ar[d]|-{\id \times \arHom{\unquotemap{}\vmap,\id}} 
&
& 
\\ 
&
\arHom{\commayon(\fstarg),\NEobj{\atype' \arrowtype \atype}} 
\times
\arHom{\commayon(\fstarg),\gsemfun\sembracket{\atype'}}
\ar[r]^-{\arHom{\id,\unquotemap{\atype'\arrowtype\atype}}\times \id}
\ar[dr]|-\iso
\ar[d]|-{\id \times \arHom{\id,\quotemap{\atype'}}}
&
\arHom{\commayon(\fstarg),
       \gsemfun\sembracket{\atype}^{\gsemfun\sembracket{\atype'}}}
\times 
\arHom{\commayon(\fstarg),\gsemfun\sembracket{\atype'}}
\ar[r]|-\iso
& 
\arHom{\commayon(\fstarg),
     \gsemfun\sembracket{\atype}^{\gsemfun\sembracket{\atype'}}
       \times \gsemfun\sembracket{\atype'}}
\ar[dddr]|-{\arHom{\id,\ev}}
&
\\ 
& 
\arHom{\commayon(\fstarg),\NEobj{\atype'\arrowtype\atype}} \times
\arHom{\commayon(\fstarg),\NFobj{\atype'}}
\ar[rd]|-\iso
& 
\arHom{\commayon(\fstarg),
       \NEobj{\atype'\arrowtype\atype}\times\gsemfun\sembracket{\atype'}} 
\ar[d]^-{\arHom{\id,\id\times\quotemap{\atype'}}}
\ar[ur]|-{\arHom{\id,\unquotemap{\atype'\arrowtype\atype}\times\id}}
\ar@{}[rrd]|(.45){\Umark}
& 
&
\\ 
\NEPSh{\atype' \arrowtype \atype} \times \NFPSh{\atype'}
\ar[d]|-{\appmap_\atype^{(\atype')}}
\ar[ur]|-{\iso\times\iso}
\ar[rr]|-\iso
& 
& 
\arHom{\commayon(\fstarg),
       \NEobj{\atype'\arrowtype\atype} \times \NFobj{\atype'}} 
\ar[d]|-{\arHom{\id,(\appmap_\atype^{(\atype')},\ev)}} 
& &
\\ 
\NEPSh\atype \ar[rr]|-\iso
& 
& 
\arHom{\commayon(\fstarg),\NEobj\atype} 
\ar[rr]_-{\arHom{\id,\unquotemap\atype}} 
& 
& 
\arHom{\commayon(\fstarg),\gsemfun\sembracket{\atype}}
}$$}

  \item
    For $\abasetype \in \BaseTypeSet$, the map 
    $\VarPSh\abasetype 
       \longstackrightarrow{\varmap_\abasetype} 
     \NFPSh\abasetype$ 
    equalises diagram~(\ref{NF_eqn}) with $\atype = \abasetype$, and hence 
    factors through 
    $\SubNFPSh\abasetype 
       \stackrightmono{\SubNFmap\abasetype} 
     \NFPSh\abasetype$
    because the diagram
\begin{center}
$\begin{array}{c}\xymatrix@C=65pt{
\terminalobj
\ar[dr]|-{\id_{\mbox{\scriptsize$\gsemfun\sembracket\abasetype$}}} 
\ar[r]^-{\varmap_{\abasetype,\ctxtemb\abasetype}(\id_{\ctxtemb\abasetype})} 
\ar@/^40pt/[rrd]<1ex>
  ^-{(\varmap_\abasetype,
      \id_{\mbox{\scriptsize$\semfun\sembracket\abasetype$}})}
\ar@/_43pt/[rrd]<-1ex>
  _-{\quotemap\abasetype \unquotemap\abasetype \vmap_\abasetype}
\ar@{}[rd]|(.4){\;\;\;\;\;\qquad\qquad\Hmark}
& \ar@{}[d]|(.625){\qquad\qquad\quad\puncture} 
\NFPSh\abasetype\ctxtemb\abasetype \ar[rd]|-\iso 
\ar[d]|-{\extNFsem{\abasetype,\ctxtemb\abasetype}}
\\
& \Hom{\gsemfun\sembracket\abasetype,\gsemfun\sembracket\abasetype} 
\ar[r]_-{\arHom{\unquotemap\abasetype
    \vmap_\abasetype,\quotemap\abasetype}_{\ctxtemb\abasetype}} 
& \Hom{\commayon\ctxtemb\abasetype,\NFobj\abasetype}
}\end{array}$
\quad in $\Set$
\end{center}
    commutes. 

  \item
    For $\abasetype \in \BaseTypeSet$ and 
    $\atype' \in \TypeClosure\BaseTypeSet$, the map 
    $\SubNEPSh{\abasetype\producttype\atype'} 
       \longstackrightmono{\SubNEmap{\abasetype\producttype\atype'}}
     \NEPSh{\abasetype\producttype\atype'}
       \longstackrightarrow{\fstmap_\abasetype^{(\atype')}}
     \NFPSh{\abasetype}$
    equalises diagram~(\ref{NF_eqn}) with $\atype = \abasetype$, and hence
    factors through 
    $\SubNFPSh\abasetype 
       \stackrightmono{\SubNFmap\abasetype} 
     \NFPSh\abasetype$, 
    as shown by the diagram below (which depends on the diagram 
    in item~\ref{Third_Case} above with $\atype = \abasetype$).
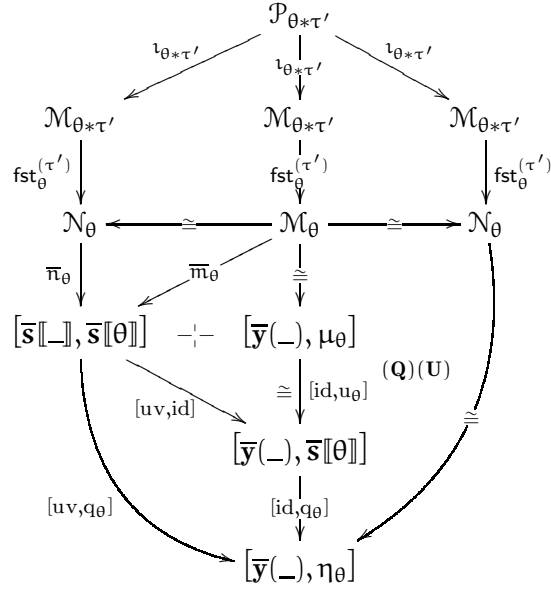
$$\xymatrix{
& 
\ar[dl]_-{\SubNEmap{\abasetype\producttype\atype'}} 
\SubNEPSh{\abasetype\producttype\atype'}
\ar[d]|-{\SubNEmap{\abasetype\producttype\atype'}}
\ar[dr]^-{\SubNEmap{\abasetype\producttype\atype'}} 
& 
\\ 
\NEPSh{\abasetype\producttype\atype'} 
\ar[d]_-{\fstmap_\abasetype^{(\atype')}}
&  
\NEPSh{\abasetype\producttype\atype'} 
\ar[d]|-{\fstmap_\abasetype^{(\atype')}}
&  
\NEPSh{\abasetype\producttype\atype'} 
\ar[d]^-{\fstmap_\abasetype^{(\atype')}}
\\ 
\NFPSh{\abasetype} 
\ar[d]_-{\extNFsem\abasetype}
& 
\ar[l]|-\iso \NEPSh{\abasetype} \ar[r]|-\iso
\ar[d]|-\iso \ar[dl]|-{\extNEsem\abasetype}
& 
\NFPSh{\abasetype} 
\ar@/^30pt/[lddd]|-\iso
\\ 
\Hom{\gsemfun\sembracket\simplefstarg,\gsemfun\sembracket\abasetype}
\ar[dr]_-{\arHom{\unquotemap{}\vmap,\id}}
\ar@{}[r]|-\puncture
\ar@/_30pt/[rdd]_-{\arHom{\unquotemap{}\vmap,\quotemap\abasetype}}
& 
\Hom{\commayon(\fstarg),\NEobj\abasetype}
\ar[d]_-\iso^-{\arHom{\id,\unquotemap\abasetype}}
\ar@{}[r]|(.625){\raisebox{-11.5mm}{\scriptsize$\Qmark\Umark$}}
& 
\\ 
& 
\Hom{\commayon(\fstarg),\gsemfun\sembracket\abasetype}
\ar[d]|-{\arHom{\id,\quotemap\abasetype}}
& 
\\ 
& 
\Hom{\commayon(\fstarg),\NFobj\abasetype}
& 
}$$
    Analogously, for $\abasetype \in \BaseTypeSet$ and 
    $\atype' \in \TypeClosure\BaseTypeSet$, the map 
    $\SubNEPSh{\atype'\producttype\abasetype} 
       \longstackrightmono{\SubNEmap{\atype'\producttype\abasetype}}
     \NEPSh{\atype'\producttype\abasetype}
       \longstackrightarrow{\sndmap_\abasetype^{(\atype')}}
     \NFPSh{\abasetype}$
    equalises diagram~(\ref{NF_eqn}) for $\atype = \abasetype$, and hence also
    factors through  
    $\SubNFPSh\abasetype 
       \stackrightmono{\SubNFmap\abasetype} 
     \NFPSh\abasetype$. 

  \item
    For $\abasetype \in \BaseTypeSet$ and 
    $\atype' \in \TypeClosure\BaseTypeSet$, the map 
    $\SubNEPSh{\atype'\arrowtype\abasetype} 
     \times
     \SubNFPSh{\atype'}
       \longlongstackrightmono
         {\SubNEmap{\atype'\arrowtype\abasetype} \times \SubNFmap{\atype'}}
     \NEPSh{\atype'\arrowtype\abasetype} \times \NFPSh{\atype'}
       \longstackrightarrow{\appmap_\abasetype^{(\atype')}}
     \NFPSh{\abasetype}$
    equalises diagram~(\ref{NF_eqn}) with $\atype = \abasetype$, and hence
    factors through 
    $\SubNFPSh\abasetype 
       \stackrightmono{\SubNFmap\abasetype} 
     \NFPSh\abasetype$, 
    as shown by the diagram below (which depends on the diagram 
    in item~\ref{Second_Case} above with $\atype = \abasetype$).
$$\xymatrix{
& 
\ar[dl]_-{\SubNEmap{\atype'\arrowtype\abasetype} \times \SubNFmap{\atype'}} 
\SubNEPSh{\atype'\arrowtype\abasetype}
\times
\SubNFPSh{\atype'}
\ar[d]|-{\SubNEmap{\atype'\arrowtype\abasetype} \times \SubNFmap{\atype'}} 
\ar[dr]^-{\SubNEmap{\atype'\arrowtype\abasetype} \times \SubNFmap{\atype'}} 
& 
\\ 
\NEPSh{\atype'\arrowtype\abasetype} 
\times
\NFPSh{\atype'}
\ar[d]_-{\appmap_\abasetype^{(\atype')}}
&  
\NEPSh{\atype'\arrowtype\abasetype} 
\times
\NFPSh{\atype'}
\ar[d]|-{\appmap_\abasetype^{(\atype')}}
&  
\NEPSh{\atype'\arrowtype\abasetype} 
\times
\NFPSh{\atype'}
\ar[d]^-{\appmap_\abasetype^{(\atype')}}
\\ 
\NFPSh{\abasetype} 
\ar[d]_-{\extNFsem\abasetype}
& 
\ar[l]|-\iso \NEPSh{\abasetype} \ar[r]|-\iso
\ar[d]|-\iso \ar[dl]|-{\extNEsem\abasetype}
& 
\NFPSh{\abasetype} 
\ar@/^30pt/[lddd]|-\iso
\\ 
\Hom{\gsemfun\sembracket\simplefstarg,\gsemfun\sembracket\abasetype}
\ar[dr]_-{\arHom{\unquotemap{}\vmap,\id}}
\ar@{}[r]|-\puncture
\ar@/_30pt/[rdd]_-{\arHom{\unquotemap{}\vmap,\quotemap\abasetype}}
& 
\Hom{\commayon(\fstarg),\NEobj\abasetype}
\ar[d]_-\iso^-{\arHom{\id,\unquotemap\abasetype}}
\ar@{}[r]|(.55){\raisebox{-11.5mm}{\scriptsize$\Qmark\Umark$}}
& 
\\ 
& 
\Hom{\commayon(\fstarg),\gsemfun\sembracket\abasetype}
\ar[d]|-{\arHom{\id,\quotemap\abasetype}}
& 
\\ 
& 
\Hom{\commayon(\fstarg),\NFobj\abasetype}
& 
}$$

  \item
    Diagram~(\ref{NF_eqn}) with $\atype = \unittype$ commutes, and hence the
    map 
    $\terminalobj 
       \myarrow{\longarrowlength}{\unitmap_\unittype}{\iso}
     \NFPSh\unittype$
    factors through the equaliser
    $\SubNFPSh\unittype 
       \doublestackrightmono{\SubNFmap\unittype}{\iso}
     \NFPSh\unittype$.

  \item
    For $\atype, \atype' \in \TypeClosure\BaseTypeSet$, the map 
    $\SubNFPSh\atype \times \SubNFPSh{\atype'}
       \longstackrightmono{\SubNFmap\atype\times\SubNFmap{\atype'}}
     \NFPSh{\atype} \times \NFPSh{\atype'}
       \myarrow{\longlongarrowlength}
               {\pairmap_{\atype\producttype\atype'}}
               {\iso}
     \NFPSh{\atype\producttype\atype'}$
    equalises diagram~(\ref{NF_eqn}), and hence factors through 
    $\SubNFPSh{\atype\producttype\atype'} 
       \longstackrightmono{\SubNFmap{\atype\producttype\atype'}} 
     \NFPSh{\atype\producttype\atype'}$, 
    as shown by the diagram below.
$$\xymatrix{
\SubNFPSh{\atype} \times \SubNFPSh{\atype'}
\ar[r]^-{\SubNFmap{\atype} \times \SubNFmap{\atype'}}
\ar[ddd]_-{\SubNFmap{\atype} \times \SubNFmap{\atype'}}
\ar@{}[ddr]|(.3){\raisebox{-22mm}{\scriptsize$\Jmark$}}
& 
\NFPSh{\atype} \times \NFPSh{\atype'}
\ar[rr]^-{\pairmap_{\atype\producttype\atype'}}_-\iso
\ar[d]^{\extNFsem{\atype}\times\extNFsem{\atype'}}
\ar@{}[drr]|(.5){\Hmark}
&
&
\NFPSh{\atype\producttype\atype'}
\ar[d]^{\extNFsem{\atype\producttype\atype'}}
\\ 
& 
\Hom{\gsemfun\sembracket\simplefstarg,\gsemfun\sembracket\atype} 
\times 
\Hom{\gsemfun\sembracket\simplefstarg,\gsemfun\sembracket{\atype'}}
\ar[rr]|-\iso 
\ar[d]_-{\arHom{\unquotemap{}\vmap,\quotemap{\atype}}
         \times 
         \arHom{\unquotemap{}\vmap,\quotemap{\atype'}}} 
& 
&
\ar[d]^-{\arHom{\unquotemap{}\vmap,\id}}
\Hom{\gsemfun\sembracket\simplefstarg,
     \gsemfun\sembracket\atype \times \gsemfun\sembracket{\atype'}}
\\ 
& 
\Hom{\commayon(\fstarg),\NFobj{\atype}} 
  \times \Hom{\commayon(\fstarg),\NFobj{\atype'}}
\ar[dr]|-\iso
& 
& 
\ar[dl]_-{\arHom{\id,\quotemap\atype\times\quotemap{\atype'}}}
\ar[dd]^-{\arHom{\id,\quotemap{\atype\producttype\atype'}}} 
\Hom{\commayon(\fstarg),
     \gsemfun\sembracket\atype\times\gsemfun\sembracket{\atype'}}
\\ 
\NFPSh{\atype}\times\NFPSh{\atype'}
\ar[rr]|-\iso \ar[ru]|-{\iso\times\iso} 
\ar[d]_-{\pairmap_{\atype\producttype\atype'}}^-\iso & & 
\Hom{\commayon(\fstarg),\NFobj{\atype}\times\NFobj{\atype'}}
\ar[dr]
  _(.3){\arHom{\id,(\pairmap_{\atype\producttype\atype'},\id)}\qquad}|-\iso 
\ar@{}[r]|(.65){\Qmark}
&
\\ 
\NFPSh{\atype\producttype\atype'}
\ar[rrr]|-\iso & & &
\Hom{\commayon(\fstarg),\NFobj{\atype\producttype\atype'}}
}$$

  \item
    For $\atype, \atype' \in \TypeClosure\BaseTypeSet$, the map
    $(\SubNFPSh{\atype'})^{\VarPSh\atype}
       \longstackrightmono{(\SubNFmap{\atype'})^{\VarPSh\atype}}
     (\NFPSh{\atype'})^{\VarPSh\atype}
       \myarrow{\longlongarrowlength}{\absmap_{\atype\arrowtype\atype'}}{\iso}
     \NFPSh{\atype\arrowtype\atype'}$
    equalises diagram~(\ref{NF_eqn}), and hence factors through
    $\SubNFPSh{\atype\arrowtype\atype'} 
       \longstackrightmono{\SubNFmap{\atype\arrowtype\atype'}} 
     \NFPSh{\atype\arrowtype\atype'}$, 
    as shown by the diagram below.
$$\xymatrix{
(\SubNFPSh{\atype'})^{\VarPSh\atype} 
\ar[r]^-{(\SubNFmap{\atype'})^{\VarPSh\atype}}
\ar[ddd]_-{(\SubNFmap{\atype'})^{\VarPSh\atype}}
\ar@{}[ddr]|(.3){\raisebox{-22mm}{\scriptsize$\Jmark$}}
& 
(\NFPSh{\atype'})^{\VarPSh\atype} 
\ar[rr]^-{\absmap_{\atype\arrowtype\atype'}}_-\iso
\ar[d]^{(\extNFsem{\atype'})^{\VarPSh\atype}}
\ar@{}[drr]|(.5){\Hmark}
&
&
\NFPSh{\atype\arrowtype\atype'}
\ar[d]^{\extNFsem{\atype\arrowtype\atype'}}
\\
& \Hom{\gsemfun\sembracket\fstarg,\gsemfun\sembracket{\atype'}
       }^{\VarPSh\atype} 
\ar[rr]|-\iso 
\ar[d]_-{\arHom{\unquotemap{}\vmap,\quotemap{\atype'}}^{\VarPSh\atype}} & &
\ar[d]^-{\arHom{\unquotemap{}\vmap,\id}}
\Hom{\gsemfun\sembracket\simplefstarg,
     \gsemfun\sembracket{\atype'}^{\gsemfun\sembracket\atype}}
\\
& \Hom{\commayon(\fstarg),\NFobj{\atype'}}^{\VarPSh\atype} \ar[dr]|-\iso
& 
& 
\ar[dl]_-{\arHom{\id,{\quotemap{\atype'}}^{\unquotemap\atype\vmap_\atype}}}
\ar[dd]^-{\arHom{\id,\quotemap{\atype\arrowtype\atype'}}} 
\Hom{\commayon(\fstarg),
     \gsemfun\sembracket{\atype'}^{\gsemfun\sembracket\atype}}
\\
(\NFPSh{\atype'})^{\VarPSh\atype}
\ar[rr]|-\iso \ar[ru]|-{(\iso)^{\VarPSh\atype}} 
\ar[d]_-{\absmap_{\atype\arrowtype\atype'}}^-\iso & & 
\Hom{\commayon(\fstarg),(\NFobj{\atype'})^{\Varobj\atype}}
\ar[dr]_(.3){\arHom{\id,(\absmap_{\atype\arrowtype\atype'},\id)}\qquad}|-\iso 
\ar@{}[r]|(.65){\Qmark}
&
\\
\NFPSh{\atype\arrowtype\atype'}
\ar[rrr]|-\iso & & &
\Hom{\commayon(\fstarg),\NFobj{\atype\arrowtype\atype'}}
}$$
\end{enumerate}

\end{document}